\documentclass[11pt]{article}
\usepackage[utf8]{inputenc}
\usepackage{amsfonts,latexsym,amsthm,amssymb,amsmath,amscd,euscript,mathrsfs,graphicx}
\usepackage{framed}
\usepackage{fullpage}
\usepackage{color}
\usepackage{enumitem}
\usepackage{authblk}
\usepackage{mathtools}
\usepackage[hidelinks]{hyperref}

\usepackage{algorithm}
\usepackage[noend]{algpseudocode}
\usepackage{todonotes}

\newtheorem{theorem}{Theorem}[section]
\newtheorem{proposition}[theorem]{Proposition}
\newtheorem{lemma}[theorem]{Lemma}
\newtheorem{corollary}[theorem]{Corollary}

\theoremstyle{definition}
\newtheorem{definition}[theorem]{Definition}

\theoremstyle{remark}
\newtheorem*{remark}{Remark}

\newcommand{\BE}{\mathbb E}

\newcommand{\BI}{\mathbb I}
\newcommand{\BN}{\mathbb N}
\newcommand{\BP}{\mathbb P}

\newcommand{\BR}{\mathbb R}

\newcommand{\eps}{\varepsilon}
\newcommand{\TV}{d_{\text{TV}}}
\newcommand{\poly}{\text{poly}}

\DeclareMathOperator{\Var}{\text{Var}}

\title{Private High-Dimensional Hypothesis Testing}
\author{Shyam Narayanan\thanks{Massachusetts Institute of Technology. Email: \texttt{shyamsn@mit.edu}. Research supported by the NSF Graduate Fellowship and the NSF TRIPODS Program (award DMS-2022448).}}

\begin{document}

\maketitle

\begin{abstract}%
    We provide improved differentially private algorithms for identity testing of high-dimensional distributions. Specifically, for $d$-dimensional Gaussian distributions with known covariance $\Sigma$, we can test whether the distribution comes from $\mathcal{N}(\mu^*, \Sigma)$ for some fixed $\mu^*$ or from some $\mathcal{N}(\mu, \Sigma)$ with total variation distance at least $\alpha$ from $\mathcal{N}(\mu^*, \Sigma)$ with $(\eps, 0)$-differential privacy, using only
\[\tilde{O}\left(\frac{d^{1/2}}{\alpha^2} + \frac{d^{1/3}}{\alpha^{4/3} \cdot \eps^{2/3}} + \frac{1}{\alpha \cdot \eps}\right)\]
    samples if the algorithm is allowed to be computationally inefficient, and only 
\[\tilde{O}\left(\frac{d^{1/2}}{\alpha^2} + \frac{d^{1/4}}{\alpha \cdot \eps}\right)\]
    samples for a computationally efficient algorithm. We also provide a matching lower bound showing that our computationally inefficient algorithm has optimal sample complexity. We also extend our algorithms to various related problems, including mean testing of Gaussians with bounded but unknown covariance, uniformity testing of product distributions over $\{-1, 1\}^d$, and tolerant testing.
    
    Our results improve over the previous best work of Canonne et al.~\cite{CanonneKMUZ20} for both computationally efficient and inefficient algorithms, and even our computationally efficient algorithm matches the optimal \emph{non-private} sample complexity of $O\left(\frac{\sqrt{d}}{\alpha^2}\right)$ in many standard parameter settings. 
    In addition, our results show that, surprisingly, private identity testing of $d$-dimensional Gaussians can be done with fewer samples than private identity testing of discrete distributions over a domain of size $d$ \cite{AcharyaSZ18}, which refutes a conjectured lower bound of~\cite{CanonneKMUZ20}.
\end{abstract}

\section{Introduction} \label{sec:intro}

Hypothesis testing is one of the oldest and most widely studied problems in statistical inference, and is critical in research for nearly every scientific branch. Simply put, hypothesis testing asks, given a collection of $N$ data points $\textbf{X} = \{X^{(1)}, \dots, X^{(N)}\}$, whether the data points come from a distribution proposed by a null hypothesis $\mathcal{H}_0$ or by an alternative hypothesis $\mathcal{H}_1$.

Classical hypothesis testing often focuses on univariate or discrete distributions.
However, these distributions may be insufficient in many applications.
For instance, if we wish to test a hypothesis on patient data, each patient may have numerous features such as those corresponding to vitals, organ function, metabolic rate, presence or severity of diseases, etc. Consequently, it is crucial to develop hypothesis testing procedures for high-dimensional distributions.

In many practical applications of hypothesis testing, the data may reveal highly sensitive information about an individual. For instance, the data may include whether a patient has a certain disease, or has participated in a crime or embarrassing activity.
For this reason, an increasingly important challenge that has arisen in hypothesis testing to ensure that the test is not only accurate but also preserves the privacy of the individuals that contribute data. 
The notion of privacy we wish to guarantee is commonly called \emph{differential privacy}, which was first formulated by \cite{DworkMNS06}. Differential privacy has emerged as the leading notion of privacy both in theory and practice, and has been utilized to measure privacy by companies such as Apple \cite{Apple}, Google \cite{ErlingssonPK14}, and Microsoft \cite{DingKY17}, as well as the US Census Bureau \cite{USCensus}.
Informally, differential privacy provably ensures that changing one piece of the data does not affect the output of the algorithm significantly. Hence, an adversary cannot reconstruct any single user's data based on the algorithm's output, thereby ensuring that each user's data is secure.

In this paper, we study \emph{differentially private} hypothesis testing for \emph{high-dimensional} distributions.
We primarily study two major classes of distributions: multivariate Gaussians and Boolean product distributions. 
Specifically, we consider the following problems, as well as certain generalizations of them, and improve over the previous state-of-the-art results by \cite{CanonneKMUZ20}.
\begin{enumerate}
    \item How many samples from a multivariate Gaussian distribution $\mathcal{N}(\mu, \Sigma)$, where $\mu \in \BR^d$ is the mean vector and $\Sigma \in \BR^{d \times d}$ is the covariance matrix, are necessary to privately determine whether $\mu = \mu^*$ for some fixed $\mu^* \in \BR^d$ or $\mu$ is ``far'' from $\mu^*$? 
    \item How many samples from a product distribution $\mathcal{D}$ over $\{-1, 1\}^d$ are necessary to privately determine whether $\mathcal{D}$ equals some fixed product distribution $\mathcal{D}^*$ or $\mathcal{D}$ is ``far'' from $\mathcal{D}^*$?
\end{enumerate}

\subsection{Private Hypothesis Testing}

In this subsection, we describe the basics of differential privacy and differentially private hypothesis testing. First, we define the notion of neighboring datasets and differential privacy.

\begin{definition}[Neighboring Datasets]
    Let $\mathcal{X}$ be some domain, and let $\textbf{X} = (X^{(1)}, \dots, X^{(N)}) \in \mathcal{X}^N$ and $\textbf{X}' = (X'^{(1)}, \dots, X'^{(N)}) \in \mathcal{X}^N$ be two datasets of size $N$ from $\mathcal{X}$. Then, we say that $\textbf{X}$ and $\textbf{X}'$ are \emph{neighboring datasets} if there is at most one value of $1 \le i \le N$ such that $X^{(i)} \neq X'^{(i)}$.
\end{definition}

\begin{definition}[Differential Privacy \cite{DworkMNS06}]
    Let $0 \le \eps, \delta \le 1$. A randomized algorithm $\mathcal{A}: \mathcal{X}^N \to \mathcal{O}$ is said to be $(\eps, \delta)$-differentially private (DP) if for any two neighboring datasets $\textbf{X}, \textbf{X}' \in \mathcal{X}^N$ and any subset $O \subset \mathcal{O}$,
\begin{equation}
    \BP[\mathcal{A}(\textbf{X}') \in O] \le e^{\eps} \cdot \BP[\mathcal{A}(\textbf{X}) \in O] + \delta. \label{eq:DP_defn}
\end{equation}
\end{definition}

We next describe hypothesis testing. A hypothesis $\mathcal{H}$ represents a class of distributions over $\mathcal{X}$, which may consist of either a single distribution or a family of distributions with certain constraints.
In distribution testing, we are given two hypotheses $\mathcal{H}_0$ and $\mathcal{H}_1$ that are \emph{disjoint}, meaning no distribution is in both classes. We are then given $N$ i.i.d. samples from $\mathcal{D}$, where $\mathcal{D}$ is a distribution from $\mathcal{H}_0 \cup \mathcal{H}_1$, and our goal is to determine whether $\mathcal{D} \in \mathcal{H}_0$ or $\mathcal{D} \in \mathcal{H}_1$.

We now give a formal definition for \emph{hypothesis testing}, and then we formally define \emph{private hypothesis testing}.

\begin{definition}
    Fix $N \in \BN$ as the number of samples, and let $\mathcal{A}: (\BR^d)^N \to \{0, 1\}$ be an algorithm that takes as input $X^{(1)}, \dots, X^{(N)} \in \BR^d$.
    Given disjoint hypotheses $\mathcal{H}_0$ and $\mathcal{H}_1$, and parameters $0 \le \eps, \delta \le 1,$ we say that an algorithm $\mathcal{A}$ is a \emph{hypothesis testing algorithm} that can \emph{distinguish} between $\mathcal{H}_0$ and $\mathcal{H}_1$
    if:
\begin{itemize}
    \item For all distributions $\mathcal{D} \in \mathcal{H}_0$, if each $X^{(i)}$ is drawn i.i.d. from $\mathcal{D}$, then for $\textbf{X} = (X^{(1)}, \dots, X^{(N)})$, $\BP[\mathcal{A}(\textbf{X}) = 0] \ge \frac{2}{3}$, where the probability is over both the samples $X^{(1)}, \dots, X^{(N)} \leftarrow \mathcal{D}$ and the randomness of the algorithm $\mathcal{A}$.
    \item For all distributions $\mathcal{D}' \in \mathcal{H}_1$, if each $X^{(i)}$ is drawn i.i.d. from $\mathcal{D}'$, then $\BP[\mathcal{A}(\textbf{X}) = 1] \ge \frac{2}{3}$.
\end{itemize}
\end{definition}

\begin{definition}
    Fix $N \in \BN$ as the number of samples, and let $\mathcal{A}: (\BR^d)^N \to \{0, 1\}$ be an algorithm that takes as input $X^{(1)}, \dots, X^{(N)} \in \BR^d$.
    Given disjoint hypotheses $\mathcal{H}_0$ and $\mathcal{H}_1$, and parameters $0 \le \eps, \delta \le 1,$ we say that an algorithm $\mathcal{A}$ can $(\eps, \delta)$-\emph{privately distinguish} between $\mathcal{H}_0$ and $\mathcal{H}_1$ if:
    \begin{itemize}
        \item $\mathcal{A}$ can distinguish between $\mathcal{H}_0$ and $\mathcal{H}_1$.
        \item  $\mathcal{A}$ is $(\eps, \delta)$-DP (where $\mathcal{X} = \BR^d$ and $\mathcal{O} = \{0, 1\}$).
    \end{itemize}
    Note that the privacy must hold for any neighboring datasets $\textbf{X}, \textbf{X}' \in (\BR^d)^N$, even if they are not drawn from any distribution.
\end{definition}

\subsection{Our Results}

In all of the problems we investigate, our goal is to devise an algorithm $\mathcal{A}$ that can $(\eps, \delta)$-privately distinguish between a null hypothesis $\mathcal{H}_0$ and an alternative hypothesis $\mathcal{H}_1,$ where the number of samples $N$ is as small as possible. We also are interested in constructing such an algorithm that is efficient, meaning that the runtime is polynomial in the number of samples $N$ and the dimension $d$.

\paragraph{Identity testing of Gaussians with known covariance:}
The main and perhaps simplest problem we study in this paper is \emph{private identity testing} of a multivariate Gaussian with known covariance matrix $\Sigma$.
In identity testing, the goal is to distinguish between a null hypothesis that only consists of $N$ samples drawn i.i.d. from a single hypothesis distribution $\mathcal{N}(\mu^*, \Sigma)$, and an alternative hypothesis consisting of $N$ samples drawn i.i.d. from $\mathcal{N}(\mu, \Sigma)$, where the mean $\mu$ is ``far'' from $\mu^*$.
In this work, we improve over work by \cite{CanonneKMUZ20} by improving their sample complexity both in the case of inefficient and efficient algorithms, and by providing an optimal lower bound to complement these results.

First, we state our sample complexity upper bound where we allow for an inefficient algorithm.

\begin{theorem}[Inefficient Upper Bound] \label{thm:slow_upper}
    Fix $\mu^* \in \BR^d$ and $\Sigma \in \BR^{d \times d}$ as a known positive definite covariance matrix. Also, fix parameters $0 < \alpha, \eps \le \frac{1}{2}$.
    Then, there exists an algorithm that, using 
$$N = \tilde{O}\left(\frac{d^{1/2}}{\alpha^2} + \frac{d^{1/3}}{\alpha^{4/3} \cdot \eps^{2/3}} + \frac{1}{\alpha \cdot \eps}\right)$$
samples\footnote{We use $N = \tilde{O}(F)$ to mean that there exists a fixed constant $C$ such that $N = O(F (\log F)^C)$.}, can $(\eps, 0)$-privately distinguish between $\mathcal{H}_0$, which solely consists of $\mathcal{N}(\mu^*, \Sigma)$, and $\mathcal{H}_1$, which consists of all distributions $\mathcal{N}(\mu, \Sigma)$ for $\mu$ with $\sqrt{(\mu-\mu^*)^T \Sigma^{-1} (\mu-\mu^*)} \ge \alpha$.
\end{theorem}

Expressed more simply, our goal is to privately distinguish between the mean of a multivariate Gaussian being some fixed $\mu^*$ and the mean being far away from $\mu^*$, given some number of samples from the multivariate Gaussian.
Our notion of far away depends on the quantity $\sqrt{(\mu-\mu^*)^T \Sigma^{-1} (\mu-\mu^*)}$, which is also called the \emph{Mahalanobis distance} $d_\Sigma(\mu^*, \mu)$. While this choice may appear more confusing than simply using $\ell_2$ (a.k.a. Euclidean) distance, it is more practical than $\ell_2$ distance as it scales properly with linear transformations of multivariate Gaussian distributions. When $\Sigma = I$, the identity matrix, this is equivalent to $\|\mu-\mu^*\|$, the Euclidean distance between $\mu$ and $\mu^*$. In addition, whenever the Mahalanobis distance is smaller than $1$, it is asymptotically equivalent to the total variation distance $\TV\left(\mathcal{N}(\mu, \Sigma), \mathcal{N}(\mu^*, \Sigma)\right)$.

Our result improves over the previous best inefficient algorithm of \cite{CanonneKMUZ20}, which had a sample complexity of $N = \tilde{O}\left(\frac{d^{1/2}}{\alpha^2} + \frac{d^{1/2}}{\alpha \sqrt{\eps}} + \frac{d^{1/3}}{\alpha^{4/3} \eps^{2/3}} + \frac{1}{\alpha \eps}\right)$.
Importantly, we remove the dependence on $\frac{d^{1/2}}{\alpha \sqrt{\eps}}$ that was present in \cite{CanonneKMUZ20}, which provides a strict improvement whenever $\eps < \alpha^2$ and $\frac{1}{d} < \alpha^2 \cdot \eps$.
We remark that the inefficient algorithm of \cite{CanonneKMUZ20} has a mild inaccuracy, and an important part of our result involves fixing the previous proof.

\medskip

Next, we show that Theorem \ref{thm:slow_upper} is tight, even if the algorithm is allowed to be $(0, \eps)$-differentially private as opposed to $(\eps, 0)$-differentially private\footnote{We remark that in the case of private \emph{hypothesis testing}, $(\eps, \delta)$-DP and $(\eps+\delta, 0)$-DP are known to be \emph{asymptotically equivalent} for any $\eps, \delta < \frac{1}{2}$ \cite{AcharyaSZ18}. In general, however, it is harder to achieve $(\eps+\delta, 0)$-DP.}. Specifically, we prove the following:

\begin{theorem}[Lower Bound] \label{thm:lower}
    Let all notation be as in Theorem \ref{thm:slow_upper}. Then, any algorithm that can $(0, \eps)$-privately distinguish between $\mathcal{H}_0$ and $\mathcal{H}_1$ must have sample complexity at least
$$N = \Omega\left(\frac{d^{1/2}}{\alpha^2} + \frac{d^{1/3}}{\alpha^{4/3} \cdot \eps^{2/3}} + \frac{1}{\alpha \cdot \eps}\right).$$
\end{theorem}

This improves over the previous lower bound of $\Omega\left(\frac{d^{1/2}}{\alpha^2} + \frac{1}{\alpha \cdot \eps}\right)$ \cite{CanonneKMUZ20}, which combines the non-private lower bound of $\Omega\left(\frac{d^{1/2}}{\alpha^2}\right)$ \cite{CanonneDKS20} and the private lower bound of $\Omega\left(\frac{1}{\alpha \cdot \eps}\right)$ for testing $1$-dimensional distributions \cite{AcharyaSZ18}. We remark that the previous lower bound was technically shown only for testing Boolean product distributions, but it extends to multivariate Gaussians easily. 

\medskip

Because the algorithm we devise for Theorem \ref{thm:slow_upper} has very slow runtime, a natural question is how many samples are necessary if the algorithm must run in polynomial time in $N$ and $d$. Indeed, we show the following result, that only needs a slightly larger number of samples but runs efficiently.

\begin{theorem}[Efficient Upper Bound] \label{thm:fast_upper}
    Let all notation be as in Theorem \ref{thm:slow_upper}.
    Then, there exists an algorithm that, using 
$$N = \tilde{O}\left(\frac{d^{1/2}}{\alpha^2} + \frac{d^{1/4}}{\alpha \cdot \eps}\right)$$
samples, can $(\eps, 0)$-privately distinguish between $\mathcal{H}_0$ and $\mathcal{H}_1$ in time polynomial in $N$ and $d$.
\end{theorem}

This improves over the previous best polynomial-time algorithm of \cite{CanonneKMUZ20}, which required $\tilde{O}\left(\frac{d^{1/2}}{\alpha^2} + \frac{d^{1/2}}{\alpha \cdot \eps}\right)$ samples. This algorithm matches even the optimal \emph{non-private} algorithm as long as $\frac{\alpha}{d^{1/4}} \le \eps$, and strictly improves over the previous best efficient algorithm whenever $\eps < \alpha$.
It also strictly improves even over the previous best \emph{inefficient} private algorithm if $\frac{1}{\sqrt{d}} < \eps < \alpha^2$. In addition, if we consider the dependence on $d$ as the bottleneck and only consider the terms dependent on $\sqrt{d}$, the number of samples needed is only roughly $\sqrt{d} \cdot \frac{1}{\alpha^2},$ matching the optimal \emph{non-private} sample complexity! In contrast, \cite{CanonneKMUZ20} required roughly $\sqrt{d} \cdot \left(\frac{1}{\alpha^2}+\frac{1}{\alpha \cdot \eps}\right)$ for efficient algorithms, and $\sqrt{d} \cdot \left(\frac{1}{\alpha^2}+\frac{1}{\alpha \cdot \sqrt{\eps}}\right)$ for inefficient algorithms.

\paragraph{Generalizations to other distributions:}
Our results above generalize to related hypothesis testing problems, such as hypothesis testing for Gaussians with unknown covariance, hypothesis testing for Boolean Product distributions, and tolerant identity testing. We describe the results informally here, and provide more formal statements in the Appendix.

First, we show that our results on privately testing multivariate Gaussians with known covariance can be extended to Gaussians with unknown but bounded covariance. One caveat is that we are no longer able to distinguish between $\mu = \mu^*$ and $\mu, \mu^*$ being far in Mahalanobis distance, as the Mahalanobis distance depends on the unknown matrix $\Sigma$. Instead, we distinguish between $\mu = \mu^*$ and $\mu, \mu^*$ being far in $\ell_2$ distance, if we are promised that $\Sigma$ has bounded spectral norm.

\begin{theorem}[Bounded but Unknown Covariance, Informal] \label{thm:cov_unknown}
    Let $\mathcal{H}_0$ consist of $\mathcal{N}(\mu^*, \Sigma)$ over all covariance matrices with bounded spectral norm $\|\Sigma\|_2 \le 1$, and $\mathcal{H}_1$ consist of $\mathcal{N}(\mu, \Sigma)$ over all covariance matrices $\|\Sigma\|_2 \le 1$ and $\mu: \|\mu-\mu^*\| \ge \alpha$. Then, to distinguish between $\mathcal{H}_0$ and $\mathcal{H}_1$, the same upper and lower bounds as in Theorems \ref{thm:slow_upper}, \ref{thm:lower}, and \ref{thm:fast_upper} hold.
\end{theorem}

Next, we show that our results on testing multivariate Gaussians also extend to identity testing for ``balanced'' Boolean product distributions over $\{-1, 1\}^d$, i.e., distributions where each coordinate is independent (but not necessarily identically distributed). Namely, we can privately test whether a product distribution is some fixed $\mathcal{P}^*$ or has total variation distance far from $\mathcal{P}^*$, as long as the expectation of $\mathcal{P}^*$ is between $-1/2$ and $1/2$ in each coordinate\footnote{$-1/2$ and $1/2$ can be replaced by any constants bounded away from $-1$ and $1$.}.

\begin{theorem}[Product Distributions, Informal] \label{thm:prod}
    Fix $\mu^* \in [-1/2, 1/2]^d$, and suppose that $\mathcal{H}_0$ consists only of the product distribution $\mathcal{P}^*$ over $\{-1, 1\}^d$ with mean $\mu^*$, and $\mathcal{H}_1$ consists of all product distributions $\mathcal{P}$ over $\{-1, 1\}^d$ such that $\TV(\mathcal{P}, \mathcal{P}^*) \ge \alpha$. Then, to distinguish between $\mathcal{H}_0$ and $\mathcal{H}_1$, the same upper and lower bounds as in Theorems \ref{thm:slow_upper}, \ref{thm:lower}, and \ref{thm:fast_upper} all hold for any $\mu^* \in [-1/2, 1/2]^d$.
\end{theorem}

Theorem \ref{thm:prod} implies bounds for private \emph{uniformity} testing of Boolean product distributions, since the uniform distribution over $\{-1, 1\}^d$ is a product distribution with mean $\textbf{0}$.
Theorem \ref{thm:prod} improves over both the previous best upper bound of $\tilde{O}\left(\frac{d^{1/2}}{\alpha^2} + \frac{d^{1/2}}{\alpha \eps^{1/2}} + \frac{d^{1/3}}{\alpha^{4/3} \eps^{2/3}} + \frac{1}{\alpha \eps}\right)$ for inefficient algorithms and $\tilde{O}\left(\frac{d^{1/2}}{\alpha^2} + \frac{d^{1/2}}{\alpha \eps}\right)$ for efficient algorithms \cite{CanonneKMUZ20}. In addition, it also improves over the best lower bound of $\Omega\left(\frac{d^{1/2}}{\alpha^2} + \frac{1}{\alpha \eps}\right)$ \cite{CanonneKMUZ20}. Note that the previous best upper and lower bounds for privately testing balanced Boolean product distributions and multivariate Gaussian distributions match, as do our bounds.

We remark that identity testing of ``unbalanced'' Boolean product distributions, i.e., where $\mu^*$ is not promised to be in $[-1/2, 1/2]^d$, is not always achievable with the same number of samples. Indeed, \cite{CanonneKMUZ20} showed that if the null hypothesis distribution $\mathcal{P}^*$ is sufficiently unbalanced, there is a sample complexity lower bound of $\Omega\left(\frac{d^{1/2}}{\alpha^2} + \frac{d^{1/2}}{\alpha \sqrt{\eps}} + \frac{d^{1/3}}{\alpha^{4/3} \eps^{2/3}} + \frac{1}{\alpha \eps}\right)$.


\medskip

Finally, we show that our results on private identity testing of Gaussians and Boolean product distributions extend to private \emph{tolerant} identity testing. In this setting, we allow for some slack in the null hypothesis, and have to distinguish between the mean $\mu$ being far from $\mu^*$ versus \emph{close} to $\mu^*$, as opposed to just equaling $\mu^*$. Tolerant testing is useful as it provides meaningful guarantees even if the underlying distribution is very close to, but does not perfectly satisfy, the null distribution.

\begin{theorem}[Tolerant Hypothesis Testing, Informal] \label{thm:tolerant}
    Theorems \ref{thm:slow_upper}, \ref{thm:lower}, and \ref{thm:fast_upper} all hold if we replace the null hypothesis $\mathcal{H}_0$ with all distributions $\mathcal{N}(\mu, \Sigma)$ such that $\sqrt{(\mu-\mu^*)^T \Sigma^{-1} (\mu-\mu^*)} \le \frac{\alpha}{2}$. 
    Likewise, for any $\mu^* \in [-1/2, 1/2]^d$ and for $\mathcal{P}^*$ the product distribution with mean $\mu^*$, Theorem \ref{thm:prod} holds even if $\mathcal{H}_0$ consists of all product distributions $\mathcal{P}$ with $\TV(\mathcal{P}, \mathcal{P}^*) \le \frac{\alpha}{C}$ and $\mathcal{H}_1$ consists of product distributions $\mathcal{P}$ with $\TV(\mathcal{P}, \mathcal{P}^*) \ge \alpha$, for a sufficiently large constant $C$. 
\end{theorem}

One surprising consequence of our algorithms is that the number of samples we require for private identity testing of $d$-dimensional Gaussians and private uniformity testing of $d$-dimensional Boolean product distributions is in fact \emph{smaller than} the number of samples needed for private uniformity testing of a discrete distribution over just $d$ elements. Indeed, private uniformity testing of a discrete distribution requires $\Theta\left(\frac{\sqrt{d}}{\alpha^2} + \frac{\sqrt{d}}{\alpha \sqrt{\eps}} + \frac{d^{1/3}}{\alpha^{4/3} \eps^{2/3}} + \frac{1}{\alpha \eps}\right)$ samples~\cite{AcharyaSZ18}, whereas for $d$-dimensional product distributions and Gaussians, we are able to remove the dependence on $\frac{\sqrt{d}}{\alpha \sqrt{\eps}}$. Hence, we refute a conjecture of \cite{CanonneKMUZ20}, which postulates that private uniformity testing of discrete distributions over $[d]$ and of product distributions over $\{-1, 1\}^d$ have asymptotically equivalent sample complexities.

\subsection{Related Work}

Ignoring privacy constraints, hypothesis testing from a statistical point of view dates back nearly a century, notably to \cite{NeymanP33}. Hypothesis testing has become a popular area of study in theoretical computer science more recently (where it is also called \emph{distribution testing}), starting with \cite{GoldreichR00, BatuFRSW00}, and with a large body of subsequent literature over the past two decades (see, for instance, \cite{CanonneSurvey} for a survey of the distribution testing field). While much of the work has focused on discrete or univariate distributions, there has recently been significant work in the multivariate setting as well. This work on multivariate hypothesis testing has come both from statistical \cite{Hotelling31, SrivastavaD08, ChenQ10, CaiM13, JavanmardM14, RamdasISW16} and computational \cite{AlonAKMRX07, RubinfeldX10, AcharyaDK15, DaskalakisP17, AcharyaBDK18, DaskalakisDK18, GheissariLP18, BezakovaBCSV20, CanonneDKS20, DiakonikolasK21} perspectives.

\emph{Differentially private} hypothesis testing began with work by \cite{VuS09, UhlerSF13}, and has seen significant work during the past decade. Apart from the work by \cite{CanonneKMUZ20}, which this paper primarily improves over, perhaps the work most closely related to ours is that of \cite{CaiDK17, AliakbarpourDR18, AcharyaSZ18}, which study the problem of private identity testing (as well as closeness testing) of discrete distributions. Various other problems in private hypothesis testing have also been studied, including testing simple hypotheses \cite{CummingsKMTZ18, CanonneKMSU19}, selection from a discrete set of multiple hypotheses \cite{BunKSW19}, goodness-of-fit and independence testing \cite{WangLK15, GaboardiLRV16, RogersK17, KakizakiFS17, AliakbarpourDKR19}, ANOVA testing \cite{CampbellBRG18, SwanbergGGRGB19}, and nonparametric hypothesis testing \cite{CouchKSBG19}. Finally, hypothesis testing has also been studied with respect to local differential privacy \cite{DuchiJW13, GaboardiR18, Sheffet18, AcharyaCFT19, GopiKKNWZ20, LamWeilLL22}.

Apart from private \emph{hypothesis testing} of high-dimensional multivariate distributions, there has also been been work on private \emph{learning} of Gaussians and multivariate distributions \cite{KarwaV18, KamathLSU19, KamathSSU19, BiswasDKU20, KamathSU20, CaiWZ21, WangX21}. Notably, however, learning a distribution requires a linear dependence on the dimension $d$, whereas hypothesis testing only requires a square-root dependence on the dimension. 

\subsection{Roadmap}

In Section \ref{sec:overview}, we give a technical outline for Theorems \ref{thm:slow_upper} through \ref{thm:tolerant}. All formal proofs are deferred to the appendix. In Appendix \ref{sec:prelim}, we define notation and prove some preliminary results. In Appendix \ref{sec:concentration}, we prove several important concentration bounds. In Appendix \ref{sec:slow}, we prove Theorem \ref{thm:slow_upper}. In Appendix \ref{sec:lower}, we prove Theorem \ref{thm:lower}. In Appendix \ref{sec:fast}, we prove Theorem \ref{thm:fast_upper}. Finally, in Appendix \ref{sec:generalization}, we prove Theorems \ref{thm:cov_unknown}, \ref{thm:prod}, and \ref{thm:tolerant}.
    
\section{Technical Overview} \label{sec:overview}

In this section, we provide an outline for each of the theorems we prove. We first describe the non-private algorithm for testing the mean of a known-covariance Gaussian. We then outline Theorems \ref{thm:fast_upper}, \ref{thm:slow_upper}, and \ref{thm:lower} (in that order), and then outline how we can generalize these results to prove Theorems \ref{thm:cov_unknown}, \ref{thm:prod}, and \ref{thm:tolerant}. We view Theorems \ref{thm:fast_upper} and \ref{thm:lower} as our most interesting results from a technical perspective.

If the covariance $\Sigma$ is known, we may scale and shift so that WLOG $\mu^* = \textbf{0}$ and $\Sigma = I$ is the identity covariance matrix in $d$ dimensions. So, our goal is to determine whether $\mu = 0$ or $\|\mu\| \ge \alpha$. For simplicity, we will assume that the alternative hypothesis is $\|\mu\| = \alpha$ as opposed to $\|\mu\| \ge \alpha$.

\paragraph{Non-Private Hypothesis Testing:}
The optimal non-private algorithm \cite{SrivastavaD08,CanonneDKS20} is incredibly simple: given $N$ samples $X^{(1)}, \dots, X^{(N)} \in \BR^d,$ it just takes the sum of the samples, $\bar{X} = X^{(1)} + \cdots + X^{(N)}$, and computes the statistic $T = \|\bar{X}\|^2.$ Based on how large $T$ is, the algorithm decides whether $\mu = 0$ or $\|\mu\| \ge \alpha$. It is simple to show that if each $X^{(i)}$ is drawn i.i.d. from $\mathcal{N}(0, I)$, then $\BE[T] = N \cdot d$ and $\Var[T] = O(N^2 \cdot d)$. Conversely, if each $X^{(i)}$ is drawn i.i.d. from $\mathcal{N}(\mu, I)$ where $\|\mu\| = \alpha,$ then $\BE[T] = N \cdot d + \alpha^2 \cdot N^2$ and $\Var[T]= O(N^2 \cdot d + N^3 \cdot \alpha^2)$.

For the statistic $T$ to successfully distinguish between the two hypotheses, by Chebyshev's inequality, the square of difference in the means must significantly exceed the variances. As the difference in means is $\alpha^2 \cdot N^2$ and the variances are $O(N^2 \cdot d + N^3 \cdot \alpha^2)$, it suffices to choose $N$ so that $(\alpha^2 \cdot N^2)^2 \ge \Omega(N^2 \cdot d + N^3 \cdot \alpha^2)$. This is equivalent to $N \ge \Omega(\sqrt{d}/\alpha^2)$, so for $N \ge C \sqrt{d}/\alpha^2$ for a sufficiently large constant $C$, we will be able to distinguish between $\mu = 0$ and $\|\mu\| = \alpha$.

\paragraph{Theorem \ref{thm:fast_upper}:}
We note that our algorithm, while somewhat based on the non-private hypothesis testing result, deviates significantly from \cite{CanonneKMUZ20} and other work on private hypothesis testing. For that reason, we do not describe the previous techniques of \cite{CanonneKMUZ20}.

Before returning to the private setting, we note that $\|\bar{X}\|^2$, where $\bar{X} = X^{(1)}+\cdots+X^{(N)}$, can be rewritten as $\sum_{i = 1}^{N} \sum_{j = 1}^{N} \langle X^{(i)}, X^{(j)} \rangle$. Because of this, if we write $T_{i, j} = \langle X^{(i)}, X^{(j)} \rangle$, the non-private algorithm can be rephrased as outputting $0$ (null hypothesis) if $\sum_{i, j} T_{i, j} < N \cdot d + \frac{\alpha^2 N^2}{2}$ and $1$ (alternative hypothesis) if $\sum_{i, j} T_{i, j} > N \cdot d + \frac{\alpha^2 N^2}{2}.$

This motivates our private algorithm, which will attempt to compute $T = \sum_{i = 1}^{N} \sum_{j=1}^{N} T_{i, j}$ privately. However, we note that when a single data sample $X^{(i)}$ changes, this affects $T_{i, j}$ and $T_{j, i}$ for all $j$. In other words, instead of preserving privacy when a single entry in the matrix $\{T_{i, j}\}$ changes, we need to preserve privacy when an entire row and column in the matrix changes.

We will modify the matrix $\textbf{T}$ consisting of all the entries $T_{i, j}$, by subtracting $d$ from each diagonal entry and then dividing the matrix by $\tilde{O}(\sqrt{d})$ to get a new matrix $\textbf{V}$.
By applying classic concentration bounds, we can show that, assuming the original datapoints $X^{(1)}, \dots, X^{(N)}$ are drawn i.i.d. from some $\mathcal{N}(\mu, I)$, with $\|\mu\| \le 1$, then each entry in $\textbf{V}$ is bounded in the range $[-1, 1]$ and each row/column sum of $\textbf{V}$ is bounded in magnitude by $\sqrt{N}.$ Our goal will roughly be to distinguish between the sum of all the entries in $\textbf{V}$ being either in the range $[-N, N]$ or $[\gamma N^2 - N, \gamma N^2 + N]$, for $\gamma \approx \frac{\alpha^2}{\sqrt{d}}$.
In addition, we wish to perform this \emph{privately}, where we consider two matrices $\textbf{V}, \textbf{V}' \in \BR^{N \times N}$ to be adjacent if they differ only in a single row or a single column.

If we may restrict ourselves to matrices where every row and column sum is at most $\sqrt{N}$ in absolute value, the algorithm is quite simple. Let $\bar{V} := \sum_{i = 1}^{N} \sum_{j = 1}^{N} V_{i, j}$. By our restriction, $\bar{V}$ cannot change by more than $2\sqrt{N}$ if we only change a single row or column, so the statistic $\bar{V} + Lap(\eps^{-1} \cdot 2 \sqrt{N})$ is $(\eps, 0)$-differentially private. This means that as long as $N, \eps^{-1} \sqrt{N} \ll \gamma N^2$, our algorithm will be accurate. This is equivalent to $N \ge \Omega\left(\frac{d^{1/2}}{\alpha^2} + \frac{d^{1/3}}{\alpha^{4/3} \eps^{2/3}}\right)$. The problem, however, is that we want differential privacy for arbitrary adjacent datasets $X^{(1)}, \dots, X^{(N)}$, for which we may not have the $\sqrt{N}$ bound on the row and column sums of the corresponding matrix $V$. Even if each entry is bounded in the range $[-1, 1]$, we can increase $\bar{V}$ by $N$ in the worst case.

A common approach to fixing this is to ``clip'' each data point, i.e., replace a data point $x$ with $\max(x_{min}, \min(x_{max}, x))$ to keep $x$ in the range $[x_{min}, x_{max}]$. This technique has been used in statistical data analysis since the early 20th century.
In our case, a natural first attempt to clip the row and column sums. Indeed, we can rewrite $\bar{V}= \frac{1}{2}\left[\sum_{i = 1}^{N} \sum_{j = 1}^{N} V_{i, j} + \sum_{j = 1}^{N} \sum_{i = 1}^{N} V_{i, j}\right],$ i.e., $\bar{V}$ is simply the average of the sum over all row sums and the sum over all column sums. Since we want each row sum and column sum to be bounded by the range $[-\sqrt{N}, \sqrt{N}],$ we can consider replacing our statistic $\bar{V}$ with
\begin{equation} \label{eq:Winsorized_statistic}
    G(\textbf{V}) := \frac{1}{2}\left[\sum_{i = 1}^{N} g\left(\sum_{j = 1}^{N} V_{i, j}\right) + \sum_{j = 1}^{N} g\left(\sum_{i = 1}^{N} V_{i, j}\right)\right],
\end{equation}
where $g(x) := \min(\sqrt{N}, \max(-\sqrt{N}, x))$ prevents each row/column sum from exceeding $\sqrt{N}$ in absolute value. However, this still runs into the same problem as before. If we alter a row of $V$ by increasing each entry in the row by $1$, while the modified row sum does not increase by more than $\sqrt{N}$ now, we still have that each column sum could potentially increase by $1$, causing an overall increase by $N$.

To fix this, we consider an even stronger version of clipping, where after $x$ exceeds $\sqrt{N}$, the function starts going back down again. Specifically, we instead consider the function
\[g(x) := \begin{cases} x & |x| \le \sqrt{N} \\ 2\sqrt{N}-x & x \ge \sqrt{N} \\ -2\sqrt{N}-x & x \le -\sqrt{N} \end{cases}.\]
Note that while $g(x) = x$ in the range $[-\sqrt{N}, \sqrt{N}]$, for general $x$ we have that $g(x) \in [-x - O(\sqrt{N}), -x + O(\sqrt{N})]$.
Replacing $g$ in our Equation \eqref{eq:Winsorized_statistic} with the new function $g$, we now see what happens when we change a single row. If we increase every element in a row by $1$, each column sum increases by $1$, from which we would ideally hope that that $g$ applied to each column sum increases by $1$. Conversely, for the row that we update, we use the fact that $g(x) \in [-x - O(\sqrt{N}), -x + O(\sqrt{N})]$ to say that in fact $g$ applied to the row \emph{decreases} by roughly $N-O(\sqrt{N})$. So, the overall change in the statistic $G(\textbf{V})$ is ideally $O(\sqrt{N})$, because the increase of each column sum by $1$ cancels out with the decrease in $g$ applied to the row sum.

The problem with this, however, is that if a column sum exceeds $\sqrt{N}$ in absolute value, $g$ applied to that column sum goes down instead. This is not an issue even if up to $\sqrt{N}$ column sums exceed $\sqrt{N}$ in absolute value, as even in this case, we have that $\sqrt{N}$ column sums are decreasing by $1$ instead of increasing by $1$, so overall the statistic $G(\textbf{V})$ still does not change by more than $O(\sqrt{N})$. To fix this, we propose another private algorithm which detects and throws out matrices if too many row and column sums exceed $\sqrt{N}$ in absolute value. We remark that if $V$ comes from $X^{(1)}, \dots, X^{(N)}$ drawn from either the null or alternative hypothesis, with high probability no row or column sum will exceed $\sqrt{N}$ in absolute value, so we do not sacrifice accuracy with this algorithm.

We will consider a new threshold function $f(x) = \max(0, \min(\frac{|x|}{\sqrt{N}}-1, 1))$: this function takes $\frac{|x|}{\sqrt{N}}-1$ and clips it to keep it in the range $[0, 1]$. Suppose we apply $f$ to each row and column sum of $V$, i.e., we consider the statistic
\[F(\textbf{V}) := \frac{1}{2}\left[\sum_{i = 1}^{N} f\left(\sum_{j = 1}^{N} V_{i, j}\right) + \sum_{j = 1}^{N} f\left(\sum_{i = 1}^{N} V_{i, j}\right)\right].\]
If $X^{(1)}, \dots, X^{(N)}$ actually came from the distribution, then $f$ applied to each row and column is $0$ with very high probability, as no row or column sum exceeds $\sqrt{N}$ in absolute value. In addition, because $f$ is capped by $0$ and $1$, $f$ applied to a row sum doesn't change by more than $1$ if we change the entire row. In addition, each column sum does not change by more than $1$, so $f$ does not change by more than $1/\sqrt{N}$ for each column sum. So, $F(\textbf{V})$ does not change by more than $\sqrt{N}$ for adjacent datasets. In addition, if more than $2\sqrt{N}$ row/column sums of $V$ exceed $2\sqrt{N}$ in absolute value, then $F(\textbf{V}) \ge \sqrt{N}$.
It will be quite simple to utilize a Laplace Mechanism to privately reject any datasets with $F(\textbf{V})$ exceeding $2\eps^{-1} \sqrt{N}$, which is an $\eps^{-1}$ factor greater than desired. However, with this weaker bound we can ensure $G(\textbf{V})$ does not change by more than $\eps^{-1} \sqrt{N}$ if a single row/column changes, which can be used to obtain a $\tilde{O}\left(\frac{d^{1/2}}{\alpha^2} + \frac{d^{1/3}}{\alpha^{4/3} \eps^{4/3}}\right)$-sample upper bound. This is already an improvement in many regimes.

To improve upon this, we create a series of logarithmic threshold functions $f_k$ for $1 \le k \le O(\log N)$, which have increasing thresholds. This will allow us to create statistics $F_k(\textbf{V})$ for each $k$, similar to $F(\textbf{V})$.
Roughly, we will show that in order for $F_{k+1}$ to change significantly if we change a single row/column, we require $F_k$ to be large. From here, we can show that the relative change in $F_{k+1}$ is much smaller than the relative change in $F_k$, unless $F_k$ is sufficiently large that we could use a Laplace Mechanism to reject such a dataset. We can use these to privately reject any $\textbf{X}$ with $F_{O(\log N)}(\textbf{V}) \gg \tilde{O}(\sqrt{N})$, which will allow for a better sample complexity bound for $N$.

\paragraph{Theorem \ref{thm:slow_upper}:} The proof of Theorem \ref{thm:slow_upper} is based on the corresponding result in \cite{CanonneKMUZ20}: while their result is not fully accurate, we show how to simultaneously fix their result and improve upon it.

We first sketch the ideas behind the computationally inefficient algorithm of \cite{CanonneKMUZ20}. The objective in \cite{CanonneKMUZ20} is to create a map $\hat{T}$ that sends any dataset $\textbf{X} \in (\BR^d)^N$ to $\BR$ with two properties. The first property is that for any two adjacent datasets $\textbf{X}, \textbf{X}'$, $\hat{T}(\textbf{X})$ and $\hat{T}(\textbf{X}')$ are ``close'' in value. The second is that $\hat{T}(\textbf{X})$ should almost always be ``small'' if $\textbf{X}$ is a sample of $N$ i.i.d. $\mathcal{N}(0, I)$ values, and $\hat{T}(\textbf{X})$ should almost always be ``large'' if $\textbf{X}$ is a sample of $N$ i.i.d. $\mathcal{N}(\mu, I)$ values, for any $\mu$ with $\|\mu\| = \alpha$. (``Close'', ``small'', and ''large'' can be effectively quantified.) By adding Laplace noise to $\hat{T}(\textbf{X})$ and determine if the output exceeds a certain threshold, one can privately distinguish between $\mu = 0$ and $\|\mu\| = \alpha$.

In the case where the data points $X^{(1)}, \dots, X^{(N)} \overset{i.i.d.}{\sim} \mathcal{N}(\mu, I)$, we have strong concentration of the row and column sums of the corresponding matrix $\textbf{T} \in \BR^{N \times N}$, due to the independence of the data points. (Recall that $\textbf{T}$ is the matrix with $T_{i,j} = \langle X^{(i)}, X^{(j)} \rangle$.) If we restrict ourselves exclusively to such datasets where the row and column sums of $\textbf{T}$ are properly bounded (call this set $\mathcal{C}$), we can successfully obtain that any two adjacent datasets $\textbf{X}, \textbf{X}'$ have relatively close values of $T(\textbf{X}) := \|\sum X^{(i)}\|^2$ and $T(\textbf{X}') = \|\sum X'^{(i)}\|^2$. Hence, $\hat{T}(\textbf{X}) = T(\textbf{X})$ would actually be a suitable choice if we could restrict ourselves to $\mathcal{C}$.

\cite{CanonneKMUZ20} combines this observation with a theorem about Lipschitz extensions \cite{McShane}. The theorem by \cite{McShane} states that if there exists a function $T: \mathcal{C} \subset \mathcal{X}$ where $\mathcal{X}$ is equipped with some metric (in our case $\mathcal{X} = (\BR^d)^N$ and the metric measures the number of different data points), and $T$ is $D$-Lipschitz for some $D$, then there exists an extension $\hat{T}: \mathcal{X} \to \BR$ that is also $D$-Lipschitz. Expressed more simply in our setting, if we can ensure that $|T(\textbf{X})-T(\textbf{X}')| \le K \cdot D$ for $\textbf{X}, \textbf{X}' \in \mathcal{C}$ that differ in exactly $K$ data points, then we can extend the function $T$ to some $\hat{T}$ which ensures that $\hat{T}$ does not change by more than $K$ on adjacent datasets. The smaller we can make $D$ as a function of $N$, the smaller our sample complexity needs to be.

The main issue in \cite{CanonneKMUZ20} is that they only prove that $|T(\textbf{X})- T(\textbf{X}')| \le D$ for adjacent datasets $\textbf{X}, \textbf{X}' \in \mathcal{C}$. While this ostensibly ensures that $|T(\textbf{X})- T(\textbf{X}')| \le K \cdot D$ for datasets $\textbf{X}, \textbf{X}' \in \mathcal{C}$ that differ in at most $K$ data points, this does not actually hold. For instance, if we change $K$ of the data points in $\textbf{X}$ to make $\textbf{X}'$, the intermediate datasets (obtained by changing the data points one at a time) may not be in $\mathcal{C}$. As a result, we in fact must prove that for all integers $K$, $|T(\textbf{X})-T(\textbf{X}')| \le K \cdot D$ for $\textbf{X}, \textbf{X}' \in \mathcal{C}$ that differ in exactly $K$ data points. To do this, we further restrict the class $\mathcal{C} \subset (\BR^d)^N$, by showing a tight concentration of the norm of $\sum_{i \in S} X^{(i)}$ for all subsets $S \subset [N]$ of size $K$ simultaneously, assuming each $X^{(i)}$ was drawn i.i.d. from some $\mathcal{N}(\mu, I)$, and restricting $\mathcal{C}$ to datasets that satisfy these tight concentration bounds. We then prove that for the restricted set $\mathcal{C}$, we obtain our desired Lipschitz property. Our concentration analysis is tighter than that of \cite{CanonneKMUZ20}, which provides a smaller Lipschitz parameter $D$, and therefore we obtain a reduced value of $N$ as well. Hence, we are able to fix their inaccuracy as well as improve upon their result.

\paragraph{Theorem \ref{thm:lower}:} Our starting point for the lower bound is a theorem of \cite{AcharyaSZ18}, which relates a \emph{coupling} of two distributions $\mathcal{U}$ and $\mathcal{V}$ over $\mathcal{X}^N$ with privacy lower bounds. A coupling of $\mathcal{U}$ and $\mathcal{V}$ is a joint distribution over $(\textbf{X}, \textbf{X}') \sim \mathcal{X}^N \times \mathcal{X}^N$ where the marginal of $\textbf{X}$ is $\mathcal{U}$ and the marginal of $\textbf{X}'$ is $\mathcal{V}$. Specifically, they prove that if there exists a coupling over $\textbf{X} = (X^{(1)}, \dots, X^{(N)})$ and $\textbf{X}' = (X'^{(1)}, \dots, X'^{(N)})$ where the expected number of $i \le N$ such that $X^{(i)} \neq X'^{(i)}$ is at most $O(1/\eps)$, then it is impossible to $(0, \eps)$-privately distinguish between $\mathcal{U}$ and $\mathcal{V}$. This method has been used to provide privacy lower bounds in the discrete distribution setting \cite{AcharyaSZ18}.

As noted by \cite{CanonneKMUZ20}, proving lower bounds for multivariate Gaussians is much more challenging than for discrete distributions, as the coupled distributions $\mathcal{U}, \mathcal{V}$ must be generated as Gaussians with identity covariance, which will usually need strong independence guarantees in each coordinate. In contrast, proving similar lower bounds for distributions over a discrete domain $\{1, 2, \dots, d\}$ do not require us to prove any independence guarantees.

Hence, to apply this result in our setting, two things are necessary. First, we need to decide the distributions $\mathcal{U}$ and $\mathcal{V}$. Next, we need to establish a suitable coupling.
The choice for $\mathcal{U}$ is simple: it will just be the distribution over $(\BR^d)^N$ where each sample is i.i.d. $\mathcal{N}(0, I)$. For $\mathcal{V}$, we wish to find some distribution of mean vectors $\mu$ with $\|\mu\|\ge \alpha$, and then sample $N$ points from $\mathcal{N}(\mu, I)$. The distribution for the mean vectors we choose will roughly be $\mu \sim \mathcal{N}(0, \frac{\alpha^2}{d} \cdot I)$. While this does not ensure that $\|\mu\| \ge \alpha$, a simple concentration inequality ensures that $\|\mu\| \ge \Omega(\alpha)$ with overwhelming probability, which will end up being sufficient.

Next, how do we establish a coupling between $\mathcal{U}$ and $\mathcal{V}$? The first trick we use is to rewrite the distributions $\mathcal{U}, \mathcal{V}$ based on the mean vector $\bar{X} = \frac{X^{(1)}+\cdots+X^{(N)}}{N}$ of the $N$ points in $\mathcal{U}$ (or $\mathcal{V}$), which we will write as $a \cdot v$, where $a := \|\bar{X}\| \in \BR_{\ge 0}$ and $v := \bar{X}/\|\bar{X}\|$ is a unit vector. Given $\bar{X}$, we consider each vector $X^{(i)} = \bar{X} + y^{(i)} \cdot v + z^{(i)} = (a+y^{(i)}) \cdot v + z^{(i)}$, where $y^{(i)} \in \BR$ and $z^{(i)} \in \BR^d$ is orthogonal to $v$. This means that we decompose $X^{(i)}-\bar{X}$ into a component in the $v$ direction and a component in the hyperplane orthogonal to $v$. (See Figure \ref{fig:decomposition} for an example of this decomposition.) The advantage of this is that one can show that for both the distribution $\mathcal{U}$ and $\mathcal{V}$, the distribution of the vector $v$ is uniform across the unit sphere, and the distribution of $z^{(i)}$ are the same for both $\mathcal{U}$ and $\mathcal{V}$. 

\begin{figure}
    \centering
\begin{tikzpicture}[scale=2]

    \coordinate (A1) at (-0.8,-0.87);
    \coordinate (A2) at (1.6,0.39);
    \coordinate (A3) at (0.73,-1.1);
    \coordinate (A4) at (0.68,0.02);
    \coordinate (A5) at (1.0,0.9);
    \coordinate (A6) at (-0.6,0.22);
    \coordinate (A7) at (0.43,0.6);
    \coordinate (A8) at (0.77,-0.76);
    \coordinate (A9) at (1.1,-0.98);
    \coordinate (A10) at (0.0,-2.0);
    
    \coordinate (A) at (0.48, -0.36);
    \coordinate (O) at (0, 0);

    \node at (A1) {\textbullet};
    \node at (A2) {\textbullet};
    \node at (A3) {\textbullet};
    \node at (A4) {\textbullet};
    \node at (A5) {\textbullet};
    \node at (A6) {\textbullet};
    \node at (A7) {\textbullet};
    \node at (A8) {\textbullet};
    \node at (A9) {\textbullet};
    \node at (A10) {\textbullet};

    \node [red] at (A) {\textbullet};
    
    \draw [red, ->, dashed, thick] (O) -- (0.44, -0.33);
    \draw [blue, ->, dashed, thick] (0.52, -0.39) -- (0.96, -0.72);
    \draw [green, ->, dashed, thick] (0.96, -0.72) -- (0.03, -1.96);
    
    \node[rotate=-37, scale=0.8] at (0.27,-0.08) {$a \cdot v$};
    \node[rotate=-37, scale=0.8] at (0.75,-0.44) {$y^{(i)} \cdot v$};
    \node[rotate=53, scale=0.8] at (0.35,-1.32) {$z^{(i)}$};
    
    \node at (0.46, -0.5) {$\bar{X}$};
    \node at (0.0, -2.15) {$X^{(i)}$};
\end{tikzpicture}

    \caption{In this figure, we have a series of black points $X^{(1)}, \dots, X^{(N)}$ and a red point representing their mean $\bar{X}$. We decompose an individual $X^{(i)}$ as $(a + y^{(i)}) \cdot v + z^{(i)}$, such that $a \cdot v$ equals $\bar{X}$, $y^{(i)} \cdot v$ goes in the same direction as $\bar{X}$ from the origin, and $z^{(i)}$ is perpendicular to $v$. Our reduction ignores the direction of $v$ and the orthogonal component $z^{(i)}$, and focuses on the quantities $a+y^{(i)}$.}
    \label{fig:decomposition}
\end{figure}
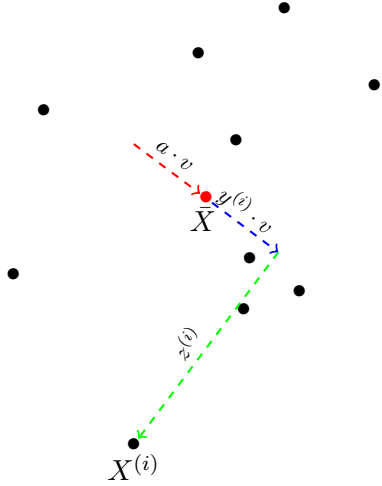

Hence, we can reduce our coupling problem to a single dimension, by considering the distributions $\{a+y^{(i)}\}_{i = 1}^{N}$ and $\{b + y^{(i)}\}_{i = 1}^{N},$ where $a = \|\bar{X}\|$ for $\textbf{X} \sim \mathcal{U}$ and $b = \|\bar{X}\|$ for $\textbf{X} \sim \mathcal{V}.$ 
Note that $a, b$ are also random variables.
One can show that $a$ and $b$ converge (in total variation distance) to roughly $\mathcal{N}(\sqrt{\frac{d}{N}}, \frac{1}{2N})$ and $\mathcal{N}(\sqrt{\frac{d}{N}+\alpha^2}, \frac{1}{2N}),$ respectively. To explain the intuition behind the Normality, since $\bar{X}$ is a spherical Gaussian, $a^2$ behaves exactly like a $\chi^2$-random variable with $d$ degrees of freedom, as the mean of the distribution $\mathcal{U}$ is the origin. Likewise, $b^2$ behaves like an off-centered $\chi^2$-random variable with $d$ degrees of freedom, as the mean of the distribution $\mathcal{V}$ is not the origin but is somewhat close. Therefore, $a$ behaves like a $\chi$ random variable and $b$ behaves like an off-centered $\chi$ random variable, and as $d$ grows, this converges to Normal. (We remark that the rate of convergence of $a, b$ to Normal will not matter: it will just matter that for sufficiently large $d$, the total variation distances between $a, b$ and their respective Normal distributions is less than, say, $0.01$.)
In addition, one can show that $\{y^{(i)}\}$ are distributed as i.i.d. standard Normals minus their mean, which has distribution $\mathcal{N}(0, \frac{1}{N})$. Indeed, one can use this observation to show that if $a, b$ were approximately Normal with variance $\frac{1}{N}$ instead of $\frac{1}{2N},$ then $\{a+y^{(i)}\}_{i = 1}^{N}$ would converge in total variation distance to the distribution of $N$ i.i.d. samples from $\mathcal{N}(\sqrt{\frac{d}{N}}, 1)$, and similarly, $\{b+y^{(i)}\}_{i = 1}^{N}$ would converge in total variation distance to the distribution of $N$ i.i.d. samples from $\mathcal{N}(\sqrt{\frac{d}{N}+2 \alpha^2}, 1)$.

To fix the issue that $a, b$ have the wrong variance, we use the fact that the total variation distance between $\mathcal{N}(0, \frac{1}{N})$ and $\mathcal{N}(0, \frac{1}{2N})$ is less than $0.2$. We use this to show that if could privately distinguish between $\{a+y^{(i)}\}_{i = 1}^{N}$ and $\{b + y^{(i)}\}_{i = 1}^{N}$ with $0.9$ probability, then we could distinguish between $N$ i.i.d. samples from $\mathcal{N}(\sqrt{\frac{d}{N}}, 1)$ and from $\mathcal{N}(\sqrt{\frac{d}{N}+2 \alpha^2}, 1)$ with a weaker $0.7$ probability. Now that we have reduced it to univariate and independent samples, we are in a position to create a coupling, which will indeed get us the correct lower bound.

\subsection{Generalizations}

\paragraph{Theorem \ref{thm:cov_unknown}:} In the case where the samples come from $\mathcal{N}(\mu, \Sigma)$ where $\Sigma$ is unknown, we run into the problem that $T := \|X^{(1)}+\cdots+X^{(N)}\|^2$ is no longer concentrated around $N d$ (if $\mu = 0$) or $Nd + \alpha^2 N^2$ (if $\|\mu\| = \alpha$), since $\Sigma$ is not necessarily the identity matrix. Instead, we will show that $T$ is concentrated around $N \cdot J$ (if $\mu = 0$) or $N \cdot J + \alpha^2 N^2$ (if $\|\mu\| = \alpha$), where $J = Tr(\Sigma)$ is unknown. Because $J$ is unknown, we cannot attempt to directly privately estimate $T$ and accept or reject the hypothesis based on whether $T$ exceeds a threshold.

Instead, we will use the fact that given samples from $\mathcal{N}(\mu, \Sigma)$, we can generate samples from $\mathcal{N}(0, \Sigma)$. This is because if $X, Y \overset{i.i.d.}{\sim} \mathcal{N}(\mu, \Sigma)$ for any $\mu \in \BR^d$, then $\frac{X-Y}{\sqrt{2}} \sim \mathcal{N}(0, \Sigma)$. Therefore, we can use this to privately estimate $J$ and privately estimate $T$, and then accept or reject based on whether our estimate for $T$ significantly exceeds our estimate for $N \cdot J$. Indeed, we can estimate both $T$ and $J$ similarly with the same sample complexity (up to constant factors) as in Theorem \ref{thm:slow_upper} and Theorem \ref{thm:fast_upper}, using the fact that $\Sigma$ is spectrally bounded by $I$. Since estimating $T$ and $J$ are sufficient for the testing problem, we can therefore generalize both Theorems \ref{thm:slow_upper} and \ref{thm:fast_upper}.

\paragraph{Theorem \ref{thm:prod}:} This case will be almost identical to the Gaussian case. For simplicity, we consider testing whether a product distribution $\mathcal{P}(\mu)$ has mean $\mu = 0$ or mean $\mu$ with $\|\mu\|_2 = \alpha$. (Indeed, the $\ell_2$ norm of $\mu$ is asymptotically equal to the total variation distance between $\mathcal{P}(\mu)$ and $\mathcal{P}(0)$.)
For the upper bound, if we let $T_{i, j} := \langle X^{(i)}, X^{(j)}\rangle,$ we are able to obtain the same concentration bounds as in the Gaussian case for each entry $T_{i, j}$ and the sum of each row/column of the matrix $\textbf{T}$, which is required in proving Theorem \ref{thm:fast_upper}. In addition, we are able to obtain the same concentration bounds as in the Gaussian case for the sum $\sum_{i \in S} X^{(i)}$ for all subsets $S \subset [N]$. From here, the rest of the upper bounds proceed in the same way as in the Gaussian case.

For the lower bound, there is a known reduction \cite{CanonneKMUZ20} from private identity testing of Gaussians to private uniformity testing of Boolean products, i.e., any upper bound for uniformity testing of Boolean products also holds for Gaussians. Therefore, our lower bound for identity testing of Gaussians implies a lower bound for uniformity testing of Boolean products.

\paragraph{Theorem \ref{thm:tolerant}:} While in many scenarios tolerant identity testing is more difficult than standard identity testing, it will not be so in our case. For instance, suppose we wished to distinguish between i.i.d. samples from $\mathcal{N}(\mu, I)$, where either $\|\mu\| = \frac{\alpha}{2}$ or $\|\mu\| = \alpha$. In the non-private setting, we can use the fact that the statistic $T$ has expectation $N \cdot d +\frac{\alpha^2}{4} \cdot N^2$ in the former setting, and expectation $N \cdot d +\alpha^2 \cdot N^2$ in the latter setting. In other words, the difference between the mean of the respective statistics is still $\Omega(\alpha^2 \cdot N^2)$. Because of this, we can still do non-private hypothesis testing using the same number of samples.
The private setting will work similarly, as our goal has been to output a private version of the statistic $T$ (or some scaled version of it like $\bar{V}$). Hence, we can apply the same algorithms as before, but change the threshold value accordingly, and still require the same number of samples up to an asymptotic factor.

\section*{Acknowledgments}

The author thanks Piotr Indyk for helpful feedback on this paper, as well as Cl\'{e}ment Canonne for helpful discussions regarding the paper \cite{CanonneKMUZ20}.

\newcommand{\etalchar}[1]{$^{#1}$}

\appendix

\section{Preliminaries} \label{sec:prelim}

\subsection{Notation}

We first note several pieces of notation that we use in our proofs.
\begin{itemize}
    \item We use $\textbf{X} = (X^{(1)}, \dots, X^{(N)})$ to represent a dataset of $N$ points in $\BR^d$. We use $\bar{X}$ to denote either the sum or average of the data points in $\textbf{X}$, depending on context.
    \item For two datasets $\textbf{X}, \textbf{X}'$, we use $\rho(\textbf{X}, \textbf{X}')$ to represent the number of data points $X^{(i)} \neq X'^{(i)}$.
    \item We use $\mathcal{N}(\mu, \Sigma)$ to represent the $d$-dimensional Multivariate Gaussian distribution with mean vector $\mu$ and covariance matrix $\Sigma$. We use $I$ to represent the $d \times d$ identity matrix.
    \item We use $\mathcal{P}(\mu)$ to represent the Boolean product distribution over $\{-1, 1\}^d$ with expectation $\mu$.
    \item We use $\|\cdot\|$ to represent the Euclidean ($\ell_2$) norm of a vector, and $\|\cdot\|_2$ to represent the spectral (operator) norm of a symmetric matrix.
    \item We use $a = \pm b$ to mean $-b \le a \le b$, and $a = b \pm c$ to mean $b-c \le a \le b+c$. 
    \item We use a lazy $\tilde{O}$ notation where $\tilde{O}(f) = f \cdot (\log N)^{O(1)}$ for general $f$, even if $f$ is sub-polynomial in $N$. Note that $N$ will always be at least $\sqrt{d}$, $\alpha^{-1}$, and $\eps^{-1}$, so any logarithmic factors in $d, \alpha^{-1}$, and $\eps^{-1}$ can be incorporated into the $\tilde{O}$ notation.
    \item We use $\TV(\mathcal{D}, \mathcal{D}')$ to represent the total variation distance between two distributions $\mathcal{D}$ and $\mathcal{D}'$ over the same domain.
\end{itemize}

Finally, we will always assume that our privacy parameter $\eps$ and error parameter $\alpha$ are at most a small constant (such as $1/3$).

\subsection{Preliminary Propositions} \label{subsec:basic_prop}

We define a \emph{subgaussian} random variable and recall some important properties of it.

\begin{definition}
    A $1$-dimensional random variable $X$ with mean $\mu$ is \emph{subgaussian} if there exists a constant $C$ such that $\BP(|X-\mu| \ge t) \le 2e^{-t^2/C^2}$ for all $t \ge 0$. We define the \emph{subgaussian norm} of such a variable $X$ to be the infimum over all $C$ such that $\BP(|X-\mu| \ge t) \le 2e^{-t^2/C^2}$ for all $t \ge 0$.
\end{definition}

We note the following standard propositions about subgaussian random variables.

\begin{proposition} \label{prop:normal_subgaussian}
    The univariate Normal distribution $\mathcal{N}(\mu, \sigma^2)$ has subgaussian norm $\sigma$.
\end{proposition}

\begin{proposition} \label{prop:sum_of_subgaussian} \cite{BoucheronLM13}
    If $X_1, \dots, X_n$ are independent variables with each $X_i$ having subgaussian norm $\sigma_i$, then the random variable $X_1+\cdots+X_n$ has subgaussian norm at most $O(\sqrt{\sum_{i = 1}^{n} \sigma_i^2})$.
\end{proposition}

Next, we note the following proposition linking the total variation distance between two multivariate Gaussians with covariance $\Sigma$ and their Mahalanobis distance.

\begin{proposition} \label{prop:mahalanobis_tv} (Folklore)
    For two distributions $\mathcal{N}(\mu, \Sigma)$ and $\mathcal{N}(\nu, \Sigma)$, recall that the \emph{Mahalanobis} distance between these two distributions is $\sqrt{(\mu-\nu)^T \Sigma^{-1} (\mu-\nu)}$. Then, $\TV(\mathcal{N}(\mu, \Sigma), \mathcal{N}(\nu, \Sigma)) = \Theta(\min(\sqrt{(\mu-\nu)^T \Sigma^{-1} (\mu-\nu)}, 1)),$ i.e., the total variation distance and the Mahalanobis distance are asymptotically equivalent assuming the Mahalanobis distance does not exceed $\Theta(1)$.
\end{proposition}

We also have a similar proposition linking the total variation distance between two product distributions with their means.

\begin{proposition} \label{prop:mean_tv} \cite[Lemma 2.8]{CanonneKMUZ20}
    Let $\mathcal{P}(\mu), \mathcal{P}(\mu^*)$ be product distributions over $\{-1, 1\}^d$, where $\mu, \mu^* \in [-1, 1]^d$. Also, suppose that $\mu^* \in [-1/2, 1/2]^d$. Then, $\TV(\mathcal{P}(\mu), \mathcal{P}(\mu^*)) = \Theta(\min(\|\mu-\mu^*\|, 1))$.
\end{proposition}

Proposition
\ref{prop:mean_tv} implies the following are equivalent: distinguishing between $\mu^* = \mu$ and $\|\mu^*-\mu\| \ge \alpha,$ and distinguishing between $\mathcal{P}(\mu) = \mathcal{P}(\mu^*)$ and $\TV(\mathcal{P}(\mu), \mathcal{P}(\mu^*)) \ge \Theta(\alpha)$, assuming that $\mu^* \in [-1/2, 1/2]^d$. Hence, in Theorem \ref{thm:prod}, it is sufficient to distinguish between $\mathcal{P}^* := \mathcal{P}(\mu^*)$ and all distributions $\mathcal{P}(\mu)$ with $\|\mu-\mu^*\| \ge \alpha$. Likewise, in Theorem \ref{thm:tolerant}, if $C$ is sufficiently large, it is sufficient to distinguish between $\mathcal{P}(\mu)$ with $\|\mu-\mu^*\| \le \frac{\alpha}{2}$ and $\mathcal{P}(\mu)$ with $\|\mu-\mu^*\| \ge \alpha$.


\medskip

Next, we recall a basic proposition about sufficient statistics of Normal distributions.

\begin{definition}
    Given a distribution $\mathcal{D}(\theta)$ over $\BR^d$ parameterized by some $\theta$ in a domain $\Theta$, and given samples $X^{(1)}, \dots, X^{(N)} \overset{i.i.d.}{\sim} \mathcal{D}(\theta)$, we say that a function $T = T(X^{(1)}, \dots, X^{(N)})$ is a \emph{sufficient statistic} for $\theta$ if the distribution of $X^{(1)}, \dots, X^{(N)}$ conditioned on $T$ is the same as the distribution of $X^{(1)}, \dots, X^{(N)}$ conditioned on $T$ and $\theta$.
\end{definition}

\begin{proposition} \label{prop:sufficient_statistic}
    Let $\Sigma$ be a fixed covariance matrix, and let $\mathcal{N}(\mu, \Sigma)$ be parameterized only by $\mu$. Then, given samples $X^{(1)}, \dots, X^{(N)},$ the empirical mean $\bar{X} = \frac{X^{(1)}+\cdots+X^{(N)}}{N}$ is a sufficient statistic for $\mu$. In other words, if given $N$ i.i.d. samples from $\mathcal{N}(\mu, \Sigma)$, the conditional distribution of $X^{(1)}, \dots, X^{(N)}$ conditioned on $\frac{X^{(1)}+\cdots+X^{(N)}}{N}$ is independent of $\mu$. 
\end{proposition}

\begin{corollary} \label{cor:sufficient_statistic}
    Let $X^{(1)}, \dots, X^{(N)}$ be distributed i.i.d. as $\mathcal{N}(\mu, \Sigma)$, and let $\bar{X} = \frac{X^{(1)}+\cdots+X^{(N)}}{N}$. Then, for $Z^{(1)}, \dots, Z^{(N)} \overset{i.i.d.}{\sim} \mathcal{N}(0, \Sigma)$ and $\bar{Z} = \frac{Z^{(1)}+\cdots+Z^{(N)}}{N}$, we have that $X^{(1)}, \dots, X^{(N)}$ has the same distribution as $\bar{X}+Z^{(1)}-\bar{Z}, \dots, \bar{X}+Z^{(N)}-\bar{Z}$.
\end{corollary}

\begin{proof}
    It suffices to show they have the same distribution if we condition on $\bar{X}$ and $\bar{Z}$.
    
    By Proposition \ref{prop:sufficient_statistic}, if we know $\bar{X}$, it does not matter whether $X^{(1)}, \dots, X^{(N)}$ were originally drawn from $\mathcal{N}(\mu, \Sigma)$ or from $\mathcal{N}(\mu', \Sigma)$ for some $\mu' \neq \mu$. So, we can instead pretend that $X^{(1)}, \dots, X^{(N)}$ were drawn from $\mathcal{N}(v, \Sigma)$, for $v = \bar{X}-\bar{Z}$. In this case, by an additive shift of $v$, the conditional distribution of $X^{(1)}, \dots, X^{(N)} \overset{i.i.d.}{\sim} \mathcal{N}(v, \Sigma)$ conditioned on $\bar{X}$ is the same as the conditional distribution of $Z^{(1)}+v, \dots, Z^{(N)}+v \overset{i.i.d.}{\sim} \mathcal{N}(0, \Sigma)$ conditioned on $\bar{Z}$. Hence, by writing $v = \bar{X}-\bar{Z}$, we are done.
\end{proof}

Finally, we note the following concentration bound, due to \cite{MassartLaurent}.

\begin{theorem} \cite[Lemma 1]{MassartLaurent} \label{thm:massartlaurent}
    Let $a_1, \dots, a_n$ be nonnegative reals, all bounded by some $|a|_{\infty} > 0$. Let $|a|_2 = \sqrt{a_1^2+\cdots+a_k^2}$. Let $Z_1, \dots, Z_n \overset{i.i.d}{\sim} \mathcal{N}(0, 1)$, and let $Z = \sum_{i = 1}^{n} a_i Z_i^2$. Then, for any positive real $t$, we have that
\[\BP\left(\left|Z - \sum_{i = 1}^{n} a_i\right| \ge 2 |a|_2 \sqrt{t} + 2 |a|_\infty t\right) \le 2e^{-t}.\]
\end{theorem}

\subsection{Reductions} \label{subsec:reductions}

In this subsection, we note several reductions between problems that will make things easier to prove.

First, we note the asymptotic equivalence between $(\eps, \delta)$-DP and $(\eps+\delta, 0)$-DP in the hypothesis testing case. Roughly, this follows because we assume the probability of acceptance/rejection has two-sided error. For instance, if we additively change the probability of acceptance from $1/3$ to $1/3 + \eps$, this is also a multiplicative $1 + O(\eps)$ change in the probability of acceptange, and a multiplicative $1-O(\eps)$ change in the probability of rejection. Formally, we have the following:

\begin{proposition} \label{prop:pure_approx} \cite[Lemma 5]{AcharyaSZ18}
    There is an $(\varepsilon, \delta)$-DP algorithm for a testing problem if and only if there is an $(O(\varepsilon+\delta), 0)$-DP algorithm for the same testing problem.
\end{proposition}

Next, we recall Propositions \ref{prop:mahalanobis_tv} and \ref{prop:mean_tv}. From here, we note that it is asymptotically equivalent to distinguish between the the distribution being far from the null hypothesis distribution in total variation distance and the mean being far from the null hypothesis mean in Mahalanobis distance (in the known-covariance Gaussian case) or $\ell_2$ distance (in the Boolean product case). As we saw in Subsection \ref{subsec:basic_prop}, in the Boolean product case the equivalence also extends to tolerant distribution testing.

Next, we note that by shifting the distribution, we may always assume WLOG (in the multivariate Gaussian cases) that $\mu^* = 0$. In addition, if the covariance is known, we may also assume WLOG that $\Sigma = I$, by multiplying all data points by the matrix $\Sigma^{-1/2}$. 

We also note that in the Boolean product case, we may assume that $\mu^* = 0$. Indeed, we have the following lemma, due to \cite{CanonneKMUZ20}:

\begin{lemma} \cite[Lemma 5.1, rephrased]{CanonneKMUZ20} \label{lem:wlog_product_mean_0}
    Suppose that $\mathcal{P}(\mu^*)$ is a known Boolean product distribution, and we are given a sample $X \sim \mathcal{P}(\mu)$ where $\mu$ is unknown. Then, there exists a randomized transformation $Y = f_{\mu^*}(X)$ such that $Y \sim \mathcal{P}(\frac{\mu-\mu^*}{2})$.
\end{lemma}

Thus, up to changing $\alpha$ by a factor of $2$, we may assume WLOG that $\mu^* = 0$ for identity testing of Boolean products, since $\mu = \mu^*$ is equivalent to $\frac{\mu-\mu^*}{2} = 0$ and $\|\mu-\mu^*\| \ge \alpha$ is equivalent to $\left\|\frac{\mu-\mu^*}{2}\right\| = \frac{\alpha}{2}.$ Note that this reduction also holds for the tolerant testing case as well.

\medskip

Next, we note a proposition showing that testing known-covariance Gaussians is \textbf{easier} than Testing Product distributions. Because of this, the only lower bounds we need to show is for testing known-covariance Gaussians. Of course, testing of unknown-covariance Gaussians and tolerant testing of Gaussians are harder problems, so the same lower bounds must hold, and the proposition we are about to state also implies that the same lower bounds hold for testing (and tolerant testing) of product distributions as well.

\begin{proposition} \cite[Theorem 3.1, rephrased]{CanonneKMUZ20}
    There exists a function $F: \BR^d \to \{-1, 1\}^d$ and an absolute constant $c > 0$ such that the following holds. Suppose that $X \sim \mathcal{N}(\mu, I)$, where $\mu$ is unknown (so, $F$ does not depend on $\mu$). Then:
\begin{itemize}
    \item If $\mu = 0$, the distribution of $F(X)$ is uniform on $\{-1, 1\}^d$.
    \item If $\mu \neq 0$, the distribution of $F(X)$ is a product distribuiton $\mathcal{P}(\mu')$ where $\TV(\mathcal{P}(0), \mathcal{P}(\mu')) \ge c \cdot \min(1, \|\mu\|)$.
\end{itemize}
\end{proposition}
    Hence, we have a reduction from testing the mean of a Gaussian with identity covariance to testing uniformity of a product distribution, which implies that any lower bound for testing the mean of a identity covariance Gaussian implies the same lower bound for testing uniformity of a product distribution.


\medskip

Finally, we note that in all of our cases, our alternative hypothesis may be phrased as $\|\mu\| \ge \alpha$ (since we may assume $\mu^* = 0$ and since we have provided an equivalence between total variation distance being large and $\|\mu\|$ being large in the Boolean product case). In the following proposition, we show that it suffices to consider the weaker alternative hypothesis of $\alpha \le \|\mu\| \le 2 \alpha$.

\begin{proposition} \label{prop:wlog_bounded}
    Let $0 < \alpha \le 1$, and fix some positive integer $N$, where $N \ge \eps^{-1}$. Suppose $\mathcal{A}$ is an algorithm that can $(\eps, 0)$-privately distinguish between $\mathcal{H}_0:$ $N$ samples drawn i.i.d. from $\mathcal{N}(0, I)$, and $\mathcal{H}_1:$ $N$ samples drawn i.i.d. from $\mathcal{N}(\mu, I)$ where $\alpha \le \|\mu\| \le 2 \alpha$. Then, using $N' = O(N \log \frac{d}{\alpha} \log \log \frac{d}{\alpha})$ samples, we can $(O(\eps), 0)$-privately distinguish between $\mathcal{H}_0:$ $N$ samples drawn i.i.d. from $\mathcal{N}(0, I)$, and $\mathcal{H}_1:$ $N$ samples drawn i.i.d. from $\mathcal{N}(\mu, I)$ where $\|\mu\| \ge \alpha$.
\end{proposition}

\begin{proof}
    First, we can learn the mean $\mu$ with $1-o(1)$ probability up to error $O(\sqrt{d})$ and with $(0, \eps)$-privacy, by picking a random point among the sampled $N' \ge N \ge \frac{1}{\eps}$ points and outputting that as our initial guess $\tilde{\mu}$. If the estimate $\tilde{\mu}$ is not $O(\sqrt{d})$ in magnitude, we can instantly reject the null hypothesis and output $1$.
    
    Next, consider partitioning the points $X^{(1)}, \dots, X^{(N')}$ into groups that each contain $N$ points, and the groups are in turn partitioned into ``supergroups'' that each contain $O(\log \log \frac{d}{\alpha})$. There are a total of $O(\log \frac{d}{\alpha})$ supergroups.
    Now, for each point $X^{(j)}$ in the $t$th supergroup, we replace it with $X'^{(j)} := (X^{(j)} + \sqrt{2^{2t-2}-1} \cdot Z_j)/2^{t-1}$, where $Z_j \sim \mathcal{N}(0, I)$ is independent for each point $X^{(j)}$. Note that if $X^{(j)} \sim \mathcal{N}(\mu, I)$, then $X'^{(j)} \sim \mathcal{N}(\mu/2^{t-1}, I)$. Now, for each group, we run $\mathcal{A}$ on the $N$ points $X'^{(j)}$ that are now in it. The output for each group will either be $0$ or $1$. Next, for each $t$, we compute the majority output for all groups in the $t$th supergroup to obtain a bit $b_t$, which is either $0$ or $1$. Our final algorithm returns $0$ if and only if $b_t = 0$ for all $1 \le t \le O(\log \frac{d}{\alpha})$.
    
    First, suppose that the samples $X^{(j)}$ are drawn i.i.d. from $\mathcal{N}(0, I)$. Then, the samples $X'^{(j)}$ are also i.i.d. from $\mathcal{N}(0, I)$. This means that each time we run $\mathcal{A}$ on some group, we output $0$ with probability at least $2/3$, which means that the majority output in any supergroup will be $0$ with probability at least $1-\frac{c}{\log d/\alpha}$ for some small constant $c$. So, each $b_t$ equals $0$ with probability at least $1-\frac{c}{\log d/\alpha}$, which means that all of the $b_t$'s equal $0$ with probability at least $3/4$. In addition, our initial guess $\tilde{\mu}$ has norm at most $O(\sqrt{d})$ with $1-o(1)$ probability, so we output $0$ as our final answer with probability at least $3/4-o(1) \ge 2/3$.
    
    Next, suppose that the samples are drawn i.i.d. from $\mathcal{N}(\mu,I)$, where $O(\sqrt{d}) \ge \|\mu\| \ge \alpha$. Let $t \ge 1$ be the unique integer such that $2^{t-1} \cdot \alpha \le \|\mu\| < 2^t \cdot \alpha$. Note that $t \le O(\log \frac{d}{\alpha})$. Then, every $X'^{(j)}$ in the $t$h supergroup is i.i.d. $\mathcal{N}(\mu/2^{t-1}, I)$, where $\alpha \le \|\mu/2^{t-1}\| \le 2 \alpha$. So, in each group in the $t$th supergroup, $\mathcal{A}$ outputs $1$ with probability at least $2/3$, which means that $b_t = 1$ with probability at least $1-\frac{c}{\log d/\alpha} \ge 2/3$. We may have rejected the null hypothesis immediately, but that only helps us.
    
    Finally, if the samples are drawn i.i.d. from $\mathcal{N}(\mu, I)$, where $\mu \ge O(\sqrt{d})$, then we would have rejected the null hypothesis immediately with probability at least $1-o(1)$. Hence, the overall algorithm is accurate.
    
    To verify that the algorithm is private, note that replacing each $X^{(j)}$ with $X'^{(j)}$ only depends on the value $X^{(j)}$ and the position $j$. In addition, we are performing $(\eps, 0)$-DP algorithms on each group, but the groups are distinct. So, any algorithm that only depends on the output of $\mathcal{A}$ on each group is still $(\eps, 0)$-DP. In addition, the first part of the algorithm of determining $\tilde{\mu}$ is $(0, \eps)$-DP, so the overall algorithm is $(\eps, \eps)$-DP by the weak composition theorem. Thus, we can modify the algorithm to be $(O(\eps), 0)$-DP, by Proposition \ref{prop:pure_approx}.
\end{proof}

\begin{remark} \label{rmk:wlog_bounded}
We note that Proposition \ref{prop:wlog_bounded} extends easily to the cases of unknown but bounded covariance Gaussians, product distributions, and tolerant identity testing. Indeed, for the case of bounded covariance Gaussians, we note that we can replace each $X^{(j)}$ in the $t$th supergroup with $X'^{(j)} := (X^{(j)}+\sqrt{2^{2t-2}-1} \cdot Z_j)/2^{t-1}$ where $Z_j \sim \mathcal{N}(0, I),$ and then each $X'^{(j)}$ in the $t$th supergroup is independent and has distribution $\mathcal{N}(\mu/2^t, (1-1/2^{2t-2}) \cdot I + (1/2^{2t-2}) \cdot \Sigma)$, and note that if $\|\Sigma\|_2 \le 1$, then $\|(1-1/2^{2t-2}) \cdot I + (1/2^{2t-2}) \cdot \Sigma\|_2 \le 1$ as well. For the case of product distributions, if each sample $X^{(j)}$ were drawn from $\mathcal{P}(\mu)$, we can replace it with $X'^{(j)}$ where each coordinate of $X'^{(j)}_i$ is independently equal to $X^{(j)}_i$ with probability $\frac{1}{2}+\frac{1}{2^t}$ and is $-X^{(j)}_i$ with probability $\frac{1}{2}-\frac{1}{2^t}$. Then, note that the distribution of $X'^{(j)}$ is precisely $\mathcal{P}(\mu/2^{t-1})$. Finally, in the case of tolerant testing, note that if $\|\mu\| \le \alpha/2$, then $\|\mu/2^{t-1}\| \le \alpha/2^t \le \alpha/2$ for all $t \ge 1$, so the same reduction in Proposition \ref{prop:wlog_bounded} holds.
\end{remark}

Therefore, in all of our theorems, it suffices to consider the null hypothesis of $\mu = 0$ and the alternative hypothesis of $\alpha \le \|\mu\| \le 2 \alpha$. (Recall that we can assume WLOG that $\mu^* = 0$ always, and that $\Sigma = I$ in the known-covariance Gaussian case.)

\section{Concentration Bounds} \label{sec:concentration}

The main focus of this section is to state and prove various concentration bounds that are crucial in establishing accuracy of our differentially private algorithms. These bounds establish important properties of a matrix $T$ where $T_{i, j} = \langle X^{(i)}, X^{(j)} \rangle$, and each $X^{(i)}$ is drawn i.i.d. either from some Gaussian $\mathcal{N}(\mu, I)$ or some Boolean Product distribution $\mathcal{P}(\mu)$. We have $4$ types of bounds.
\begin{enumerate}
    \item Lemmas \ref{bound:total_gaussian} and \ref{bound:total_product} provide bounds on the sum over all entries in $T$.
    \item Lemmas \ref{bound:entry_gaussian} and \ref{bound:entry_product} provide bounds on individual terms of the matrix $T_{i, j}$.
    \item Proposition \ref{bound:prop_row} and Lemma \ref{bound:lem_row} provide bounds on the sum of a row or column of $T$.
    \item Proposition \ref{bound:gaussian_mgf} to Lemma \ref{bound:product_submatrix} provide bounds on the sum of submatrices of $T$.
\end{enumerate}

We will combine these bounds together in Theorem \ref{thm:concentration}.

We will also prove related results (Propositions \ref{bound:unknown_cov} and \ref{bound:unknown_subgaussian}, and Theorem \ref{thm:concentration_2}) for when each sample $X^{(i)}$ is drawn i.i.d. from some Gaussian $\mathcal{N}(\mu, \Sigma)$, where $\|\Sigma\|_2 \le 1$ is unknown. This will be important for generalizing our upper bound results in Theorem \ref{thm:cov_unknown}.

Finally, we will also prove two results relating to univariate Normal and chi variables (Propositions \ref{prop:chi_normal_tv} and \ref{prop:scaled_normal_tv}), which will be important in establishing our sample complexity lower bound.

\begin{lemma} \cite{CanonneKMUZ20} \label{bound:total_gaussian}
    Let $X^{(1)}, \dots, X^{(N)} \overset{i.i.d.}{\sim} \mathcal{N}(\mu, I)$, and let $\bar{X} = \sum_{i = 1}^{N} X^{(i)}$. Then,
\begin{itemize}
    \item $\BE\left[\|\bar{X}\|^2\right] = Nd + N^2 \cdot \|\mu\|^2$.
    \item $Var\left[\|\bar{X}\|^2\right] = O(N^2 d + N^3 \cdot \|\mu\|^2)$.
\end{itemize}
\end{lemma}

\begin{lemma} \cite{CanonneKMUZ20} \label{bound:total_product}
    Let $X^{(1)}, \dots, X^{(N)} \overset{i.i.d.}{\sim} \mathcal{P}(\mu)$, where $\|\mu\| \le 1$, and let $\bar{X} = \sum_{i = 1}^{N} X^{(i)}$. Then,
\begin{itemize}
    \item $\BE\left[\|\bar{X}\|^2\right] = Nd + \Theta(N^2 \cdot \|\mu\|^2)$.
    \item $Var\left[\|\bar{X}\|^2\right] = O(N^2 d + N^3 \cdot \|\mu\|^2)$.
\end{itemize}
\end{lemma}

\begin{lemma} \label{bound:entry_gaussian}
    Let $X, Y \overset{i.i.d.}{\sim} \mathcal{N}(\mu, I)$, where $\|\mu\| \le 1$. Then, with probability at least $1 - \frac{1}{N^3}$, $\|X\|^2 = d \pm O(\sqrt{d} \cdot \log N)$. Likewise, with probability at least $1 - \frac{1}{N^3}$, $\langle X, Y \rangle = \pm O(\sqrt{d} \cdot \log N)$.
\end{lemma}

\begin{proof}
    By rotational symmetry of Gaussians, we may assume WLOG that $\mu = (\beta, 0, 0, \dots, 0)$, where $\beta = \|\mu\|.$ Then, $\|X\|^2 = (\beta + Z_1)^2 + Z_2^2 + \cdots + Z_d^2,$ where each $Z_i \overset{i.i.d.}{\sim} \mathcal{N}(0, 1)$. First, note that $|Z_1| \le O(\sqrt{\log N})$ with at least $1-\frac{1}{2N^3}$ probability, in which case $(\beta + Z_1)^2 = O(\log N)$. Next, $Z_2^2 + \cdots + Z_d^2 \sim \chi_{d-1}^2$. By Theorem \ref{thm:massartlaurent}, we have that $\BP(|\chi_k^2-k| \ge 2 \sqrt{kt}+2t) \le 2e^{-t}$ for all $t \ge 0$, so by setting $k = d-1$ and $t = O(\log N)$, we have that $\BP(|\chi_{d-1}^2 - (d-1)| \ge O(\sqrt{d \log N}+\log N)) \le \frac{1}{4N^3}$. Therefore, we have that $\BP\left(\left|\|X\|^2-d\right| \ge O(\sqrt{d \log N} + \log N)\right) \le 1 - \frac{1}{N^3}$.
    
    Next, $\langle X, Y \rangle= (\beta + Z_1)(\beta + Z_1') + Z_2Z_2' + \cdots + Z_dZ_d',$ where each $Z_i, Z_i' \overset{i.i.d.}{\sim} \mathcal{N}(0, 1)$. We have $|Z_1|, |Z_1'| \le O(\sqrt{\log N})$ with at least $1-\frac{1}{2N^3}$ probability, in which case $(\beta + Z_1)(\beta+Z_1') = O(\log N)$. Next, $$Z_2Z_2' + \cdots + Z_dZ_d' = \frac{1}{4}\left[(Z_2+Z_2')^2 + \cdots + (Z_d+Z_d')^2 - (Z_2-Z_2')^2 + \cdots + (Z_d-Z_d')^2\right],$$ and using the fact that $Z_i+Z_i', Z_i-Z_i'$ are i.i.d. $\mathcal{N}(0, 2)$ variables, this equals $\frac{1}{2}[A-B]$ where $A, B \overset{i.i.d}{\sim} \chi_{d-1}^2$. As we have already seen, $\BP(|A - (d-1)| \ge O(\sqrt{d \log N}+ \log N)) \le \frac{1}{4N^3}$ and $\BP(|B - (d-1)| \ge O(\sqrt{d \log N}+ \log N)) \le \frac{1}{4N^3}$, so $\BP(|A-B| \ge O(\sqrt{d \log N}+ \log N)) \le \frac{1}{2N^3}$. Therefore, writing $\langle X, Y \rangle = \frac{1}{2}(A-B) + (\beta + Z_1)(\beta+Z_1'),$ we have that $\BP(\langle X, Y \rangle \ge O(\sqrt{d \log N} + \log N)) \le \frac{1}{N^3}$.
\end{proof}

\begin{lemma} \label{bound:entry_product}
    Let $\mu \in [-1, 1]^d$ be such that $\|\mu\| \le 1.$ Let $X, Y \overset{i.i.d.}{\sim} \mathcal{P}(\mu)$. Then, $\|X\|^2 = d$ with probability $1$, and with probability at least $1 - \frac{1}{N^3}$, $\langle X, Y \rangle = \pm O(\sqrt{d \cdot \log N})$.
\end{lemma}

\begin{proof}
    Since $X \in \{-1, 1\}^d$, we always have $\|X\|^2 = 1$. So, the first part of the lemma is immediate.

    For the second part of the lemma, we have $\langle X, Y \rangle = \sum_{i=1}^{d} X_i Y_i$. For each $i \in [d]$, $X_i Y_i$ has mean $\BE[X_i] \cdot \BE[Y_i] = \mu_i^2$. In addition, $X_i Y_i$ is bounded, so it has variance $O(1)$.  By Hoeffding's inequality, with probability at least $1-\frac{1}{N^3}$, $\sum_{i=1}^{d} X_i Y_i$ is not more than $O(\sqrt{d \log N})$ away from its expectation, which is $\sum_{i = 1}^{d} \mu_i^2 = \|\mu\|^2 = O(1)$. This proves the second part of the lemma.
\end{proof}

\begin{proposition} \label{bound:prop_row}
    Fix a vector $v \in \BR^d$ of magnitude $1$ and a vector $\mu \in [-1, 1]^d$, and let $K \ge 1$ be a fixed integer. Let $X^{(1)}, \dots, X^{(K)}$ be distributed either as i.i.d. $d$-dimensional Gaussians $\mathcal{N}(\mu, I)$ or as i.i.d. product distributions $\mathcal{P}(\mu)$, and define $\bar{X} := X^{(1)} + \cdots + X^{(K)}$. Then, $Y := \langle \bar{X} - K \cdot \mu, v \rangle$ is a subgaussian random variable with subgaussian norm at most $O(\sqrt{K}),$ meaning that $\BP[|Y| \ge t] \le 2e^{-\Omega(t^2/K)}$ for all $t \ge 0$.
\end{proposition}

\begin{proof}
    Suppose $\bar{X}$ is distributed as the sum of $K$ i.i.d. $d$-dimensional Gaussians $\mathcal{N}(\mu, I)$. Then, $\bar{X} \sim \mathcal{N}(K \cdot \mu, K \cdot I)$. Then, $\bar{X}-K \cdot \mu = \mathcal{N}(0, K \cdot I)$, so $\langle \bar{X}-K \cdot \mu, v \rangle = \mathcal{N}(0, K)$. Hence it has subgaussian norm at most $\sqrt{K}$.
    
    Suppose $\bar{X}$ is distributed as the sum of $K$ i.i.d. $d$-dimensional product distributions over $\{-1, 1\}^d$ with mean $\mu$. Then, $\bar{X}_i - K \cdot \mu_i$ is the sum of $K$ i.i.d. mean $0$ variables bounded in magnitude by $2$. Hence, each of these bounded variables has subgaussian norm $O(1)$, so $\bar{X}_i - K \cdot \mu_i$ has subgaussian norm $O(\sqrt{K})$ by Proposition \ref{prop:sum_of_subgaussian}.
    Now, note that $\langle \bar{X} - K \cdot \mu, v \rangle = \sum_{i = 1}^{d} v_i (\bar{X}_i - K \cdot \mu_i)$, and the random variables $\bar{X}_i - K \cdot \mu_i$ have mean $0$ and are independent. Therefore, by Proposition \ref{prop:sum_of_subgaussian}, $\sum_{i = 1}^{d} v_i (\bar{X}_i - K \cdot \mu_i)$ has squared subgaussian norm at most $O(K) \cdot \sum_{i = 1}^{d} v_i^2 = O(K)$, so $\langle \bar{X} - K \cdot \mu, v \rangle = \sum_{i = 1}^{d} v_i (\bar{X}_i - K \cdot \mu_i)$ has subgaussian norm at most $O(\sqrt{K})$.
\end{proof}

\begin{lemma} \label{bound:lem_row}
     Let $X^{(1)}, \dots, X^{(N)}$ be distributed either as i.i.d. $d$-dimensional Gaussians $\mathcal{N}(\mu, I)$ or as i.i.d. product distributions $\mathcal{P}(\mu)$, where $\|\mu\| \le 1$. Then, $\langle X^{(1)}, X^{(1)} + \cdots + X^{(N)}\rangle = d \pm \tilde{O}(\sqrt{Nd}+\|\mu\| \cdot N)$ with probability at least $1-\frac{3}{N^3}$.
\end{lemma}

\begin{proof}
    First, note that by Lemmas \ref{bound:entry_gaussian} and \ref{bound:entry_product}, we have that $\|X^{(1)}\|^2 = d \pm O(\sqrt{d} \cdot \log N)$ with probability at least $1 - \frac{1}{N^3}$. Next, since $X^{(1)}$ is independent of $X^{(2)}, \dots, X^{(N)}$, by setting $K = N-1$, $v = X^{(1)}/\|X^{(1)}\|$, and $t = O(\sqrt{(N-1) \log N})$ in Proposition \ref{bound:prop_row}, we have that with probability at least $1-\frac{1}{N^3},$ $\left|\langle X^{(1)}, X^{(2)} + \cdots + X^{(N)} - (N-1) \mu \rangle\right| \le O(\sqrt{N \log N}) \cdot \|X^{(1)}\|$. In addition, we know that $\|X^{(1)}\| \le O\left(\sqrt{d + \sqrt{d} \log N}\right) = O(\sqrt{d \log N})$ with probability at least $1-\frac{1}{N^3}$. Finally, we note that in the Gaussian setting, $\langle X^{(1)}, \mu \rangle$ is distributed as $\mathcal{N}(0, \|\mu\|^2)$, and in the product distribution testing, $\langle X^{(1)}, \mu \rangle$ is distributed as $\sum_{i = 1}^{d} X^{(1)}_i \mu_i$, which has subgaussian norm $O\left(\sqrt{\sum_{i = 1}^{d} \mu_i^2}\right) = O(\|\mu\|)$. Thus, with probability at least $1-\frac{1}{N^3}$, $\langle X^{(1)}, \mu \rangle = O(\sqrt{\log N} \cdot \|\mu\|)$, for both the Gaussian and product settings.
    
    To summarize, we have that
\[\langle X^{(1)}, X^{(1)} + \cdots + X^{(N)}\rangle = \|X^{(1)}\|^2 + \langle X^{(1)}, X^{(2)} + \cdots + X^{(N)} - (N-1) \mu\rangle + (N-1) \cdot \langle X^{(1)}, \mu \rangle,\]
    which with probability at least $1-\frac{3}{N^3}$ is
\[d \pm O\left(\sqrt{d} \log N + \sqrt{N \log N} \cdot \sqrt{d \log N} + N \sqrt{\log N} \cdot \|\mu\|\right) = d \pm \tilde{O}\left(\sqrt{Nd} + \|\mu\| \cdot N\right).\]
\end{proof}

\begin{proposition} \label{bound:gaussian_mgf}
    Let $Z \sim \mathcal{N}(0, 1)$ be a standard univariate Gaussian. Then, for any $-\frac{1}{10} \le \gamma \le \frac{1}{10}$, $\BE[e^{\gamma \cdot Z^2}] \le 1 + \gamma + O(\gamma^2)$.
\end{proposition}

\begin{proof}
    By integrating the PDF of the standard normal, we have that
\begin{align*}
    \BE\left[e^{\gamma \cdot Z^2}\right] &= \int_{-\infty}^{\infty} e^{\gamma \cdot x^2} \cdot \frac{1}{\sqrt{2\pi}} e^{-x^2/2} dx \\
    &= \int_{-\infty}^{\infty} \frac{1}{\sqrt{2\pi}} e^{-(1-2\gamma)x^2/2} dx \\
    &= \frac{1}{\sqrt{1-2\gamma}} = 1 + \gamma + O(\gamma^2)
\end{align*}
    for all $\gamma \in [-1/10, 1/10]$.
\end{proof}

\begin{proposition} \label{bound:subgaussian_mgf}
    Let $K \ge 1$ be an integer, and $Z' = \frac{1}{\sqrt{K}} \sum_{i = 1}^{K} Z'_i$, where each $Z'_i$ is an i.i.d. random variable with support in $[-1, 1]$. Finally, let $Z = Z' - \BE[Z']$, and let $\eta = Var(Z) = Var(Z')$. Then, for any $-\frac{1}{10} \le \gamma \le \frac{1}{10}$, $\BE[e^{\gamma \cdot Z^2}] \le 1 + \eta \cdot \gamma + O(\gamma^2)$.
\end{proposition}

\begin{proof}
    First, by Hoeffding's inequality, we have that for any $t \ge 0$, $\BP(|Z| \ge t) \le 2e^{-t^2/2}$. Therefore,
\begin{align*}
    \BE\left[e^{1/3 \cdot Z^2}\right] &= \int_0^\infty \BP\left(e^{1/3 \cdot Z^2} \ge t\right) dt \\
    &= \int_0^\infty \BP\left(|Z| \ge \sqrt{3 \ln t}\right) dt \\
    &= \int_0^\infty \min\left(1, 2e^{-3 \ln t/2}\right) dt \\
    &\le 1 + \int_1^\infty \frac{2}{t^{3/2}}dt = O(1).
\end{align*}

Now, for all $x \in \BR$, note that $e^x \le 1 + x + x^2 e^{|x|}$, which is immediate by the Taylor series expansion of $x$. Hence, for any $-\frac{1}{10} \le \gamma \le \frac{1}{10}$, $e^{\gamma Z^2} \le 1 + \gamma Z^2 + \gamma^2 Z^4 e^{|\gamma| Z^2}$. Further, we can write $x \le e^x$ for all $x$, which means $Z^2/10 \le e^{Z^2/10}$, so $Z^2 \le 10 e^{Z^2/10}$. Therefore, $$e^{\gamma Z^2} \le 1 + \gamma Z^2 + \gamma^2 Z^4 e^{|\gamma| Z^2} \le 1 + \gamma Z^2 + 10^2 \gamma^2 e^{(|\gamma| + 1/5) Z^2}.$$
Now, using the fact that $\BE[Z^2] = \text{Var}(Z) = \eta$, we have that
\begin{align*}
    \BE\left[e^{\gamma \cdot Z^2}\right] &\le \BE\left[1 + \gamma Z^2 + 10^2 \gamma^2 e^{(|\gamma| + 1/5) Z^2}\right] \\
    &= 1 + \eta \cdot \gamma + O(\gamma^2),
\end{align*}
    since $|\gamma| + \frac{1}{5} \le \frac{1}{3}$.
\end{proof}

\begin{proposition} \label{bound:unbiased_gaussian_tail}
    Let $X^{(1)}, \dots, X^{(K)} \sim \mathcal{N}(0, 1)$ be i.i.d. $d$-dimensional Gaussians. Then, for any sufficiently large $C$, $\|X^{(1)}+\cdots+X^{(K)}\|^2 = Kd \pm C(K \sqrt{Kd}+K^2)$ with failure probability at most $e^{-(C/20) \cdot K}$.
\end{proposition}

\begin{proof}
    Let $\hat{X} = X^{(1)}+\cdots+X^{(K)}$. Note that each coordinate $\hat{X}_i$ has distribution $\mathcal{N}(0, K)$. Therefore, for any $\gamma \in [-0.1, 0.1]$, $\BE[e^{(\gamma/K) \cdot \hat{X}_i^2}] \le 1+\gamma+O(\gamma^2) = e^{\gamma + O(\gamma^2)}$. Therefore, since each coordinate of $\hat{X}$ is independent, by Proposition \ref{bound:gaussian_mgf},
\begin{equation}
    \BE\left[e^{(\gamma/K) \cdot \|\hat{X}\|^2}\right] = \prod_{i = 1}^d \BE[e^{(\gamma/K) \cdot \hat{X}_i^2}] \le e^{d \gamma + O(d \cdot \gamma^2)}.
\end{equation}

    If we set $0 < \gamma < 0.1$, this means that for any $C \ge 0$,
\begin{align*}
    \BP(\|\hat{X}\|^2 \ge Kd + C(K\sqrt{Kd}+K^2)) &= \BP\left(e^{(\gamma/K) \cdot \|\hat{X}\|^2} \ge e^{(\gamma/K) \cdot (Kd + CK \sqrt{Kd}+CK^2)}\right) \\
    &\le e^{d \gamma + O(d \gamma^2) -(\gamma/K) \cdot (Kd + C K\sqrt{Kd}+ CK^2)} \\
    &\le e^{-C \cdot |\gamma| \sqrt{Kd}- C \cdot |\gamma| \cdot K + O(d \gamma^2)}
\end{align*}
    Likewise, if we set $-0.1 < \gamma < 0$, this means that for any $C \ge 0$,
\begin{align*}
    \BP(\|\hat{X}\|^2 \le Kd - C(K\sqrt{Kd}+K^2)) &= \BP\left(e^{(\gamma/K) \cdot \|\hat{X}\|^2} \ge e^{(\gamma/K) \cdot (Kd - CK \sqrt{Kd}-CK^2)}\right) \\
    &\le e^{d \gamma + O(d \gamma^2) -(\gamma/K) \cdot (Kd - C K\sqrt{Kd}- CK^2)} \\
    &\le e^{-C \cdot |\gamma| \sqrt{Kd}- C \cdot |\gamma| \cdot K + O(d \gamma^2)}.
\end{align*}

    In the case that $K \le d/100$, we set $\gamma = \pm \sqrt{K/d}$ to obtain that $e^{-C \cdot |\gamma| \sqrt{Kd}- C \cdot |\gamma| \cdot K + O(d \gamma^2)} \le e^{-C \cdot K + O(K)} \le e^{-(C/2) \cdot K}$, assuming $C$ is sufficiently large. In the case that $K \ge d/100$, we set $\gamma = \pm \frac{1}{10}$ to obtain that $e^{-C \cdot |\gamma| \sqrt{Kd}- C \cdot |\gamma| \cdot K + O(d \gamma^2)} \le e^{-C/10 \cdot K  + O(d/100)} \le e^{-C/20 \cdot K}$, assuming $C$ is sufficiently large. Hence, we have that
\[\BP\left(\|\hat{X}\|^2 = Kd \pm C(K \sqrt{Kd}+K^2)\right) \ge 1 - 2 e^{-C/20 \cdot K}.\]
\end{proof}

\begin{proposition} \label{bound:unbiased_product_tail}
    Let $X^{(1)}, \dots, X^{(K)} \sim \mathcal{N}(0, 1)$ be i.i.d. from $\mathcal{P}(\mu)$, where $\mu \in [-1, 1]^d$. Then, for any sufficiently large $C$, $\|X^{(1)}+\cdots+X^{(K)} - K \cdot \mu\|^2 = (d-\|\mu\|^2) K \pm C(K \sqrt{Kd}+K^2)$ with failure probability at most $e^{-(C/20) \cdot K}$.
\end{proposition}

\begin{proof}
    The proof is almost identical to that of Proposition \ref{bound:unbiased_gaussian_tail}. Now, note that $X^{(k)}_i$ for each $k \in [K]$ is i.i.d., bounded in magnitude by $1$, and has mean $\mu_i$. Therefore, the $i$th coordinate of $\hat{X} := X^{(1)}+\cdots+X^{(K)} - K \cdot \mu$ has a distribution $\sqrt{K} \cdot Z$, where $Z$ is captured by the distribution in Proposition \ref{bound:subgaussian_mgf}, and $Var(Z) = 1-\mu_i^2$. Therefore, for any $\gamma \in [-0.1, 0.1]$, $\BE[e^{(\gamma/K) \cdot \hat{X}_i^2}] \le 1+(1-\mu_i^2)\gamma+O(\gamma^2) = e^{(1-\mu_i^2)\gamma + O(\gamma^2)}$. Therefore, since each coordinate of $\hat{X}$ is independent,
\begin{equation}
    \BE\left[e^{(\gamma/K) \cdot \|\hat{X}\|^2}\right] = \prod_{i = 1}^d \BE[e^{(\gamma/K) \cdot \hat{X}_i^2}] \le \prod_{i = 1}^{d} e^{(1-\mu_i)^2 \gamma + O(\gamma^2)} = e^{(d-\|\mu\|^2)\gamma + O(d \cdot \gamma^2)}.
\end{equation}
    The remainder of the proof is identical, except with replacing $Kd + C(K \sqrt{Kd}+K^2)$ with $K(d-\|\mu\|^2) + C(K\sqrt{Kd}+K^2)$ and $Kd - C(K \sqrt{Kd}+K^2)$ with $K(d-\|\mu\|^2) - C(K\sqrt{Kd}+K^2)$.
\end{proof}

\begin{lemma} \label{bound:gaussian_submatrix}
    Let $X^{(1)}, \dots, X^{(K)}$ be distributed i.i.d. from $\mathcal{N}(\mu, I)$, where $\|\mu\| \le 1$. Then, with probability at least $1-O(N^{-2K})$, we have that $\|X^{(1)}+\cdots+X^{(K)}\|^2 = Kd \pm \tilde{O}(K \sqrt{Kd}+K^2)$.
\end{lemma}

\begin{proof}
    By Proposition \ref{bound:unbiased_gaussian_tail}, if we set $C$ as a sufficiently large multiple of $\log N$, then with probability at least $1-N^{-2K}$, $\|X^{(1)}+\cdots+X^{(K)}-K \mu\|^2 = Kd \pm O(\log N) \cdot (K \sqrt{Kd}+K^2)$. In addition, since $X^{(1)} + \cdots + X^{(K)}-K \mu \sim \mathcal{N}(0, K)$, we have that for any fixed vector $\mu$, $\langle X^{(1)} + \cdots + X^{(K)} - K \mu, \mu \rangle \sim \mathcal{N}(0, K \cdot \|\mu\|^2)$, which with probability at least $1-2N^{-2K}$ does not exceed $O(\sqrt{K \cdot \|\mu\|^2} \cdot \sqrt{K \log N}) = O(K \sqrt{\log N})$ in absolute value.
    
    Therefore, with probability at least $1-O(N^{-2K})$, we have
\begin{align*}
    \|X^{(1)}+\cdots+X^{(K)}\|^2 &= \|X^{(1)}+\cdots+X^{(K)}-K \mu\|^2 + 2 K \langle X^{(1)}+\cdots+X^{(K)}-K \mu, \mu \rangle + K^2 \|\mu\|^2 \\
    &= Kd \pm \tilde{O}(K \sqrt{Kd}+K^2) \pm 2K \cdot \tilde{O}(K) \pm K^2 \\
    &= Kd \pm \tilde{O}(K \sqrt{Kd}+K^2).
\end{align*}
\end{proof}

\begin{lemma} \label{bound:product_submatrix}
    Let $X^{(1)}, \dots, X^{(K)}$ be distributed i.i.d. from $\mathcal{P}(\mu)$, where $\|\mu\| \le 1$. Then, with probability at least $1-O(N^{-2K})$, we have that $\|X^{(1)}+\cdots+X^{(K)}\|^2 = Kd \pm \tilde{O}(K \sqrt{Kd}+K^2)$.
\end{lemma}

\begin{proof}
    By Proposition \ref{bound:unbiased_product_tail}, if we set $C$ as a sufficiently large multiple of $\log N$, then with probability at least $1-N^{-2K}$, $\|X^{(1)}+\cdots+X^{(K)}-K \mu\|^2 = K(d-\|\mu\|^2) \pm O(\log N) \cdot (K \sqrt{Kd}+K^2)$. However, note that $\|\mu\|^2 \le 1,$ so $K(d-\|\mu\|^2) = Kd-O(K)$. In addition, for any fixed vector $\mu$, by Proposition \ref{bound:prop_row}, $\langle X^{(1)} + \cdots + X^{(K)} - K \mu, \mu|$, with probability at least $1-2N^{-2K}$, does not exceed $O(\|\mu\| \cdot \sqrt{K} \cdot \sqrt{K \log N}) = O(K \sqrt{\log N})$.
  
Therefore, with probability at least $1-O(N^{-2K})$, we have
\begin{align*}
    \|X^{(1)}+\cdots+X^{(K)}\|^2 &= \|X^{(1)}+\cdots+X^{(K)}-K \mu\|^2 + 2 K \langle X^{(1)}+\cdots+X^{(K)}-K \mu, \mu \rangle + K^2 \|\mu\|^2 \\
    &= Kd - O(K) \pm \tilde{O}(K \sqrt{Kd}+K^2) \pm 2K \cdot \tilde{O}(K) \pm K^2 \\
    &= Kd \pm \tilde{O}(K \sqrt{Kd}+K^2).
\end{align*}
\end{proof}

Combining these concentration bounds together, we obtain the following theorem.

\begin{theorem} \label{thm:concentration}
    Let $\alpha \le \frac{1}{2}$, and let $\mu \in \BR^d$ satisfy $\|\mu\| \le 2 \alpha$. Let $X^{(1)}, \dots, X^{(N)}$ be drawn i.i.d. either according to $\mathcal{N}(\mu, I)$ or according to $\mathcal{P}(\mu)$. Finally, let $\textbf{T}$ be the matrix with entries $T_{i, j} = \langle X^{(i)}, X^{(j)}\rangle$. Then, with probability at least $0.99$, the following all hold simultaneously.
\begin{enumerate}
    \item We have $\sum_{i = 1}^{N} \sum_{j = 1}^{N} T_{i, j} = Nd + N^2 \|\mu\|^2 \pm O(N \sqrt{d} + \alpha N \sqrt{N})$
    \item For all $i$, $T_{i, i} = d \pm \tilde{O}(\sqrt{d})$ and for all $i \neq j$, $T_{i, j} = \pm \tilde{O}(\sqrt{d})$.
    \item For all $i$, $\sum_{j = 1}^{N} T_{i, j} = d \pm \tilde{O}(\sqrt{Nd} + \alpha \cdot N)$, and for all $j$, $\sum_{i = 1}^{N} T_{i, j} = d \pm \tilde{O}(\sqrt{Nd} + \alpha \cdot N)$.
    \item For all subsets $S \subset [N]$, if $|S| = K$, then $\sum_{i, j \in S} T_{i, j} = Kd \pm \tilde{O}(K \sqrt{Kd}+K^2)$.
\end{enumerate}
\end{theorem}

\begin{proof}
    Part 1 of the theorem is immediate from Lemmas \ref{bound:total_gaussian} and \ref{bound:total_product}, if we use Chebyshev's inequality to make the failure probability $0.005$. Part 2 is immediate from Lemmas \ref{bound:entry_gaussian} and \ref{bound:entry_product}, and since we need this to be true for all pairs $(i, j)$, the failure probability is at most $O(1/N^3) \cdot N^2 = O(1/N)$. 
    
    To see why Part 3 follows from Lemma \ref{bound:lem_row}, note that by symmetry, the result of Lemma \ref{bound:lem_row} also holds for $\sum_{j = 1}^{N} T_{i, j} = \langle X^{(i)}, X^{(1)} + \dots + X^{(N)}\rangle$. Since we need this to be true for all $i$, the failure probability is at most $O(1/N^3) \cdot N = O(1/N^2)$. In addition, since $\textbf{T}$ is a symmetric matrix, we have that the sum of the entries in row $i$ and the sum of the entries in column $i$ are the same, so we also have that for all $j$, $\sum_{i = 1}^{N} T_{i, j} = d \pm \tilde{O}(\sqrt{Nd} + \alpha \cdot N)$.
    
    Finally, Part 4 follows from Lemmas \ref{bound:gaussian_submatrix} and \ref{bound:product_submatrix}, since they imply that $\|\sum_{i \in S} X^{(i)}\|^2 = \sum_{i, j \in S} T_{i, j} = Kd \pm \tilde{O}(K \sqrt{Kd}+K^2)$ with failure probability $O(N^{-2K})$. Union bounding this over all subsets of size $K$ in $[N]$ and over all choices of $K$ still means the failure probability is at most $\sum_{K = 1}^{N} O(N^{-2K}) \cdot {N \choose K} = \sum_{K = 1}^{N} O(N^{-K}) = O(1/N)$.
    
    Hence, the overall failure probability is at most $0.005 + O(1/N) \le 0.01$.
\end{proof}

\begin{proposition} \label{bound:unknown_cov}
    Let $X \sim \mathcal{N}(0, \Sigma)$, where $\|\Sigma\|_2 \le 1.$ In addition, let $J := Tr(\Sigma)$. Then, for any sufficiently large $C$, $\|X\|^2 = J \pm (\sqrt{Cd}+C)$ with probability at least $1-e^{-C/10}$.
\end{proposition}

\begin{proof}
    By rotating, we may assume WLOG that $\Sigma$ is diagonal. Then, note that $\|X\|^2 = \sum_{i = 1}^{d} \Sigma_{ii} Z_i^2$, where each $Z_i \overset{i.i.d.}{\sim} \mathcal{N}(0, 1)$. We apply Theorem \ref{thm:massartlaurent}, where $n = d$, and $a_i = \Sigma_{ii}$ (so $J = \sum_{i=1}^{n} a_i$). Note that $|a|_\infty = \|\Sigma\|_2 \le 1$ and $|a|_2 \le \sqrt{d}$. Thus, we have that $\BP\left(\left|\|X\|^2-J\right| \ge \sqrt{C d} + C\right) \le 2e^{-C/4} \le e^{-C/10}$ if $C$ is sufficiently large.
%
\end{proof}

\begin{proposition} \label{bound:unknown_subgaussian}
    Let $X \sim \mathcal{N}(0, \Sigma)$, where $\|\Sigma\|_2 \le 1.$ Then, for any vector $\mu$ of magnitude at most $2 \alpha$, $\langle \mu, \Sigma \rangle$ has subgaussian norm at most $O(\alpha)$.
\end{proposition}

\begin{proof}
    Note that $\langle \mu, \Sigma \rangle$ is a Gaussian with variance $\mu^T \Sigma \mu \le \|\mu\|^2 \cdot \|\Sigma\|_2 \le 4 \alpha^2$. So, the subgaussian norm is at most $O(\alpha)$.
\end{proof}

\begin{theorem} \label{thm:concentration_2}
    Let $\alpha \le \frac{1}{2}$, and let $\mu \in \BR^d$ satisfy $\|\mu\| \le 2 \alpha$. Let $X^{(1)}, \dots, X^{(N)}$ be drawn i.i.d. either according to $\mathcal{N}(\mu, \Sigma)$, where $\|\Sigma\|_2 \le 1$ is unknown. Finally, let $\textbf{T}$ be the matrix with entries $T_{i, j} = \langle X^{(i)}, X^{(j)}\rangle$. Then, with probability at least $0.99$, the following all hold simultaneously.
\begin{enumerate}
    \item We have $\sum_{i = 1}^{N} \sum_{j = 1}^{N} T_{i, j} = N \cdot J + N^2 \|\mu\|^2 \pm O(N \sqrt{d} + \alpha N \sqrt{N})$
    \item For all $i$, $T_{i, i} = d \pm \tilde{O}(\sqrt{d})$ and for all $i \neq j$, $T_{i, j} = \pm \tilde{O}(\sqrt{d})$.
    \item For all $i$, $\sum_{j = 1}^{N} T_{i, j} = d \pm \tilde{O}(\sqrt{Nd} + \alpha \cdot N)$, and for all $j$, $\sum_{i = 1}^{N} T_{i, j} = J \pm \tilde{O}(\sqrt{Nd} + \alpha \cdot N)$.
    \item For all subsets $S \subset [N]$, if $|S| = K$, then $\sum_{i, j \in S} T_{i, j} = K J \pm \tilde{O}(K \sqrt{Kd}+K^2)$.
\end{enumerate}
\end{theorem}

\begin{proof}
    For part 1 of the theorem, note that $\sum_{i = 1}^{N} \sum_{j = 1}^{N} T_{i, j} = \|\bar{X}\|^2,$ where $\bar{X} = \sum_{i = 1}^{N} X^{(i)} \sim \mathcal{N}(N \cdot \mu, N \cdot \Sigma)$. We can write $\bar{X} = N \cdot \mu + \sqrt{N} \cdot X,$ where $X \sim \mathcal{N}(0, \Sigma)$. Then, $\|\bar{X}\|^2 = N^2 \|\mu\|^2 + 2N^{3/2} \cdot \langle \mu, X \rangle + N \cdot \|X\|^2.$ But we know that since $\langle \mu, X \rangle$ has subgaussian norm at most $O(\alpha)$, so doesn't exceed $O(\alpha)$ in absolute value with probability at least $0.999$. In addition, by Proposition \ref{bound:unknown_cov}, we have that $\|X\|^2 = J \pm O(\sqrt{d})$ with probability at least $0.999$. Together, this implies that $\|\bar{X}\|^2 = N \cdot J + N^2 \|\mu\|^2 \pm O(N \sqrt{d} + \alpha N \sqrt{N})$.
    
    For part 2 of the theorem, note that $T_{i, i}$ has distribution $\|\mu+X\|^2$, where $\|\mu\| \le 1$ and $X \sim \mathcal{N}(0, \Sigma)$. In addition, note that $\|\mu+X\|^2 = \|\mu\|^2 + 2 \langle \mu, X \rangle + \|X\|^2.$ By Proposition \ref{bound:unknown_cov}, $\|X\|^2 = J \pm \tilde{O}(\sqrt{d})$ with probability at least $1-\frac{1}{N^3}$, and by Proposition \ref{bound:unknown_subgaussian}, $\langle \mu, X \rangle = \pm \tilde{O}(1)$ with probability at least $1-\frac{1}{N^3}$. Finally, $\|\mu\|^2 \le 1.$ Overall, this implies that with probability at least $1-\frac{2}{N^2}$, $T_{i, i} = J \pm \tilde{O}(\sqrt{d})$ for all $i$. 
    
    For $i \neq j$, note that $T_{i, j}$ has distribution $\langle \mu+X, \mu+X'\rangle$, where $X, X' \overset{i.i.d.}{\sim} \mathcal{N}(0, \Sigma)$. So, $\langle \mu+X, \mu+X'\rangle = \|\mu\|^2 + \langle \mu, X \rangle + \langle \mu, X' \rangle + \langle X, X' \rangle$. We already know that $\|\mu\|^2, \langle \mu, X \rangle, \langle \mu, X' \rangle$ are all in the range $\pm \tilde{O}(1)$ with probability $1-\frac{2}{N^3}$. In addition, $\langle X, X' \rangle = \frac{1}{4}\left[(X+X')^2-(X-X')^2\right]$, but since $X, X'$ are independent, this means $X+X', X-X'$ are both $\mathcal{N}(0, 2 \cdot \Sigma)$, which means $\|X+X'\|^2, \|X-X'\|^2$ are both in the range $2J \pm \tilde{O}(\sqrt{d})$ with probability at least $1-\frac{2}{N^3}$. Overall, we have that for all $i \neq j$, $T_{i, j} = \langle \mu+X, \mu+X'\rangle = \pm \tilde{O}(\sqrt{d})$ with probability at least $1-O(\frac{1}{N})$.
    
    For part 3 of the theorem, by symmetry it suffices to show that $\sum_{j = 1}^{N} T_{i, j} = d \pm \tilde{O}(\sqrt{Nd}+\alpha \cdot N)$ for any fixed $i$ with probability at least $1-O(\frac{1}{N^2})$. Note that $\sum_{j \neq i} T_{i, j} = \langle X^{(i)}, \sum_{j \neq i} X^{(j)} \rangle$ and $\sum_{j \neq i} X^{(j)}$ has distribution $\mathcal{N}((N-1) \cdot \mu, (N-1) \cdot \Sigma).$ Hence, $\sum_{j \neq i} T_{i, j}$ has distribution $\langle \mu + X, (N-1) \mu + \sqrt{N-1} X' \rangle,$ where $X, X' \overset{i.i.d.}{\sim} \mathcal{N}(0, \Sigma)$. We can write this as $(N-1) \|\mu\|^2 + (N-1) \langle \mu, X \rangle + \sqrt{N-1} \langle \mu, X' \rangle + \sqrt{N-1} \cdot \langle X, X' \rangle$. We know that $(N-1) \|\mu\|^2 = O(N \cdot \alpha^2)$. Also, we already saw that with probability at least $1-O(\frac{1}{N^2})$, $\langle \mu, X \rangle, \langle \mu, X' \rangle = \pm \tilde{O}(\alpha)$, and $\langle X, X' \rangle = \tilde{O}(\sqrt{d}).$ So overall, with failure probability at most $O(\frac{1}{N^2})$, $\sum_{j \neq i} T_{i, j} = \pm \tilde{O}(\alpha \cdot N + \sqrt{Nd}).$ Finally, since $T_{i, i} = J \pm \tilde{O}(\sqrt{d})$ with probability at least $1-O(\frac{1}{N^3}),$ we have that $\sum_{j = 1}^{N} T_{i, j} = J \pm \tilde{O}(\alpha \cdot N + \sqrt{Nd})$ with the desired failure probability.
    
    For part 4, note that for any fixed subset $S \subset [N]$ of size $K$, $\sum_{i, j \in S} T_{i, j} = \|\sum_{i \in S} X^{(i)}\|^2,$ and $\sum_{i \in S} X^{(i)} = K \cdot \mu + \sqrt{K} \cdot X$, for $X \sim \mathcal{N}(0, \Sigma)$. So, $\sum_{i, j \in S} T_{i, j}$ has distribution $K^2 \|\mu\|^2 + 2K \sqrt{K} \langle X, \mu\rangle + K \cdot \|X\|^2$. Note that $K^2 \|\mu\|^2 \le K^2$ always, and since $\langle X, \mu \rangle$ has subgaussian norm at most $O(1)$ by Proposition \ref{bound:unknown_subgaussian}, we have that $2K \sqrt{K} \langle X, \mu\rangle = \pm \tilde{O}(K^2)$ with probability at least $1-N^{-2K}$. Finally, by Proposition \ref{bound:unknown_cov}, we have that $K \cdot \|X\|^2 = K \cdot \left[J \pm \tilde{O}(\sqrt{Kd} + K)\right]$ with probability at least $1-N^{-2K}$. Hence, with probability at least $1-2N^{-2K}$, we have that $\sum_{i, j \in S} T_{i, j} = K \cdot J \pm \tilde{O}(K \sqrt{Kd} + K^2)$. Taking a union bound over all $1 \le K \le N$ and all ${N \choose K} \le N^K$ sets $S$ of size $K$, we have that this holds for all $S$ with probability at least $1-O(\frac{1}{N})$.
\end{proof}

Finally, we also prove a result relating the Normal distribution and the $\chi_d$ distribution, as well as a result relating Normal distributions with different scalings.

\begin{proposition} \label{prop:chi_normal_tv}
    As $d \to \infty$, we have that $\lim_{d \to \infty} \TV\left(\chi_d, \mathcal{N}(\sqrt{d}, \frac{1}{2})\right) = 0$.
\end{proposition}

\begin{proof}
    It is well known that $\chi_d$ has PDF
\begin{align*}
    \frac{1}{2^{(d/2)-1} \Gamma(d/2)} \cdot x^{d-1} e^{-x^2/2} &= \left(1 \pm O(\frac{1}{d})\right) \cdot \frac{1}{2^{(d/2)-1} \cdot \left(\frac{d}{2e}\right)^{d/2} \cdot \sqrt{\frac{2 \pi}{d/2}}} \cdot x^{d-1} e^{-x^2/2} \\
    &= \left(1 \pm O(\frac{1}{d})\right) \cdot \frac{e^{d/2}}{d^{(d-1)/2} \cdot \sqrt{\pi}} \cdot x^{d-1} e^{-x^2/2},
\end{align*}
    where the first equality follows by Stirling's formula and the second equality follows by rearrangement. For $x = \sqrt{d}+\theta$ for any $\theta$ bounded by a constant, this equals
\begin{align*}
    &\hspace{0.5cm} \left(1 \pm O(\frac{1}{d})\right) \cdot \frac{e^{d/2}}{d^{(d-1)/2} \cdot \sqrt{\pi}} \cdot \sqrt{d}^{d-1} \left(1+\frac{\theta}{\sqrt{d}}\right)^{d-1} \cdot e^{-(\sqrt{d}+\theta)^2/2} \\
    &= \left(1 \pm O(\frac{1}{d})\right) \cdot \frac{e^{d/2}}{\sqrt{\pi}} \cdot e^{(\theta/\sqrt{d} - \theta^2/2d) \cdot d \pm O(1/\sqrt{d})} \cdot e^{-(\sqrt{d}+\theta)^2/2} \\
    &= \left(1 \pm O(\frac{1}{\sqrt{d}})\right) \cdot \frac{1}{\sqrt{\pi}} \cdot e^{-\theta^2},
\end{align*}
    where the first equality follows by setting $x = \sqrt{d}+\theta$, the second equality follows using a Taylor expansion for $\ln (1+x)$ for $x = \frac{\theta}{\sqrt{d}}$ and noting that $\theta = \pm O(1),$ and the third equality follows by rearrangement.
    
    Hence, for any fixed constant $C$, the PDF of $\chi_d$ at $\sqrt{d}+\theta$, for any $|\theta| \le C$, equals $\left(1 \pm O(\frac{1}{\sqrt{d}})\right) \cdot \frac{1}{\sqrt{\pi}} \cdot e^{-\theta^2}.$ For any two distributions $\mathcal{F}, \mathcal{G}$ over with PDFs $f(x), g(x),$ it is well known that $\TV(\mathcal{F}, \mathcal{G}) \le \int_{-\infty}^{\infty} \max(f(x)-g(x), 0)$. Hence, by setting $f(x)$ as the PDF of $\mathcal{N}(0, 1/2)$ and $g(x)$ as the PDF of $\chi_d$, we have that
\[\TV\left(\mathcal{N}\left(0, \frac{1}{2}\right), \chi_d\right) \le O\left(\frac{1}{\sqrt{d}}\right) \cdot \int_{-C}^{C} \frac{1}{\sqrt{\pi}} e^{-\theta^2} + \int_{x \not\in [-C, C]} f(x) = O\left(\frac{1}{\sqrt{d}} + e^{-C^2}\right),\]
    which means that $\limsup_{d \to \infty} \TV\left(\chi_d, \mathcal{N}(\sqrt{d}, \frac{1}{2})\right) \le O(e^{-C^2})$. As this is true for all $C$, the result is immediate.
\end{proof}

\begin{proposition} \label{prop:scaled_normal_tv}
    We have that $\TV\left(\mathcal{N}(0, \frac{1}{2}), \mathcal{N}(0, 1)\right) \le \frac{1}{6}$.
\end{proposition}

\begin{proof}
    We can compute the total variation distance as
\[\frac{1}{2} \cdot \int_{-\infty}^{\infty} \left|\frac{1}{\sqrt{\pi}} \cdot e^{-x^2} - \frac{1}{\sqrt{2\pi}} \cdot e^{-x^2/2}\right| dx \le 0.1661,\]
    which is at most $\frac{1}{6}$.
\end{proof}

\section{Proof of Theorem \ref{thm:slow_upper}} \label{sec:slow}

In this section, we prove Theorem \ref{thm:slow_upper}, which establishes a (tight) sample complexity upper bound of $\tilde{O}\left(\frac{d^{1/2}}{\alpha^2} + \frac{d^{1/3}}{\alpha^{4/3} \cdot \eps^{2/3}} + \frac{1}{\alpha \cdot \eps}\right)$ for identity testing of Gaussians with known covariance. We also conclude with a brief section describing the necessity of changing the construction used by~\cite{CanonneKMUZ20} in order to make the construction fully correct.

First, we note by Proposition \ref{prop:wlog_bounded}, we may assume that the null hypothesis $\mathcal{H}_0$ is that we are given $N$ samples i.i.d. from $\mathcal{N}(0, I)$, and the alternative hypothesis $\mathcal{H}_1$ is that we are given $N$ samples i.i.d. from $\mathcal{N}(\mu, I)$, where $\alpha \le \|\mu\| \le 2 \alpha$.

The major tool we use in developing our sample-optimal but inefficient upper bound is the McShane-Whitney Extension theorem, which we now state.

\begin{theorem} \cite{McShane} \label{thm:mcshane}
    Let $\mathcal{X}$ be any metric space with a metric $\rho$, and let $\mathcal{C} \subset \mathcal{X}$. Suppose there is a function $F: \mathcal{C} \to \BR$ and a positive real number $D$ such that $F$ is $D$-Lipschitz on $\mathcal{C}$, i.e., for all $x, y \in \mathcal{C}$, $|F(x)-F(y)| \le D \cdot \rho(x, y)$. Then, there exists an extension $\hat{F}: \mathcal{X} \to \BR$ such that $\hat{F}(x) = F(x)$ for all $x \in \mathcal{C}$, and $\hat{F}$ is $D$-Lipschitz on all of $\mathcal{X}$, meaning that for all $x, y \in \mathcal{X}$, $|F(x)-F(y)| \le D \cdot \rho(x, y)$.
\end{theorem}


We will apply Theorem \ref{thm:mcshane} with $\mathcal{X} = (\BR^d)^N$, i.e., $\mathcal{X}$ is the set of all datasets of size $N$ in $\BR^d$. The distance metric $\rho$ that we use measures the number of differing rows, i.e., for $\textbf{X} = \{X^{(1)}, \dots, X^{(N)}\}$ and $\textbf{X}' = \{X'^{(1)}, \dots, X'^{(N)}\}$, we define $\rho(\textbf{X}, \textbf{X}') = \left|\{i: X^{(i)} \neq X'^{(i)}\}\right|$.
For a dataset $\textbf{X} \in (\BR^d)^N$, we associate with it the matrix $\textbf{T} \in \BR^{N \times N}$ with entries $T_{i, j} := \langle X^{(i)}, X^{(j)} \rangle$.
We will define $\mathcal{C} := \mathcal{C}_\alpha \subset (\BR^d)^N$ to be the set of datasets $\{X^{(1)}, \dots, X^{(N)}\}$ that satisfy Properties 2, 3, and 4 listed in Theorem \ref{thm:concentration}. More precisely, we choose some sufficiently large $L = \poly(\log N)$, and define $\mathcal{C}$ to be the set of all datasets $\textbf{X} \in (\BR^d)^N$ such that 
\begin{itemize}
    \item For all $i \in [N]$, $T_{i, i} = d \pm L \cdot \sqrt{d}$, and for all $i \neq j \in [N]$, $T_{i, j} = \pm L \cdot \sqrt{d}$.
    \item For all $i$, $\sum_{j = 1}^{N} T_{i, j} = d \pm L(\sqrt{Nd}+\alpha N)$, and for all $j$, $\sum_{i = 1}^{N} T_{i, j} = d \pm L(\sqrt{Nd}+\alpha N)$.
    \item For all subsets $S \subset [N]$, if $|S| = K$, then $\sum_{i, j \in S} T_{i, j} = K \cdot \left[d \pm L \cdot (\sqrt{Kd}+K)\right]$.
\end{itemize}
We prove the following lemma.

\begin{lemma} \label{lem:lipschitz_1}
    Let $\textbf{X}, \textbf{X}' \subset \mathcal{C}$ be datasets that differ in $K \le \frac{C}{\eps}$ rows, for some fixed $C$. Let $\bar{X} = \sum_{i = 1}^{N} X^{(i)}$, and $\bar{X}' = \sum_{i = 1}^{N} X'^{(i)}$. Then, 
\[\left|\|\bar{X}\|^2 - \|\bar{X}'\|^2\right| \le 6K \cdot L \cdot \left(\sqrt{Nd} + \alpha N + \frac{C}{\eps}\right).\]
\end{lemma}

\begin{proof}
    Write $\bar{X} = Y+Z$ and $\bar{X}' = Y+Z'$, where $Y$ is the sum of the identical rows in $\textbf{X}$ and $\textbf{X}'$, $Z$ is the sum of the rows in $\textbf{X}$ that are not the same in $\textbf{X}'$, and $Z'$ is the sum of the rows in $\textbf{X}'$ that are not the same in $\textbf{X}$. Let $\textbf{T}$ be the matrix with $T_{i, j} = \langle X^{(i)}, X^{(j)} \rangle$, and $\textbf{T}'$ be the matrix with $T'_{i, j} = \langle X'^{(i)}, X'^{(j)} \rangle$. Let $S \subset [N]$ represent the indices of the rows in $\textbf{X}$ and $\textbf{X}'$ that differ. Then, note that $\langle Z, \bar{X}\rangle$ equals the sum of the entries $T_{i, j}$ with $i \in S, j \in [N]$ (and similarly $\langle Z', \bar{X}'\rangle$ is the sum of the entries $T'_{i, j}$ with $i \in S, j \in [N]$). Also, $\|Z\|^2$ is the sum of the entries $T_{i, j}$ with $i, j \in S$ (and similarly $\|Z'\|^2$ is the sum of the entries $T'_{i, j}$ with $i \in S$).
    
    Next, note that $\|\bar{X}\|^2-\|\bar{X}'\|^2 = 2 \langle Z, \bar{X} \rangle - 2 \langle Z', \bar{X}' \rangle - (\|Z\|^2 - \|Z'\|^2)$. Then, note that $\langle Z, \bar{X} \rangle = K \cdot (d \pm L(\sqrt{Nd} + \alpha \cdot N))$, since it is the sum of the entries of $K$ rows in $T$. Likewise, $\langle Z', \bar{X}' \rangle = K \cdot (d \pm L(\sqrt{Nd} + \alpha \cdot N))$. In addition, since $\|Z\|^2$ is the sum of the entries $T_{i, j}$ with $i, j \in S$, and since $|S| = K$, we have that $\|Z\|^2 = Kd \pm L(K \sqrt{Kd}+K^2)$. Similarly, we also have $\|Z'\|^2 = Kd \pm L(K \sqrt{Kd}+K^2)$.
    
    Therefore,
\begin{align*}
    \|\bar{X}\|^2-\|\bar{X}'\|^2 &= 2 \langle Z, \bar{X} \rangle - 2 \langle Z', \bar{X}' \rangle - (\|Z\|^2 - \|Z'\|^2)\\
    &= \pm K \cdot 4L \cdot (\sqrt{Nd}+\alpha N) \pm 2L \cdot K \cdot (\sqrt{Kd}+K) \\
    &= \pm 6 K L \cdot \left(\sqrt{Nd}+\alpha N + \frac{C}{\eps}\right),
\end{align*}
where the last line follows because $K \le \frac{C}{\eps}$ and $K \le N$.
\end{proof}

Let $C$ be a sufficiently large constant, and let $\Delta := 6L \left(\sqrt{Nd} + \alpha N + \frac{C}{\eps}\right)$.
Define $D := \frac{\eps}{C} \cdot \alpha^2 \cdot N^2$, and suppose that $\Delta \le D$.
We now consider the statistic function $\tilde{T}: \mathcal{C} \to \BR$ that sends a dataset $\textbf{X}$ to $\max(0, \min(\|\bar{X}\|^2 - Nd, \alpha^2 \cdot N^2)),$ where $\bar{X} = X^{(1)} + \cdots + X^{(N)}$. Note that $\tilde{T}(\textbf{X})$ is simply $\|\bar{X}\|^2-Nd$ but clipped by $0$ and $\alpha^2 N^2$.

We first show that $\tilde{T}: \mathcal{C} \to \BR$ is $D$-Lipschitz, where we recall that the distance metric $\rho$ over $\mathcal{C}$ is the number of differing rows.

\begin{lemma} \label{lem:lipschitz_2}
    Assume that $\Delta \le D$. Then, for any integer $K$ and any datasets $\textbf{X}, \textbf{X}' \in \mathcal{C}$ that differ in $K$ rows, $|\tilde{T}(\textbf{X}) - \tilde{T}(\textbf{X}')| \le K \cdot D$.
\end{lemma}

\begin{proof}
    First, assume that $K \ge C \cdot \eps^{-1}$, where we recall that $C$ is a large constant and that $D=\frac{\eps}{C} \cdot \alpha^2 \cdot N^2$. Then, since $\tilde{T}(\textbf{X})$ and $\tilde{T}(\textbf{X}')$ are bounded below by $0$ and bounded above by $\alpha^2 \cdot N^2 = \eps^{-1} \cdot C \cdot D \le K \cdot D$, we indeed have that $|\tilde{T}(\textbf{X}) - \tilde{T}(\textbf{X}')| \le K \cdot D$.
    
    Next, assume that $K \le C \cdot \eps^{-1}$. Then, by Lemma \ref{lem:lipschitz_1}, we have that $\left|\|\bar{X}\|^2 - \|\bar{X}'\|^2\right| \le K \cdot \Delta$. Since $\tilde{T}(\textbf{X})$ simply takes $\|\bar{X}\|^2$, subtracts $Nd$, and then prevents the statistic from being less than $0$ or more than $\alpha^2 \cdot N^2$ (and similar for $\tilde{T}(\textbf{X}')$), we will also have that $|\tilde{T}(\textbf{X})-\tilde{T}(\textbf{X}')| \le K \cdot \Delta$. Since we are assuming that $\Delta \le D$, this also implies that $|\tilde{T}(\textbf{X})-\tilde{T}(\textbf{X}')| \le K \cdot D$.
\end{proof}

Hence, by Theorem \ref{thm:mcshane}, there exists an extension $\hat{T}$ of $\tilde{T}$ from $(\BR^d)^N \to \BR$, such that for any datasets $\textbf{X}, \textbf{X}' \in (\BR^d)^N$ that differ in $K$ rows, $|\hat{T}(\textbf{X})-\hat{T}(\textbf{X}')| \le K \cdot D$. Our final (computationally inefficient) algorithm will be the following. We compute $\hat{T}(\textbf{X}) + Lap(D/\eps)$. If this exceeds $\frac{\alpha^2 \cdot N^2}{2}$, we output $1$, and otherwise we output $0$. (Note that this algorithm is computationally inefficient because computing the extension $\hat{T}$ is computationally inefficient.)

We are now ready to prove Theorem \ref{thm:slow_upper}.

\begin{proof} of Theorem \ref{thm:slow_upper}:
First, we note that this algorithm is $(\eps, 0)$-differentially private because $\hat{T}(\textbf{X})$ has sensitivity at most $D$ whenever we change a single row of $\textbf{X}$, by the properties of our Lipschitz extension.

To verify accuracy, note that if each row of $\textbf{X}$ is drawn i.i.d. from $\mathcal{N}(0, I)$, then with probability at least $0.99$, $\textbf{X} \in \mathcal{C}$ and $\|\bar{X}\|^2 \le Nd \pm C N \sqrt{d}$ by Theorem \ref{thm:concentration}, assuming $C$ is sufficiently large. Hence, with probability at least $0.99$, $\hat{T}(\textbf{X}) = \tilde{T}(\textbf{X}) \le C N \sqrt{d}$, which means that with probability at least $0.9$, $\hat{T}(\textbf{X}) + Lap(D/\eps) \le CN \sqrt{d} + 10 \cdot \frac{D}{\eps},$ which we want to be at most $\frac{\alpha^2 \cdot N^2}{2}$. Indeed, if $C$ is sufficiently large, then $10 \cdot \frac{D}{\eps} \le \frac{1}{4}\cdot \alpha^2 N^2$, so we just need that $\Delta \le D = \frac{\eps}{C} \cdot \alpha^2 \cdot N^2$ and $C N \sqrt{d} \le \frac{1}{4} \cdot \alpha^2 \cdot N^2$.

Next, if each row of $\textbf{X}$ is drawn i.i.d. from $\mathcal{N}(\mu, I)$, where $\alpha \le \|\mu\| \le 2 \alpha,$ then with probability at least $0.99$, $\textbf{X} \in \mathcal{C}$ and $\|\bar{X}\|^2 \ge Nd + N^2 \alpha^2 - C(N \sqrt{d} + \alpha N \sqrt{N})$ by Theorem \ref{thm:concentration}. Hence, with probability at least $0.99$, $\hat{T}(\textbf{X}) = \tilde{T}(\textbf{X}) \ge N^2 \alpha^2 - C(N \sqrt{d} + \alpha N \sqrt{N})$, which means that with probability at least $0.9$, $\hat{T}(\textbf{X}) + Lap(D/\eps) \ge N^2 \alpha^2 - C(N \sqrt{d} + \alpha N \sqrt{N}) - 10 \cdot \frac{D}{\eps}$, which we want to be at least $\frac{\alpha^2 \cdot N^2}{2}.$ If $C$ is sufficiently large, then $10 \cdot \frac{D}{\eps} \le \frac{1}{4}\cdot \alpha^2 N^2$, so we just need that $\Delta \le D$ and that $C(N \sqrt{d} + \alpha N \sqrt{N}) \le \frac{1}{4} \cdot \alpha^2 \cdot N^2$.

Hence, our algorithm is both accurate and private as long as $N$ is sufficiently large so that $\Delta = 6L\left(\sqrt{Nd}+\alpha N + \frac{C}{\eps}\right) \le \frac{\eps}{C} \cdot \alpha^2 \cdot N^2$ and $C(N \sqrt{d} + \alpha N \sqrt{N}) \le \frac{1}{4} \alpha^2 \cdot N^2$. 
Since $C$ is a fixed large constant and $L$ is polylogarithmic in $N$, it suffices that
\[N \ge \tilde{O}\left(\frac{\sqrt{d}}{\alpha^2} + \frac{d^{1/3}}{\alpha^{4/3} \cdot \eps^{2/3}} + \frac{1}{\alpha \cdot \eps}\right).\]
This completes the proof.
\end{proof}

\subsection{On changing the construction of \cite{CanonneKMUZ20}}

An important part of both our argument and the argument of \cite{CanonneKMUZ20} for the inefficient testing algorithm was that the statistic $F$ on the subset $\mathcal{C}$ is $D$-Lipschitz. The previous work of \cite{CanonneKMUZ20} made the subtle mistake of assuming that it suffices for any two adjacent datasets $\mathcal{X}, \mathcal{X}'$ in $\mathcal{C}$ to satisfy $|F(\mathcal{X})-F(\mathcal{X}')| \le D$. In reality, we need that for any datasets $\mathcal{X}, \mathcal{X}' \in \mathcal{C}$ that differ in $K$ positions (for any $K \le N$), $|F(\mathcal{X})-F(\mathcal{X}')| \le D \cdot K$. It is reasonable to believe these definitions are equivalent: indeed, if $\mathcal{X}, \mathcal{X}' \in \mathcal{C}$ differ by $K$ positions, then we can construct a series of adjacent datasets $\mathcal{X} =: \mathcal{X}^{(0)}, \mathcal{X}^{(1)}, \mathcal{X}^{(2)}, \dots, \mathcal{X}^{(N)} := \mathcal{X}'$ where $\mathcal{X}^{(i-1)}$ and $\mathcal{X}^{(i)}$ are adjacent. A natural guess is that if $\mathcal{X}, \mathcal{X}' \in \mathcal{C}$, then one can ensure that each $\mathcal{X}^{(i)}$ is also in $\mathcal{C}.$ If so, then a triangle inequality implies the statistic $F$ is $D$-Lipschitz on $\mathcal{C}$, as long as adjacent datasets in $\mathcal{C}$ have the statistic differ by at most $D$.

Unfortunately, it may be impossible to make the intermediate datasets $\mathcal{X}^{(i)}$ be in $\mathcal{C}$ for the choice of $\mathcal{C}$ used in \cite{CanonneKMUZ20}. Indeed, the choice of $\mathcal{C}$ used in \cite{CanonneKMUZ20} was the set of data $\textbf{X} = \{X^{(1)}, \dots, X^{(N)}\}$ (where each $X^{(i)} \in \BR^d$) such that $\|X^{(i)}\|^2 \le \Delta$ and $|\langle X^{(i)}, X^{(1)} + \cdots + X^{(N)}\rangle| \le \Delta$ for some appropriate choice of $\Delta$. However, even if $N = 2$, $d = 3$, and $K = 2$, we lose our desired property. For instance, define $a = (1, 0, 1), b = (-1, 0, 1), c = (0, 1, 1)$, and $d = (0, -1, 1)$ to be points in $\BR^d$. Let $\mathcal{C}$ be the set of pairs $\{X^{(1)}, X^{(2)}\}$ such that $\|X^{(i)}\|^2 \le 2$ and $|\langle X^{(i)}, X^{(1)}+X^{(2)}\rangle| \le 2$, for both $i = 1$ and $i = 2$. It is clear that the sets $\{a, b\}$ and $\{c, d\}$ are in $\mathcal{C}$, and they differ in $K = 2$ positions. However, any method of $2$ steps used to convert $\{a, b\}$ to $\{c, d\}$ cannot stay in $\mathcal{C}$ throughout, because one can see that none of $\{a, c\}, \{a, d\}, \{b, c\}$, or $\{b, d\}$ are in $\mathcal{C}.$

\section{Proof of Theorem \ref{thm:lower}} \label{sec:lower}

In this section, we prove Theorem \ref{thm:lower}. In other words, we show that any $(0, \eps)$-differentially private algorithm that can successfully distinguish between samples from $\mathcal{N}(0, I)$ and $\mathcal{N}(\mu, I)$ for $\|\mu\| \ge \alpha$ requires at least $\Omega\left(\frac{d^{1/2}}{\alpha^2} + \frac{d^{1/3}}{\alpha^{4/3} \cdot \eps^{2/3}} + \frac{1}{\alpha \cdot \eps}\right)$ samples. We reiterate that in the case of differentially private \emph{hypothesis testing}, the sample complexity required for $(\eps, 0)$-DP and $(0, \eps)$-DP are known to be \emph{asymptotically equivalent} for any $\eps < \frac{1}{2}$ \cite{AcharyaSZ18}.

First, we note that even a non-private algorithm that can successfully distinguish between $N$ samples from $\mathcal{N}(0, I)$ and $\mathcal{N}(\mu, I)$ where $\|\mu\| \ge \alpha$ requires $N \ge \Omega\left(\frac{d^{1/2}}{\alpha^2}\right)$ \cite{SrivastavaD08, CanonneDKS20}. Hence, it suffices to show that a $(0, \eps)$-private algorithm that can successfully distinguish between $N$ samples from $\mathcal{N}(0, I)$ and $\mathcal{N}(\mu, I)$ where $\|\mu\| \ge \alpha$ requires $N \ge \Omega\left(\frac{d^{1/3}}{\alpha^{4/3} \cdot \eps^{2/3}}\right)$ and $N \ge \Omega\left(\frac{1}{\alpha \cdot \eps}\right)$. We remark that in the related scenario of privately testing a binary discrete distribution, the $N \ge \Omega\left(\frac{1}{\alpha \cdot \eps}\right)$ has already been established \cite{AcharyaSZ18}. The proof in the Multivariate Gaussian case will not be significantly different (in fact, this part of the lower bound holds even for a univariate Gaussian), but we include the proof for completeness.

First, we will generalize our previous notion of privately distinguishing between two hypotheses of i.i.d. samples to privately distinguishing between hypotheses of possibly non-i.i.d. samples. Now, we will consider the problem of privately distinguishing between two distributions $\mathcal{U}$ and $\mathcal{V}$ over $\mathcal{X}^N.$ Here, $\mathcal{U}$ and $\mathcal{V}$ need not be i.i.d. samples. For an algorithm $\mathcal{A}: \mathcal{X}^N \to \{0, 1\}$ to $(\eps, \delta)$-privately distinguish between two distributions $\mathcal{U}$ and $\mathcal{V}$ over $\mathcal{X}^N$ (in our case $\mathcal{X} = \BR^d$), it must hold that:
\begin{itemize}
    \item $\mathcal{A}$ is $(\eps, \delta)$-differentially private.
    \item If $\textbf{X} = (X^{(1)}, \dots, X^{(N)}) \sim \mathcal{U}$, then $\BP(\mathcal{A}(X) = 0) \ge \frac{2}{3}$.
    \item If $\textbf{X} = (X^{(1)}, \dots, X^{(N)}) \sim \mathcal{V}$, then $\BP(\mathcal{A}(X) = 1) \ge \frac{2}{3}$.
\end{itemize}

To prove our lower bound on the number of samples $N$, it suffices to prove a lower bound for distinguishing between $\mathcal{U}$ and $\mathcal{V}$ where $\mathcal{U}$ is the distribution of $N$ i.i.d. samples from $\mathcal{N}(0, I)$, and $\mathcal{V}$ is generated by first taking some distribution $\mathcal{D}$ supported only on points $\mu \in \BR^d$ with $\|\mu\|_2 \ge \alpha$, and then sampling $X^{(1)}, \dots, X^{(N)} \sim \mathcal{N}(\mu, I)$, where $\mu \sim \mathcal{D}$. To see why, suppose there existed an algorithm $\mathcal{A}$ that could privately distinguish between $\mathcal{H}_0$ of all samples being i.i.d. $\mathcal{N}(0, I)$ and $\mathcal{H}_1$ of all samples being i.i.d. $\mathcal{N}(\mu, I)$ for any $\|\mu\| \ge \alpha$. Then, $\mathcal{A}$ would be $(\eps, \delta)$-DP, and if $\textbf{X} \sim \mathcal{U}$, then $\BP(\mathcal{A}(X) = 0) \ge \frac{2}{3}$. Finally, if $\textbf{X} = (X^{(1)}, \dots, X^{(N)}) \sim \mathcal{V}$, then $\BP(\mathcal{A}(X) = 1) \ge \frac{2}{3}$. This is true because we can condition on the choice of $\mu \sim \mathcal{D}$, and we always have that $\BP(\mathcal{A}(X) = 1) \ge \frac{2}{3}$ regardless of what choice of $\mu$ we picked.\footnote{This idea of choosing a distribution for $\mu$ is quite standard, and is similar to Yao's Minimax Principle \cite{YaoMinimax}.}

Next, we note the following theorem which will be crucial in establishing our lower bounds.

\begin{theorem} \cite[Theorem 11, rephrased]{AcharyaSZ18} \label{thm:lb_emd}
    Let $\mathcal{X}$ represent the domain that samples are coming from, and assume $\mathcal{U}$ and $\mathcal{V}$ be distributions over $\mathcal{X}^N$ (which may not necessarily be i.i.d. samples), such that there is a coupling between $\mathcal{U},\mathcal{V}$ such that $\BE_{\textbf{X} \sim \mathcal{U}, \textbf{X}' \sim \mathcal{V}} [\rho(\textbf{X}, \textbf{X}')] \le D.$ Then, if there existed an algorithm that could $(0, \eps)$-privately distinguish between $\mathcal{H}_0: \textbf{X} \sim \mathcal{U}$ and $\mathcal{H}_1: \textbf{X} \sim \mathcal{V},$ then $D \ge c_1/\eps$ for some constant $c_1 > 0$.
\end{theorem}

As a corollary, we have the following result.

\begin{corollary} \label{cor:lb_emd}
    Again, let $\mathcal{X}$ represent the domain, and let $\mathcal{U}$, $\mathcal{V}$ be distributions over $\mathcal{X}^N$. Suppose that there exist distributions $\mathcal{U}', \mathcal{V}'$ such that $\TV(\mathcal{U}, \mathcal{U}'), \TV(\mathcal{V}, \mathcal{V}') \le \frac{1}{4}$. In addition, suppose there is a coupling between $\mathcal{U}',\mathcal{V}'$ such that $\BE_{\textbf{X} \sim \mathcal{U}', \textbf{X}' \sim \mathcal{V}'} [\rho(\textbf{X}, \textbf{X}')] \le D.$ Then, if there existed an algorithm that could $(0, \eps)$-privately distinguish between $\mathcal{H}_0: \textbf{X} \sim \mathcal{U}$ and $\mathcal{H}_1: \textbf{X} \sim \mathcal{V},$ then $D \ge c_2/\eps$ for some constant $c_2 > 0$.
\end{corollary}

\begin{proof}
    Suppose there exists an algorithm $\mathcal{A}$ that could $(0, \eps)$-privately distinguish between $\mathcal{H}_0: \textbf{X} \sim \mathcal{U}$ and $\mathcal{H}_1: \textbf{X} \sim \mathcal{V}$. This means that if $\textbf{X} \sim \mathcal{U},$ $\mathcal{A} = 0$ with probability at least $2/3$, and if $\textbf{X} \sim \mathcal{V},$ $\mathcal{A} = 1$ with probability at least $2/3.$ By repeating the algorithm $O(1)$ times and taking the majority, we can improve the value $2/3$ to $11/12$, though the algorithm is now $(0, C \eps)$-DP for some constant $C$. But then, since $\TV(\mathcal{U}, \mathcal{U}') \le 1/4,$ this means that $\left|\BP_{\textbf{X} \sim \mathcal{U}} (\mathcal{A}(\textbf{X}) = 0) - \BP_{\textbf{X} \sim \mathcal{U}'} (\mathcal{A}(\textbf{X}) = 0)\right| \le 1/4$ for any algorithm $\mathcal{A},$ which means that $\BP_{\textbf{X} \sim \mathcal{U}'} (\mathcal{A}(\textbf{X}) = 0) \ge \frac{11}{12}-\frac{1}{4} \ge \frac{2}{3}$. By a symmetric argument, $\BP_{\textbf{X} \sim \mathcal{V}'} (\mathcal{A}(\textbf{X}) = 1) \ge \frac{2}{3}$. Therefore, we can apply Theorem \ref{thm:lb_emd} to conclude that $D \ge c_1/(C \cdot \eps) = c_2/\eps$ for $c_2 = c_1/C$.
\end{proof}

Before we prove the main part of the lower bound (i.e., that $N \ge \Omega\left(\frac{d^{1/3}}{\alpha^{4/3} \cdot \eps^{2/3}}\right)$), we first show how Theorem \ref{thm:lb_emd} implies an $\Omega\left(\frac{1}{\alpha \cdot \eps}\right)$-lower bound. We first have the following proposition.

\begin{proposition} \label{prop:emd_normal}
    Let $\mathcal{U}$ represent the distribution of selecting $\textbf{X} = \{X^{(1)}, \dots, X^{(N)}\} \overset{i.i.d.}{\sim} \mathcal{N}(0, I)$, and let $\mathcal{V}$ represent the distribution of selecting $\textbf{X}' = \{X'^{(1)}, \dots, X'^{(N)}\} \overset{i.i.d.}{\sim} \mathcal{N}(\mu, I),$ where $\mu = (\alpha, 0, 0, \dots, 0) \in \BR^d$, and $0 < \alpha \le 1.$ Then, there exists a coupling between $\mathcal{U}, \mathcal{V}$ such that $\BE_{\textbf{X} \sim \mathcal{U}, \textbf{X}' \sim \mathcal{V}}[\rho(\textbf{X}, \textbf{X}')] \le O(\alpha) \cdot N$.
\end{proposition}

\begin{proof}
    We start by considering univariate Normal distributions. Consider drawing the PDF of $\mathcal{N}(0, 1)$ and of $\mathcal{N}(\alpha, 1)$. Let $S$ be the region of points $(x, y) \in \BR^2$, where $y$ is nonnegative but below the PDF of $\mathcal{N}(0, 1)$ at $x$. Likewise, let $S'$ be the region of points $(x, y) \in \BR^2$, where $y$ is nonnegative but below the PDF of $\mathcal{N}(\alpha, 1)$ at $x$. Both $S, S'$ have area $1$ by definition of PDF, and the intersection of $S, S'$ has area $1-\TV(\mathcal{N}(0, 1), \mathcal{N}(\alpha, 1)) = 1-\Theta(\alpha)$.
    
    Now, we consider the following method of generating samples $X^{(1)}, \dots, X^{(N)}$ and $X'^{(1)}, \dots, X'^{(N)}$ over $\BR^d$. For each $i$, pick a point $(x, y)$ uniformly at random in $S$, and let $X_1^{(i)} = x$. If $(x, y) \in S'$, then we define $(x', y') = (x, y)$. Otherwise, we choose $(x', y')$ uniformly at random from $S' \backslash S$. Finally, we let $X_1'^{(i)} = x'$. Finally, we draw $X_2^{(i)} = X_2'^{(i)}, \dots, X_n^{(i)} = X_n'^{(i)}$ each as an i.i.d. $\mathcal{N}(0, 1)$.
    
    First, note that each $X^{(i)}$ is independent and has the correct distribution. Also, note that each $X'^{(i)}$ is independent. To check the distribution of $X'^{(i)}$, note that we select a point $(x, y)$ in $S \cap S'$ with probability $1-\TV(\mathcal{N}(0, 1), \mathcal{N}(\alpha, 1)) = 1-\Theta(\alpha),$ which equals the area of $S \cap S'$. Hence, with probability equal to this area, we select a uniform point $(x', y')$ in $S \cap S'$, and otherwise, we select a uniform point $(x', y')$ in $S' \backslash S$, so the overall distribution of $(x, y)$ is just uniform from $S'$. Hence, $X_1'^{(i)}$ has the distribution $\mathcal{N}(\alpha, 1)$, which means that $X'^{(i)}$ has the correct distribution.
    
    The final thing we verify is that $\BE[\rho(\textbf{X}, \textbf{X}')] = O(\alpha) \cdot N$. Indeed, note that $X^{(i)} = X'^{(i)}$ whenever we selected $(x, y) \in S \cap S'$, so $\BP(X^{(i)} \neq X'^{(i)})$ is at most $1$ minus the area of $S \cap S'$, which is $O(\alpha)$. Hence, by linearity of expectation, we have that $\BE[\rho(\textbf{X}, \textbf{X}')] = O(\alpha) \cdot N$.
\end{proof}

Note that we immediately get an $\Omega\left(\frac{1}{\alpha \cdot \eps}\right)$ sample complexity lower bound, even in $1$ dimension. Indeed, to determine whether we have $N$ samples from $\mathcal{N}(0, I)$ or from $\mathcal{N}(\mu, I)$, by combining Proposition \ref{prop:emd_normal} and Theorem \ref{thm:lb_emd}, we need that $O(\alpha) \cdot N \ge c_1/\eps,$ which means that $N \ge \Omega\left(\frac{1}{\alpha \cdot \eps}\right)$.

\medskip 

We now set up our main lower bound. Let $\mathcal{D}$ be the distribution $\mathcal{N}(0, \frac{2 \alpha^2}{d} \cdot I)$ conditioned on the magnitude being at least $\alpha$.
Note that $\mathcal{D}$ has no support in the region $\{x: \|x\| \le \alpha\}$. Let $\mathcal{V}$ be the distribution over $\{X^{(1)}, \dots, X^{(N)}\}$ where we first draw $\mu \sim \mathcal{D}$ and then draw each $X^{(i)}$ from $\mathcal{N}(\mu, I)$. Also, let $\mathcal{U}$ be the distribution over $\{X^{(1)}, \dots, X^{(N)}\}$ where we draw each $X^{(i)}$ from $\mathcal{N}(0, I)$. Our goal will be to show the following: unless $N \ge \Omega\left(\frac{d^{1/3}}{\eps^{2/3} \cdot \alpha^{4/3}}\right),$ there exist distributions $\mathcal{U}', \mathcal{V}'$ over $(\BR^d)^N$ such that $\TV(\mathcal{U}, \mathcal{U}'), \TV(\mathcal{V}, \mathcal{V}') \le \frac{1}{4}$, and a coupling of $\mathcal{U}', \mathcal{V}'$ such that $\BE_{\textbf{X} \sim \mathcal{U}', \textbf{X}' \sim \mathcal{V}'}[\rho(\textbf{X}, \textbf{X}')] < c_2/\eps$.

We will also need the following two basic propositions.

\begin{proposition} \label{prop:lb_1}
    Let $X^{(1)}, \dots, X^{(N)} \sim \mathcal{N}(0, I)$ and $\bar{X} = \frac{X^{(1)}+\cdots+X^{(N)}}{N}$. Then, $\|\bar{X}\| \sim \frac{1}{\sqrt{N}} \cdot \chi_d$, where $\chi_d$ is the square root of a chi-square variable with $d$ degrees of freedom.
\end{proposition}

\begin{proof}
    This is immediate from the fact that $\bar{X} \sim \mathcal{N}(0, \frac{1}{N} \cdot I) = \frac{1}{\sqrt{N}} \cdot \mathcal{N}(0, I)$. Thus, $\|\bar{X}\|^2 = \frac{1}{N} \cdot \sum_{i = 1}^d Z_i^2,$ where each $Z_i \overset{i.i.d.}{\sim} \mathcal{N}(0, 1)$, which completes the proof.
\end{proof}

\begin{proposition} \label{prop:lb_2}
    Let $X^{(1)}, \dots, X^{(N)} \sim \mathcal{N}(\mu, I)$, where $\mu \sim \mathcal{N}(0, \frac{2 \alpha^2}{d} I)$, and let $\bar{X} = \frac{X^{(1)}+\cdots+X^{(N)}}{N}$. Then, $\|\bar{X}\| \sim \sqrt{\frac{1}{N} + \frac{2 \alpha^2}{d}} \cdot \chi_d$.
\end{proposition}

\begin{proof}
    This is immediate from the fact that $\bar{X} = \mu + \mathcal{N}(0, \frac{1}{N} \cdot I) = \mathcal{N}(0, \frac{2 \alpha^2}{d} + \frac{1}{N})$.
\end{proof}


We now prove the following result.

\begin{theorem} \label{thm:emd}
    Let $\mathcal{D}, \mathcal{U}, \mathcal{V}$ be as defined previously. Suppose that $d \ge C$ and $\frac{d}{\alpha^2} \ge C \cdot N$ for some sufficiently large constant $C$. Then, unless $N \ge \Omega\left(\frac{d^{1/3}}{\eps^{2/3} \cdot \alpha^{4/3}}\right)$, there exist distributions $\mathcal{U}', \mathcal{V}'$ over $(\BR^d)^N$ such that $\TV(\mathcal{U}, \mathcal{U}') \le 0.25$, $\TV(\mathcal{V}, \mathcal{V}') \le 0.25$, along with a coupling of $\mathcal{U}', \mathcal{V}'$ such that $\BE_{\textbf{X} \sim \mathcal{U}', \textbf{X}' \sim \mathcal{V}'} [\rho(\textbf{X}, \textbf{X}')] \le \frac{c_2}{\eps}$.
\end{theorem}

\begin{proof}
We first show how to sample $X^{(1)}, \dots, X^{(N)}$ in a different manner that ends up with the same distribution as $\mathcal{U}$.
By Corollary \ref{cor:sufficient_statistic},
the distribution of $X^{(1)}, \dots, X^{(N)}$ is the same as the distribution of $\bar{X}+Z^{(1)}-\bar{Z}, \bar{X}+Z^{(2)}-\bar{Z}, \dots, \bar{X}+Z^{(N)}-\bar{Z}$, where $Z^{(1)}, \dots, Z^{(N)} \overset{i.i.d.}{\sim} \mathcal{N}(0, I)$ and $\bar{Z} = \frac{Z^{(1)}+\cdots+Z^{(N)}}{N}.$
By Proposition \ref{prop:lb_1} and the rotational symmetry of Gaussians, we can sample $\bar{X} = a \cdot v,$ where $v$ is drawn uniformly from the unit sphere, and $a \sim \sqrt{\frac{1}{N}} \cdot \chi_d$.
Again using the rotational symmetry of Gaussians, we can write $Z^{(i)} = y^{(i)} \cdot v + z^{(i)}$, where we draw each $y^{(i)}$ i.i.d. from the univariate distribution $\mathcal{N}(0, 1)$ and each $z^{(i)}$ i.i.d. from the distribution $\mathcal{N}(0, I - vv^T),$ or equivalently, from the distribution $\mathcal{N}(0, I)$ but projected onto the subspace orthogonal to $v$. Hence, if we write $y'^{(i)} := y^{(i)} - \frac{y^{(1)}+\cdots+y^{(N)}}{N}$ and $z'^{(i)} := z^{(i)} - \frac{z^{(1)}+\cdots+z^{(N)}}{N}$ for all $i \in [N]$, we obtain the following equivalent sampling procedure to sampling from $\mathcal{U}$.
First, sample $v$ uniformly from the unit sphere and $a \sim \sqrt{\frac{1}{N}} \cdot \chi_d$, and compute $\bar{X} = a \cdot v$. Next, sample each $y'^{(i)}, z'^{(i)}$ as above, and let $X^{(i)} = (a+y'^{(i)}) \cdot v + z'^{(i)}$ for all $i \in [N]$.

Next, we replace the distribution $a \sim \sqrt{\frac{1}{N}} \cdot \chi_d$ with $a' \sim \mathcal{N}\left(\sqrt{\frac{d}{N}}, \frac{1}{N}\right)$. By Proposition \ref{prop:chi_normal_tv}, if $d$ is sufficiently large, then $\TV(\chi_d, \mathcal{N}(\sqrt{d}, 1/2)) \le 0.01$, and by Proposition \ref{prop:scaled_normal_tv}, $\TV(\mathcal{N}(\sqrt{d}, 1/2),$ $\mathcal{N}(\sqrt{d}, 1)) \le \frac{1}{6}.$ Therefore, if $d$ is sufficiently large, then $\TV(\chi_d, \mathcal{N}(\sqrt{d}, 1)) \le 0.01 + \frac{1}{6} \le 0.2$. By scaling, we have that $\TV\left(\sqrt{\frac{1}{N}} \cdot \chi_d, \mathcal{N}\left(\sqrt{\frac{d}{N}}, \frac{1}{N}\right)\right) \le 0.2$. Thus, if we let $\mathcal{U}'$ be the distribution where we sample $v$ uniformly from the unit sphere, $a' \sim \mathcal{N}\left(\sqrt{\frac{d}{N}}, \frac{1}{N}\right)$, and $X^{(i)} = (a'+y'^{(i)}) \cdot v + z'^{(i)}$ for all $i$ (where $y'^{(i)}, z'^{(i)}$ are sampled as before), we have that $\TV(\mathcal{U}, \mathcal{U}') \le 0.2$. 

Now, note that $a'$ has the same distribution as $\sqrt{\frac{d}{N}} + \bar{w},$ where $\bar{w} \sim \mathcal{N}(0, \frac{1}{N})$ is the distribution of averaging $N$ i.i.d. $\mathcal{N}(0, 1)$ variables. Hence, by Corollary \ref{cor:sufficient_statistic}, the tuple $\{a' + y'^{(i)}\}_{i = 1}^{N}$ has the same overall distribution as $\{\sqrt{\frac{d}{N}} + y^{(i)}\}_{i = 1}^{N}$, where we recall that $y^{(i)} \overset{i.i.d.}{\sim} \mathcal{N}(0, 1)$. So, to generate a sample from $\mathcal{U}'$ we can instead sample $X^{(i)} = \left(\sqrt{\frac{d}{N}} + y^{(i)}\right) \cdot v + z'^{(i)}$ for all $i$, where each $y^{(i)} \overset{i.i.d.}{\sim} \mathcal{N}(0, 1)$.


\medskip

We similarly generate a method of sampling from $\mathcal{V}'$ with $\TV(\mathcal{V}', \mathcal{V})$ small.
First, we modify $\mathcal{D}$ to $\mathcal{D}_1$ by removing the conditioning on the magnitude being at least $\alpha$. Note that if $\mu \sim \mathcal{D}_1,$ $\|\mu\|^2 \sim \frac{2 \alpha^2}{d} \cdot \chi_d^2$. Thus, basic concentration of $\chi_d^2$ distributions, if $d$ is sufficiently large then $\|\mu\|^2 \ge \alpha^2$ with probability at least $0.99$. Hence, $\TV(\mathcal{D}, \mathcal{D}_1) \le 0.01$. This also implies that if we let $\mathcal{V}_1$ be the distribution over $\{X^{(1)}, \dots, X^{(N)}\}$ where we first select $\mu \sim \mathcal{D}_1$ and then draw each $X^{(i)} \overset{i.i.d.}{\sim} \mathcal{N}(\mu, I)$, then $\TV(\mathcal{V}, \mathcal{V}_1) \le 0.01$.

Next, by Proposition \ref{prop:lb_2} and a similar calculation as in the $\mathcal{U}$ case, we can generate a sample from $\mathcal{V}_1$ by writing $X^{(i)} = (b + y'^{(i)}) \cdot v + z'^{(i)},$ where $v, y'^{(i)}, z'^{(i)}$ are generated in the same manner as previously, but $b$ is distributed as $\sqrt{\frac{1}{N} + \frac{2 \alpha^2}{d}} \cdot \chi_d$. Since $\TV(\chi_d, \mathcal{N}(\sqrt{d}, 1)) \le 0.2$, we have that by scaling, $$\TV\left(\sqrt{\frac{1}{N}+\frac{2\alpha^2}{d}} \cdot \chi_d, \mathcal{N}\left(\sqrt{\frac{d}{N} + 2 \alpha^2}, \frac{1}{N}+\frac{2\alpha^2}{d}\right)\right) \le 0.2.$$
Using the fact that $\frac{1}{N} + \frac{2 \alpha^2}{d} \le \frac{1+\frac{2}{C}}{N}$ assuming that $\frac{d}{\alpha^2} \ge C \cdot N$, and letting $C$ be sufficiently large, we have that
$$\TV\left(\sqrt{\frac{1}{N}+\frac{2\alpha^2}{d}} \cdot \chi_d, \mathcal{N}\left(\sqrt{\frac{d}{N} + 2 \alpha^2}, \frac{1}{N}\right)\right) \le 0.21.$$

Therefore, we can replace the distribution $b \sim \sqrt{\frac{1}{N}+\frac{2 \alpha^2}{d}} \cdot \chi_d$ with $b' \sim \mathcal{N}\left(\sqrt{\frac{d}{N} + 2\alpha^2}, \frac{1}{N}\right)$. 
From here, we consider the distribution $\mathcal{V}'$ over $\{X'^{(1)}, \dots, X'^{(N)}\}$ where we draw each $X'^{(i)} \sim (b' + y'^{(i)}) \cdot v + z'^{(i)}$. Since the total variation distance between $\sqrt{\frac{1}{N}+\frac{2 \alpha^2}{d}} \cdot \chi_d$ and $\mathcal{N}\left(\sqrt{\frac{d}{N} + 2 \alpha^2}, \frac{1}{N}\right)$ is at most $0.21,$ we have that $\TV(\mathcal{V}_1, \mathcal{V}') \le 0.21$. Thus, $\TV(\mathcal{V}, \mathcal{V}') \le \TV(\mathcal{V}, \mathcal{V}_1) + \TV(\mathcal{V}_1, \mathcal{V}') \le 0.01 + 0.21 \le 0.22$. In addition, as in the $\mathcal{U}$ case, we have that samples from $\{b' + y'^{(i)}\}_{i = 1}^{N}$ has the same distribution as samples from $\left\{\sqrt{\frac{d}{N} + 2 \alpha^2} + y^{(i)}\right\}$, where each $y^{(i)} \overset{i.i.d.}{\sim} \mathcal{N}(0, 1)$. So, we can draw samples from $\mathcal{V}'$ by sampling $v$ uniformly from the unit sphere, sampling $y^{(i)}, z'^{(i)}$ as before, and letting $X^{(i)} = \left(\sqrt{\frac{d}{N}+2\alpha^2}+y^{(i)}\right) \cdot v + z'^{(i)}$ for all $i$.

To finish, we have to establish a coupling between $\mathcal{U}'$ and $\mathcal{V}'$. To do this, we sample the same $v$ uniformly on the sphere for both $\mathcal{U}'$ and $\mathcal{V}'$, and the same choices of $z'^{(i)}$ for each $i \in [N]$ for both $\mathcal{U}'$ and $\mathcal{V}'$. Next, we note that we can couple the samples of $\sqrt{\frac{d}{N}} + y^{(i)}$ and the samples of $\sqrt{\frac{d}{N} + 2 \alpha^2} + y^{(i)},$ using the fact that $\sqrt{\frac{d}{N} + 2\alpha^2} - \sqrt{\frac{d}{N}} = \frac{\sqrt{d}}{\sqrt{N}} \cdot \left(\sqrt{1 + \frac{2\alpha^2 N}{d}} - 1\right) = \frac{\sqrt{d}}{\sqrt{N}} \cdot O\left(\frac{2\alpha^2 N}{d}\right) = O\left(\frac{\alpha^2 \sqrt{N}}{\sqrt{d}}\right)$.
Therefore, by setting $\gamma = \sqrt{\frac{d}{N} + 2\alpha^2} - \sqrt{\frac{d}{N}} = O\left(\frac{\alpha^2 \sqrt{N}}{\sqrt{d}}\right)$ and applying Proposition \ref{prop:emd_normal} in the one-dimensional case, we have that there exists a coupling between $\mathcal{U}'$ and $\mathcal{V}'$ that differ in $O\left(\frac{\alpha^2 \sqrt{N}}{\sqrt{d}} \cdot N\right) \le c_3\left(\frac{\alpha^2 \cdot N^{3/2}}{\sqrt{d}}\right)$ samples for some constant $c_3$. We wish to show that this is at most $\frac{c_2}{\eps}$, unless $\Omega\left(\frac{d^{1/3}}{\eps^{2/3} \cdot \alpha^{4/3}}\right)$. But note that if $c_3\left(\frac{\alpha^2 \cdot N^{3/2}}{\sqrt{d}}\right) \ge \frac{c_2}{\eps}$, then $N \ge \left(\frac{c_2}{c_3}\right)^{2/3} \cdot \frac{d^{1/3}}{\alpha^{4/3} \eps^{2/3}},$ as desired. This completes the proof.
\end{proof}

We are now ready to prove Theorem \ref{thm:lower}.

\begin{proof} of Theorem \ref{thm:lower}:
    First, we already know that $N \ge \Omega\left(\frac{\sqrt{d}}{\alpha^2}\right)$ even in the non-private setting, and that $N \ge \Omega\left(\frac{1}{\alpha \cdot \eps}\right)$ using Proposition \ref{prop:emd_normal}. So, it suffices to show that $N \ge \Omega\left(\frac{d^{1/3}}{\alpha^{4/3} \eps^{2/3}}\right).$
    
    Now, if $d \ge C$ and $\frac{d}{\alpha^2} \ge C \cdot N$ for some sufficiently large constant $C$, we will apply Theorem \ref{thm:emd}. Consider the null hypothesis $\mathcal{H}_0$ where the distribution $\{X^{(1)}, \dots, X^{(N)}\}$ is generated by sampling each $X^{(i)} \overset{i.i.d.}{\sim} \mathcal{N}(0, I)$, and the alternative hypothesis $\mathcal{H}_1$ where the distribution is generated by sampling $\mu \sim \mathcal{D}$ and $\{X^{(1)}, \dots, X^{(N)}\}$ is generated by sampling each $X^{(i)} \overset{i.i.d.}{\sim} \mathcal{N}(\mu, I)$. Since $\mathcal{D}$ only has support on $\|\mu\|_2 \ge \alpha$, this is a valid alternative hypothesis. In this case, by Corollary \ref{cor:lb_emd} and Theorem \ref{thm:emd}, we must have that if an algorithm can $(0, \eps)$-privately distinguish between $\mathcal{H}_0$ and $\mathcal{H}_1$, then $N \ge \Omega\left(\frac{d^{1/3}}{\alpha^{4/3} \eps^{2/3}}\right)$.
    
    Finally, we consider the cases where $d \le C$ or $\frac{d}{\alpha^2} \le C \cdot N$. If $d \le C$, then $d = O(1)$, which means that $\frac{d^{1/3}}{\alpha^{4/3} \eps^{2/3}} \le O\left(\frac{1}{\alpha \cdot \eps} + \frac{\sqrt{d}}{\alpha^2}\right)$. This is true because when $d = O(1)$, $\frac{1}{\alpha^{4/3} \eps^{2/3}}$ is a weighted geometric average of $\frac{1}{\alpha \cdot \eps}$ and $\frac{1}{\alpha^2}$.
    So, the lower bound of $\Omega\left(\frac{1}{\alpha \cdot \eps} + \frac{\sqrt{d}}{\alpha^2}\right)$ implies the lower bound of $\Omega\left(\frac{d^{1/3}}{\alpha^{4/3} \eps^{2/3}}\right)$ if $d = O(1)$. If $\frac{d}{\alpha^2} \le C \cdot N,$ then $N \ge \Omega\left(\frac{d}{\alpha^2}\right)$. In addition, we also know that $N \ge \Omega\left(\frac{1}{\alpha \cdot \eps}\right)$. It is simple to verify that $\frac{d^{1/3}}{\alpha^{4/3} \cdot \eps^{2/3}}$ is a weighted geometric average of $\frac{d}{\alpha^2}$ and $\frac{1}{\alpha \cdot \eps}$, so we therefore also have $N \ge \Omega\left(\frac{d^{1/3}}{\alpha^{4/3} \cdot \eps^{2/3}}\right)$. This completes the proof.
\end{proof}


\section{Proof of Theorem \ref{thm:fast_upper}} \label{sec:fast}

In this section, we prove Theorem \ref{thm:fast_upper}.
Again, we may apply Proposition \ref{prop:wlog_bounded}, and assume WLOG the null hypothesis $\mathcal{H}_0$ is that we are given $N$ samples i.i.d. from $\mathcal{N}(0, I)$, and the alternative hypothesis $\mathcal{H}_1$ is that we are given $N$ samples i.i.d. from $\mathcal{N}(\mu, I)$, where $\alpha \le \|\mu\| \le 2 \alpha$. In addition, due to Proposition \ref{prop:pure_approx}, it suffices to show $(0, O(\eps))$-DP rather than $(\eps, 0)$-DP.

The main technical contribution in our proof is the following theorem, which provides an algorithm that can test whether the sum of the entries of a matrix is large or small, but the algorithm's output does not change significantly even if we alter an entire row or column of the matrix. 

\begin{theorem} \label{thm:matrix}
    Fix parameters $\gamma \le 1$ and $L \ge 1$. Suppose that $\gamma \cdot N^2 \ge C \left(\frac{L \cdot \log N}{\eps} + \frac{\log^2 N}{\eps^2}\right)$ for some sufficiently large constant $C$. Then, there exists an algorithm $\mathcal{A}$ that runs in $\text{poly}(N)$ time on matrices in $[-1, 1]^{N \times N}$ and outputs either $0$ or $1$, with the following properties.
\begin{itemize}
    \item For any matrices $\textbf{V}, \textbf{V}' \in [-1, 1]^{N \times N}$ that only differ in a single row (or only differ in a single column), $\left|\BP[\mathcal{A}(\textbf{V}) = 0]-\BP[\mathcal{A}(\textbf{V}') = 0]\right| \le 2\eps$.
    \item For any matrix $\textbf{V} \in [-1, 1]^{N \times N}$ that has every row and column sum at most $L$ in absolute value, and with sum of all entries at most $\frac{\gamma \cdot N^2}{4}$, then $\BP[\mathcal{A}(\textbf{V}) = 0] \ge 0.99.$
    \item For any matrix $\textbf{V} \in [-1, 1]^{N \times N}$ that has every row and column sum at most $L$ in absolute value, and with sum of all entries at least $\frac{3 \gamma \cdot N^2}{4}$, then $\BP[\mathcal{A}(\textbf{V}) = 1] \ge 0.99.$
\end{itemize}
\end{theorem}

\begin{proof}
    Consider setting a threshold $\tau := L + \frac{16 K}{\eps}$, where $K := \log_2 N$. Now, consider the following set of functions $\{f_k\}_{k = 1}^{K}$. For each $k \in [N]$, we define $f_k(x) := \max\left(\min(1, \frac{|x|}{\tau}-k), 0\right)$, i.e., it takes the function $\frac{|x|}{\tau}-k$ and prevents it from exceeding $1$ or becoming negative.

    Now, given a matrix $\textbf{V} \in [-1, 1]^{N \times N}$ and for each $k \in [K]$, we define the statistic
\[F_k(\textbf{V}) := \sum_{i = 1}^{N} f_k\biggr(\sum_{j = 1}^{N} V_{i, j}\biggr) + \sum_{j = 1}^{N} f_k\biggr(\sum_{i = 1}^{N} V_{i, j}\biggr).\]
    We prove the following fact.

\begin{proposition} \label{prop:ub_1}
    Consider two matrices $\textbf{V}, \textbf{V}' \in [-1, 1]^{N \times N}$ that only differ either in a single row or in a single column, and consider some $k \ge 2$. Then,
\[|F_k(\textbf{V}) - F_k(\textbf{V}')| \le 1 + \frac{4}{\tau} \cdot F_{k-1}(\textbf{V}).\]
\end{proposition}

\begin{proof}
    By symmetry of $F_k$ with respect to rows and columns, assume WLOG that $\textbf{V}, \textbf{V}'$ only differ in a single row $i$. Then, we can bound the absolute difference between $F_k(\textbf{V})$ and $F_k(\textbf{V}')$ by considering only how the $i$th row changes, and how each column changes at the $i$th row position. As a result, we have
\[
    |F_k(\textbf{V}) - F_k(\textbf{V}')| \le \left|f_k\biggr(\sum_{j = 1}^{N} V_{i, j}\biggr) - f_k\biggr(\sum_{j = 1}^{N} V'_{i, j}\biggr)\right| + \sum_{j = 1}^{N} \left|f_k\biggr(\sum_{i' = 1}^{N} V_{i', j}\biggr) - f_k\biggr(\sum_{i' = 1}^{N} V'_{i', j}\biggr)\right|.
\]
    
    We note that since the range of $f_k$ is bounded in the range $[0, 1]$, the first part of the summand above is at most $1$. To bound the second part of the summand, we note that for all $j \in [N]$, $\sum_{i' = 1}^{N} V_{i', j}$ and $\sum_{i' = 1}^{N} V'_{i', j}$ differ by at most $2$, since only a single row changes and each entry remains in $[-1, 1]$. In this case, the difference between $f_k\left(\sum_{i' = 1}^{N} V_{i', j}\right)$ and $f_k\left(\sum_{i' = 1}^{N} V'_{i', j}\right)$ is at most $\frac{2}{\tau}$, and in fact only changes if either $\sum_{i' = 1}^{N} V_{i', j}$ or $\sum_{i' = 1}^{N} V'_{i', j}$ is at least $k \cdot \tau$ in absolute value. This means that $\sum_{i' = 1}^{N} V_{i', j}$ is at least $k \cdot \tau - 2$ in absolute value, which means that $f_{k-1}\left(\sum_{i' = 1}^{N} V_{i', j}\right) \ge 1 - \frac{2}{\tau} \ge \frac{1}{2}$. Overall, this means that for all $j \in [N]$,
\[\left|f_k\left(\sum_{i' = 1}^{N} V_{i', j}\right) - f_k\left(\sum_{i' = 1}^{N} V'_{i', j}\right)\right| \le \frac{4}{\tau} \cdot f_{k-1}\left(\sum_{i' = 1}^{N} V_{i', j}\right).\]

    Combining everything, we have that if $\textbf{V}, \textbf{V}' \in [-1, 1]^{N \times N}$ differ in a single row, then
\begin{align*}
    |F_k(\textbf{V}) - F_k(\textbf{V}')| &\le \left|f_k\biggr(\sum_{j = 1}^{N} V_{i, j}\biggr) - f_k\biggr(\sum_{j = 1}^{N} V'_{i, j}\biggr)\right| + \sum_{j = 1}^{N} \left|f_k\biggr(\sum_{i' = 1}^{N} V_{i', j}\biggr) - f_k\biggr(\sum_{i' = 1}^{N} V'_{i', j}\biggr)\right| \\
    &\le 1 + \sum_{j = 1}^{N} \frac{4}{\tau} f_{k-1}\left(\sum_{i' = 1}^{N} V_{i', j}\right) \\
    &\le 1 + \frac{4}{\tau} \cdot F_{k-1}(\textbf{V}).
\end{align*}
\end{proof}

    Now, let $\eta := \frac{8 \cdot K}{\eps \cdot \tau}$, and define $F(\textbf{V}) := \sum_{k = 1}^{K} \lfloor \frac{F_k(\textbf{V})}{2 + 2N \cdot \eta^{k-1}} \rfloor$. Recall that $\tau \ge \frac{16 K}{\eps}$, which means that $\frac{16 K}{\tau \eps} \le 1$, and $\eta \le \frac{1}{2}$. We now show the following.

\begin{proposition} \label{prop:ub_2}
    Suppose that $\textbf{V}, \textbf{V}' \in [-1, 1]^{N \times N}$ only differ either in a single row or in a single column. If $F(\textbf{V}) \le K/\eps$, then $|F(\textbf{V}) - F(\textbf{V}')| \le K$.
\end{proposition}

\begin{proof}
    Again, by symmetry of $F_k$ with respect to rows and columns, assume WLOG that $\textbf{V}, \textbf{V}'$ only differ in a single row $i$.
    It suffices to show that for all $k \in [K]$, $|F_k(\textbf{V}) - F_k(\textbf{V}')| \le 2 + 2N \cdot \eta^{k-1}$. This would imply that $\lfloor \frac{F_k(\textbf{V})}{2 + 2N \cdot \eta^{k-1}} \rfloor$ and $\lfloor \frac{F_k(\textbf{V}')}{2 + 2N \cdot \eta^{k-1}} \rfloor$ differ by at most $1$ for all $k \in [K]$, which is sufficient.
    
    For $k = 1,$ we have that $F_1(\textbf{V}), F_1(\textbf{V}')$ are trivially between $0$ and $2N$, and since $2 + 2N \cdot \eta^{k-1} = 2+2N$, the claim is immediate.
    
    For $k \ge 2,$ note that $|F_k(\textbf{V}) - F_k(\textbf{V}')| \le 1 + \frac{4}{\tau} \cdot F_{k-1}(\textbf{V})$ by Proposition \ref{prop:ub_1}. In addition, since $F(\textbf{V}) \le \frac{K}{\eps}$, we have that $F_{k-1}(\textbf{V}) \le (2 + 2N \cdot \eta^{k-2}) \cdot (\frac{K}{\eps}+1) \le (2 + 2N \cdot \eta^{k-2}) \cdot \frac{2K}{\eps}$. Therefore, 
\[|F_k(\textbf{V}) - F_k(\textbf{V}')| \le 1 + \frac{4}{\tau} \cdot (2+2N \cdot \eta^{k-2}) \cdot \frac{2K}{\eps} \le 1 + \frac{16 K}{\tau \eps} + 2N \cdot \frac{8K}{\tau \eps} \cdot \eta^{k-2} \le 2 + 2N \cdot \eta^{k-1}.\]
\end{proof}

    Next, we consider the piecewise linear function
\[g(x) := \begin{cases} x & |x| \le (K+2) \cdot \tau \\ 2(K+2) \cdot \tau - x & x \ge (K+2) \cdot \tau \\ -2(K+2) \cdot \tau - x & x \le -(K+2) \cdot \tau \end{cases}\]
    and define
\[G(\textbf{V}) := \sum_{i = 1}^{N} g\biggr(\sum_{j = 1}^{N} V_{i, j}\biggr) + \sum_{j = 1}^{N} g\biggr(\sum_{i = 1}^{N} V_{i, j}\biggr).\]
    First note that for any real numbers $x, y,$ that $\left|(g(x)-g(y)) - (y-x)\right| \le 4 (K+2) \cdot \tau$. In addition, note that $g$ is a $1$-Lipschitz function, meaning $|g(x)-g(y)| \le |x-y|$ for all $x, y \in \BR$. Given this, we can prove the following.

\begin{proposition} \label{prop:ub_3}
    Suppose that $\textbf{V}, \textbf{V}' \in [-1, 1]^{N \times N}$ only differ either in a single row or in a single column. If $F(\textbf{V}) \le K/\eps$, then $|G(\textbf{V}) - G(\textbf{V}')| \le 4(K+2) \tau + \frac{48 K}{\eps}.$
\end{proposition}

\begin{proof}
    Again, by symmetry of $G$ with respect to rows and columns, assume WLOG that $\textbf{V}, \textbf{V}'$ only differ in a single row $i$.
    To compute $G(\textbf{V})-G(\textbf{V}')$, note that we only have to consider how the $i$th row changes, and how each column sum changes due to the $i$th row position. For each $j \in [N]$, define $a_j := V'_{i, j} - V_{i, j}$. Then, $\sum_{j = 1}^{N} V'_{i, j} - \sum_{j = 1}^{N} V_{i, j} = \sum_{j = 1}^{N} a_j,$ which means that $g\left(\sum_{j = 1}^{N} V'_{i, j}\right) - g\left(\sum_{j = 1}^{N} V_{i, j}\right) = -\sum_{j = 1}^{N} a_j \pm 4(K+2) \tau$. 
    
    Next, we consider how each column $j$ changes. We note that for any fixed $j \in [N]$, $\sum_{i' = 1}^{N} V'_{i', j} - \sum_{i' = 1}^{N} V_{i', j} = a_j$. In addition, if $\left|\sum_{i' = 1}^{N} V_{i', j}\right| \le (K+1) \cdot \tau,$ then $g(x) = x$ for $x = \sum_{i' = 1}^{N} V_{i', j}$ and $x = \sum_{i' = 1}^{N} V'_{i', j}$, since $\left|\sum_{i' = 1}^{N} V'_{i', j}\right| \le (K+1) \cdot \tau + |a_j| \le (K+1) \cdot \tau + 2 \le (K+2) \cdot \tau$. Therefore, if $\left|\sum_{i' = 1}^{N} V_{i', j}\right| \le (K+1) \cdot \tau,$ then $g\left(\sum_{i' = 1}^{N} V'_{i', j}\right) - g\left(\sum_{i' = 1}^{N} V_{i', j}\right) = a_j$. Alternatively, because $g$ is $1$-Lipschitz, we still have that $\left|g\left(\sum_{i' = 1}^{N} V'_{i', j}\right) - g\left(\sum_{i' = 1}^{N} V_{i', j}\right)\right| \le 2,$ which means that $g\left(\sum_{i' = 1}^{N} V'_{i', j}\right) - g\left(\sum_{i' = 1}^{N} V_{i', j}\right) = a_j \pm 4$.
    
    Now, note that since $F(\textbf{V}) \le K/\eps,$ this means that $F_K(\textbf{V}) \le (2+2N \cdot \eta^{K-1}) \cdot \left(\frac{K}{\eps}+1\right) \le (2+2N \cdot \eta^{K-1}) \cdot \frac{2K}{\eps}$. Since $\eta \le \frac{1}{2}$ and $K = \log_2 N$, this means that $F_K(\textbf{V}) \le \frac{12 K}{\eps}$. This, in turn, implies that the number of columns $j$ such that $\left|\sum_{j = 1}^{N} V_{i, j}\right| > (K+1) \cdot \tau$ is at most $\frac{12 K}{\eps}.$ So, we have that
\[\sum_{j = 1}^{N} \left[g\left(\sum_{i' = 1}^{N} V'_{i', j}\right) - g\left(\sum_{i' = 1}^{N} V_{i', j}\right)\right] = \left[\sum_{j = 1}^{N} a_j\right] \pm 4 \cdot \frac{12 K}{\eps}.\]
    
    So, overall, we have that
\[G(\textbf{V}) - G(\textbf{V}') = - \sum_{j = 1}^{N} a_j \pm 4 (K+2) \tau + \sum_{j = 1}^{N} a_j \pm 4 \cdot \frac{12 K}{\eps} = \pm \left[4(K+2) \tau + \frac{48 K}{\eps}\right].\]
\end{proof}
    
    The algorithm $\mathcal{A}$ will work as follows. First, we compute $F(\textbf{V}).$ With probability $\min\left(1, F(\textbf{V}) \cdot \frac{\eps}{K}\right),$ we will automatically output $1$. Alternatively, if we have not yet output $1$, we compute $\tilde{G}(\textbf{V}) := G(\textbf{V}) + Lap\left(\frac{4(K+2) \tau}{\eps} + \frac{48 K}{\eps^2}\right)$. If $\tilde{G}(T)$ is more than $\gamma \cdot N^2,$ we output $1$; otherwise we output $0$. 
    
    We will prove that this algorithm is private up to changing an entire row or column of $\textbf{V}$, and then we prove that this algorithm is accurate. Together, these complete the proof of Theorem \ref{thm:matrix}.

\begin{proposition}
    For any matrices $\textbf{V}, \textbf{V}' \in [-1, 1]^{N \times N}$ that only differ on a single row (or only differ on a single column), $\left|\BP[\mathcal{A}(\textbf{V}) = 0]-\BP[\mathcal{A}(\textbf{V}') = 0]\right| \le 2\eps$.
\end{proposition}

\begin{proof}
    First, note that if both $F(\textbf{V}) \ge \frac{K}{\eps}$ and $F(\textbf{V}') \ge \frac{K}{\eps}$, then $\mathcal{A}$ automatically outputs $1$ with probability $1$ for both $\textbf{V}$ and $\textbf{V}'$. 
    
    So, we assume that either $F(\textbf{V}) < \frac{K}{\eps}$ or $F(\textbf{V}') < \frac{K}{\eps}$. In this case, Proposition \ref{prop:ub_2} tells us that $|F(\textbf{V}) - F(\textbf{V}')| \le K$, which means that the probability of automatically outputting $1$ before even computing $\tilde{G}(T)$ does not change by more than $\eps$. Next, we note that by Proposition \ref{prop:ub_3}, $|G(\textbf{V}) - G(\textbf{V}')| \le 4(K+2) \tau + \frac{48 K}{\eps}.$ Hence, by applying a Laplace mechanism scaled by $\frac{1}{\eps} \cdot \left(4(K+2) \tau + \frac{48 K}{\eps}\right),$ we have that the probability of $\tilde{G}(\textbf{V})$ being more than $\gamma \cdot N^2$ and the probability of $\tilde{G}(\textbf{V}')$ being more than $\gamma \cdot N^2$ do not differ by more than $\eps$. Overall, we have that $\left|\BP[\mathcal{A}(\textbf{V}) = 0]-\BP[\mathcal{A}(\textbf{V}') = 0]\right| \le 2\eps$, as desired.
\end{proof}

\begin{proposition}
    Suppose that $\textbf{V} \in [-1, 1]^{N \times N}$ has every row and column sum at most $L$ in absolute value. Also, suppose that $\gamma \cdot N^2 \ge C \left(\frac{L \cdot \log N}{\eps} + \frac{\log^2 N}{\eps^2}\right)$. Then, if the sum of the entries of $\textbf{V}$ is at most $\frac{\gamma \cdot N^2}{4},$ then $\BP[\mathcal{A}(T) = 0] \ge 0.99$, and if the sum of the entries of $\textbf{V}$ is at least $\frac{3 \gamma \cdot N^2}{4},$ then $\BP[\mathcal{A}(T) = 1] \ge 0.99$.
\end{proposition}

\begin{proof}
    Note that since every row and column sum is at most $L \le \tau$ in absolute value, we have that for all $k \in [K]$, $f_k(\sum_{j=1}^{N} V_{i, j}) = 0$ for all rows $i$, and $f_k(\sum_{i=1}^{N} V_{i, j}) = 0$ for all columns $j$. Therefore, $F_k(\textbf{V}) = 0$ for all $k \in [K]$, so with probability $1$ we compute $\tilde{G}(T)$.
    
    Likewise, we have that $g(\sum_{j=1}^{N} V_{i, j}) = \sum_{j=1}^{N} V_{i, j}$ for all rows $i$, and $g(\sum_{i=1}^{N} V_{i, j}) = \sum_{j=1}^{N} V_{i, j}$ for all columns $j$. Therefore,
\[G(\textbf{V}) = \sum_{i = 1}^{N} \sum_{j = 1}^{N} V_{i, j} + \sum_{j = 1}^{N} \sum_{i = 1}^{N} V_{i, j} = 2 \cdot \sum_{i = 1}^{N} \sum_{j = 1}^{N} V_{i, j}.\]
    In other words, $G(\textbf{V})$ is just twice the sum of the entries of $\textbf{V}$. 
    
    So, if the sum of all entries in $\textbf{V}$ is at most $\frac{\gamma N^2}{4}$, then $G(\textbf{V}) \le \frac{\gamma N^2}{2}.$ So, assuming that $\gamma \cdot N^2 \ge 10 \left(\frac{4 (K+2) \tau}{\eps} + \frac{48 K}{\eps^2}\right),$ we have that the probability that $\tilde{G}(\textbf{V}) \ge \gamma N^2$ is at most $e^{-5} \le 0.01$. Likewise, if the sum of all entries in $\textbf{V}$ is at least $\frac{3 \gamma N^2}{4}$, then $G(\textbf{V}) \le \frac{3\gamma N^2}{2}.$ So, assuming that $\gamma \cdot N^2 \ge 10 \left(\frac{4 (K+2) \tau}{\eps} + \frac{48 K}{\eps^2}\right),$ we have that the probability that $\tilde{G}(\textbf{V}) \le \gamma N^2$ is at most $e^{-5} \le 0.01$. But since $\tau = L + \frac{16 K}{\eps}$, it suffices for 
\[\gamma \cdot N^2 \ge 10\left(\frac{4 (K+2) L}{\eps} + \frac{64 (K+2) K}{\eps^2} + \frac{48 K}{\eps^2}\right) = O\left(\frac{L \cdot \log N}{\eps} + \frac{\log^2 N}{\eps^2}\right).\]
    Thus, $\mathcal{A}$ succeeds with at least $0.99$ probability either if the sum of the entries is at least $\frac{3 \gamma N^2}{4}$ or if the sum of the entries is at most $\frac{\gamma N^2}{4}$.
\end{proof}
    
\end{proof}

We now return to the scenario of having $N$ samples $X^{(1)}, \dots, X^{(N)}.$ 
We first note the following corollary of Theorem \ref{thm:concentration}.

Given $N$ samples $X^{(1)}, \dots, X^{(N)} \in \BR^d$, we define the matrix $\textbf{T} \in \BR^{N \times N}$ such that $T_{i,j} = \langle X^{(i)}, X^{(j)} \rangle$ for all $i, j \in [N]$. Next, we modify the matrix to create a matrix $\textbf{V}$ such that $V_{i, i} = \frac{T_{i,i}-d}{R}$ for all $i \in [N]$ and $V_{i, j} = \frac{T_{i, j}}{R}$ for all $i \neq j \in [N]$, where $R = \tilde{O}(\sqrt{d})$ is a sufficiently large poly-logarithmic multiple of $\sqrt{d}$. As a direct corollary of Theorem \ref{thm:concentration}, we have the following.

\begin{corollary} \label{cor:concentration_fast_upper}
    Suppose that $X^{(i)} \overset{i.i.d.}{\sim} \mathcal{N}(\mu, I_d)$ for all $i \in [N]$, where $0 \le \|\mu\|_2 \le 2 \alpha$. Then, with probability at least 0.99,
\begin{itemize}
    \item The sum of the entries of $\textbf{V}$ equals $\frac{N^2 \cdot \|\mu\|_2^2}{R} \pm O\left(\frac{N \sqrt{d} + \alpha N \sqrt{N}}{R}\right)$.
    \item All individual entries in $\textbf{V}$ are in the range $[-1,1]$.
    \item The sum of each row and each column is at most $\sqrt{N} + \alpha \cdot \frac{N}{\sqrt{d}}$ in absolute value.
\end{itemize}
\end{corollary}

\begin{proof}{of Theorem \ref{thm:fast_upper}:}
    Our computationally efficient algorithm works as follows. First, sample $X^{(1)}, \dots, X^{(N)} \in \BR^d$, and construct the matrix $\textbf{V}$ as above. Now, let $\gamma = \frac{\alpha^2}{R}$ and $L = \sqrt{N} + \alpha \cdot \frac{N}{\sqrt{d}}$. Finally, we apply the algorithm $\mathcal{A}$ in Theorem \ref{thm:matrix} on $\textbf{V}$. 
    
    Suppose that $N^2 \cdot \alpha^2 \ge C(N \sqrt{d}+\alpha N \sqrt{N})$ for some sufficiently large constant $C$ and that $\gamma \cdot N^2 \ge C \left(\frac{L \log N}{\eps} + \frac{\log^2 N}{\eps^2}\right)$. Then, for any $\textbf{X}, \textbf{X}' \in (\BR^d)^N$ that differ in a single row, the respective matrices $\textbf{V}, \textbf{V}'$ differ in a single row and a single column, so $|\BP[\mathcal{A}(\textbf{V}) = 0]-\BP[\mathcal{A}(\textbf{V}') = 0]| \le 4 \eps$, i.e., the algorithm is $(0, 4 \eps)$-differentially private. Next, if $X^{(1)}, \dots, X^{(N)} \overset{i.i.d.}{\sim} \mathcal{N}(0, I)$, then the sum of the entries of $V$ is at most $\frac{C}{4} \cdot \frac{N \sqrt{d} + \alpha N \sqrt{N}}{R} \le \frac{1}{4} \cdot \frac{N^2 \cdot \alpha^2}{R} = \frac{1}{4} \cdot \gamma \cdot N^2$ with probability at least $0.99$, which means the algorithm outputs $0$ with probability at least $0.97$. Finally, if $X^{(1)}, \dots, X^{(N)} \overset{i.i.d.}{\sim} \mathcal{N}(\mu, I)$, for some $\alpha \le \|\mu\|_2 \le 2 \alpha$, then the sum of the entries of $V$ is at least $\frac{N^2 \cdot \alpha^2}{R} - \frac{C}{4} \cdot \frac{N \sqrt{d} + \alpha N \sqrt{N}}{R} \ge \frac{N^2 \cdot \alpha^2}{R} - \frac{1}{4} \cdot \frac{N^2 \cdot \alpha^2}{R} = \frac{3}{4} \cdot \gamma \cdot N^2$, which means the algorithm outputs $1$ with probability at least $0.97$.
    
    Therefore, we just need $N$ to satisfy $N^2 \cdot \alpha^2 \ge C(N \sqrt{d}+\alpha N \sqrt{N})$ and $\gamma \cdot N^2 \ge C \left(\frac{L \log N}{\eps} + \frac{\log^2 N}{\eps^2}\right)$. Recalling that $R = \tilde{O}(\sqrt{d})$ and that $L = \sqrt{N} + \alpha \cdot \frac{N}{\sqrt{d}}$, it is sufficient for $$N = \tilde{O}\left(\frac{\sqrt{d}}{\alpha^2} + \frac{1}{\alpha^2} + \frac{d^{1/3}}{\alpha^{4/3} \eps^{2/3}} + \frac{d^{1/4}}{\alpha \cdot \eps} + \frac{1}{\alpha \cdot \eps}\right) = \tilde{O}\left(\frac{\sqrt{d}}{\alpha^2} + \frac{d^{1/4}}{\alpha \cdot \eps}\right),$$
    where we note that $\frac{1}{\alpha^2}, \frac{1}{\alpha \cdot \eps}$ are smaller terms, and that $\frac{d^{1/3}}{\alpha^{4/3} \eps^{2/3}}$ is a weighted geometric average of $\frac{\sqrt{d}}{\alpha^2}$ and $\frac{d^{1/4}}{\alpha \cdot \eps}$ and therefore can be omitted.
\end{proof}
    
\section{Generalizations} \label{sec:generalization}

In this section, we prove theorems \ref{thm:cov_unknown}, \ref{thm:prod}, and \ref{thm:tolerant}. We will also restate these theorems more formally while proving them. We note that the theorems proceed very similarly to Theorems \ref{thm:slow_upper}, \ref{thm:lower}, and \ref{thm:fast_upper}, so we will focus on the parts of these new theorems that differ, and simply sketch the remaining parts.

\subsection{Lower Bounds}

We start with the easiest part, the lower bound. In the case of unknown covariance Gaussians, the lower bound we have in the known-covariance case immediately implies the same lower bound in the unknown-covariance case (where we are promised $\|\Sigma\|_2 \le 1$). Indeed, we can assume that $\Sigma = I$ and the same lower bound applies.

For the case of identity testing for balanced product distributions, we remark that there is a known reduction \cite[Theorem 3.1]{CanonneKMUZ20} which reduces the problem of identity testing of Gaussians with identity covariance to uniformity testing of product distributions over $\{-1, 1\}^d$. The reduction also holds for reducing to identity testing of product distributions over $\{-1, 1\}^d$, where the null hypothesis distribution $\mathcal{P}_{\mu^*}$ satisfies $|\mu^*_i| \le 1-\Omega(1)$ for all $i \in [d]$, i.e., if the null hypothesis distribution is balanced. Because this reduction shows that identity testing of Gaussians is easier than identity testing of balanced product distributions, it implies that any lower bound for identity testing of Gaussians implies the same lower bound for identity testing of balanced product distributions.

Finally, we note that tolerant testing only expands the possible null hypotheses, so any lower bounds for standard identity testing also imply the same lower bounds for the tolerant version.

Hence, we have the following theorem, which encompasses all lower bounds.

\begin{theorem}
    Any ($0, \eps$)-private algorithm that can distinguish between $\mathcal{H}_0$ and $\mathcal{H}_1$ requires at least $\Omega\left(\frac{d^{1/2}}{\alpha^2}+\frac{d^{1/3}}{\alpha^{4/3} \cdot \eps^{2/3}} + \frac{1}{\alpha \cdot \eps}\right)$ samples, for the following choices of $\mathcal{H}_0$ and $\mathcal{H}_1$:
    \begin{enumerate}
        \item \textbf{(Theorem \ref{thm:cov_unknown}, lower bound)} $\mathcal{H}_0$ consists of $\mathcal{N}(\mu^*, \Sigma)$ over all covariance matrices with bounded spectral norm $\|\Sigma\|_2 \le 1$, and $\mathcal{H}_1$ consists of $\mathcal{N}(\mu, \Sigma)$ over all covariance matrices $\|\Sigma\|_2 \le 1$ and $\mu: \|\mu-\mu^*\| \ge \alpha$, where $\mu^* \in \BR^d$ is fixed.
        \item \textbf{(Theorem \ref{thm:prod}, lower bound)} $\mathcal{H}_0$ consists of the product distribution $\mathcal{P}^*$ over $\{-1, 1\}^d$ with mean $\mu^*$, and $\mathcal{H}_1$ consists of all product distributions $\mathcal{P}$ over $\{-1, 1\}^d$ such that $\TV(\mathcal{P}, \mathcal{P}^*) \ge \alpha$, where $\mu^* \in [-1/2, 1/2]^d$ is fixed.
        \item \textbf{(Theorem \ref{thm:tolerant}, lower bound for Gaussians)} $\mathcal{H}_0$ consists of the Gaussians $\mathcal{N}(\mu, \Sigma)$ over all $\mu$ such that the Mahalanobis distance $\sqrt{(\mu-\mu^*)^T \Sigma^{-1} (\mu-\mu^*)} \le \frac{\alpha}{2}$ and $\mathcal{H}_1$ consists of $\mathcal{N}(\mu, \Sigma)$ over all $\mu$ such that $\sqrt{(\mu-\mu^*)^T \Sigma^{-1} (\mu-\mu^*)} \ge \alpha$, where $\mu^* \in \BR^d$ and $\Sigma \in \BR^{d \times d}$ are fixed (and $\Sigma$ is positive definite).
        \item \textbf{(Theorem \ref{thm:tolerant}, lower bound for Products)} $\mathcal{H}_0$ consists of all Boolean product distributions $\mathcal{P}(\mu)$ such that $\TV(\mathcal{P}(\mu), \mathcal{P}(\mu^*)) \le \frac{\alpha}{C}$ and $\mathcal{H}_1$ consists of all $\mathcal{P}(\mu)$ such that  $\TV(\mathcal{P}(\mu), \mathcal{P}(\mu^*)) \ge \alpha$, where $\mu^* \in [-1/2, 1/2]^d$ is fixed and $C$ is a sufficiently large constant.
    \end{enumerate}
\end{theorem}

\subsection{Upper Bounds: Unknown Covariance Case}

We next consider scenario of identity mean testing where the samples come from a Gaussian with unknown but bounded covariance. In other words, we want to $(\eps, 0)$-privately distinguish between $\mathcal{H}_0$, where we are given $N$ samples from $\mathcal{N}(\mu^*, \Sigma)$ for some $\|\Sigma\|_2 \le 1$, and $\mathcal{H}_1$, where we are given $N$ samples from $\mathcal{N}(\mu, \Sigma)$ for some $\|\Sigma\|_2 \le 1$ and $\|\mu-\mu^*\|_2 \ge \alpha$.

Specifically, we show the following theorem.
\begin{theorem}
    Let $\mu^* \in \BR^d$ be fixed, and let $\mathcal{H}_0$ consist of $\mathcal{N}(\mu^*, \Sigma)$ over all covariance matrices with bounded spectral norm $\|\Sigma\|_2 \le 1$, and $\mathcal{H}_1$ consist of $\mathcal{N}(\mu, \Sigma)$ over all covariance matrices $\|\Sigma\|_2 \le 1$ and $\mu: \|\mu-\mu^*\| \ge \alpha$. Then:
\begin{enumerate}
    \item \textbf{(Theorem \ref{thm:cov_unknown}, computationally inefficient case)} There exists a computationally inefficient algorithm that, using
\[N = \tilde{O}\left(\frac{d^{1/2}}{\alpha^2} + \frac{d^{1/3}}{\alpha^{4/3} \cdot \eps^{2/3}} + \frac{1}{\alpha \cdot \eps}\right)\]
    samples, can $(\eps, 0)$-privately distinguish between $\mathcal{H}_0$ and $\mathcal{H}_1$.
    \item \textbf{(Theorem \ref{thm:cov_unknown}, computationally efficient case)} There exists a computationally efficient algorithm that, using
\[N = \tilde{O}\left(\frac{d^{1/2}}{\alpha^2} + \frac{d^{1/4}}{\alpha \cdot \eps}\right)\]
    samples, can $(\eps, 0)$-privately distinguish between $\mathcal{H}_0$ and $\mathcal{H}_1$.
\end{enumerate}
\end{theorem}

Recall (by Proposition \ref{prop:wlog_bounded} and Remark \ref{rmk:wlog_bounded}) that we may think of $\mathcal{H}_1$ as being given $N$ i.i.d. samples from $\mathcal{N}(\mu, \Sigma)$ for some $\|\Sigma\|_2 \le 1$ and $\alpha \le \|\mu-\mu^*\|_2 \le 2 \alpha$. Also, by shifting, we assume WLOG that $\mu^* = 0$, but since $\Sigma$ is unknown, we may not assume WLOG that $\Sigma = I$. Instead, we note that if we sample $Y, Z \overset{i.i.d.}{\sim} \mathcal{N}(\mu, \Sigma)$, then even if $\mu$ is unknown, $W := \frac{Y-Z}{\sqrt{2}} \sim \mathcal{N}(0, \Sigma)$. So, by using $3N$ samples, we can generate $X^{(1)}, \dots, X^{(N)} \overset{i.i.d.}{\sim} \mathcal{N}(\mu, \Sigma)$, and $W^{(1)}, \dots, W^{(N)} \overset{i.i.d.}{\sim} \mathcal{N}(0, \Sigma)$. We start by considering the matrix $\textbf{U}$ where $U_{i, j} = \langle X^{(i)}, X^{(j)} \rangle - \langle W^{(i)}, W^{(j)} \rangle$. Assuming that $\|\mu\| \le 2 \alpha$ where $\alpha \le \frac{1}{2}$, we can apply Theorem \ref{thm:concentration_2} for some appropriate choice of $L = \poly \log(N)$ to obtain that with probability at least $0.98$,
\begin{itemize}
    \item $\sum_{i = 1}^{N} \sum_{j = 1}^{N} U_{i, j} = N^2 \|\mu\|^2 \pm O(N \sqrt{d}+\alpha N \sqrt{N})$.
    \item For all $i, j \in [N]$ (including if $i = j$), $U_{i, j} = \pm L\cdot \sqrt{d}$.
    \item For all $i$, $\sum_{j = 1}^{N} U_{i,j} = \pm L(\sqrt{Nd}+\alpha \cdot N)$, and for all $j$, $\sum_{i = 1}^{N} U_{i,j} = \pm L(\sqrt{Nd}+\alpha \cdot N)$.
    \item For all subsets $S \subset [N]$, if $|S| = K$, then $\sum_{i, j\in S} U_{i, j} = K \cdot L \cdot (\sqrt{Kd}+K)$.
\end{itemize}

Then, if we let $\textbf{T}$ be the matrix where add $d$ to every diagonal entry in $\textbf{U}$, we have that with probability at least $0.98$,
\begin{itemize}
    \item $\sum_{i = 1}^{N} \sum_{j = 1}^{N} T_{i, j} = Nd + N^2 \|\mu\|^2 \pm O(N \sqrt{d}+\alpha N \sqrt{N})$.
    \item For all $i \neq j$, $T_{i, j} = \pm L\cdot \sqrt{d}$, and for all $i$, $V_{i, i} = d \pm L \cdot \sqrt{d}$.
    \item For all $i$, $\sum_{j = 1}^{N} T_{i,j} = d \pm L(\sqrt{Nd}+\alpha \cdot N)$, and for all $j$, $\sum_{i = 1}^{N} T_{i,j} = d \pm L(\sqrt{Nd}+\alpha \cdot N)$.
    \item For all subsets $S \subset [N]$, if $|S| = K$, then $\sum_{i, j\in S} T_{i, j} = K \cdot \left[d \pm L \cdot (\sqrt{Kd}+K)\right]$.
\end{itemize}
These are precisely the bounds we used for the matrix to distinguish between $\mu = 0$ and $\alpha \le \|\mu\|_2 \le 2 \alpha$ in Section \ref{sec:slow}. Hence, we can apply the same algorithm using this matrix $\textbf{T}$ instead. Indeed, we define $\mathcal{C}$ to be the set of datasets $(\textbf{X}, \textbf{Y}, \textbf{Z}) \in (\BR^d)^{3N}$ that satisfy the second, third, and fourth properties above. Then, the same proof as in Lemmas \ref{lem:lipschitz_1} and \ref{lem:lipschitz_2} imply that for any datasets $(\textbf{X}, \textbf{Y}, \textbf{Z}) \in \mathcal{C}$ and $(\textbf{X}', \textbf{Y}', \textbf{Z}') \in \mathcal{C}$ that differ in $K$ rows, $|\tilde{T}(\textbf{X}, \textbf{Y}, \textbf{Z})-\tilde{T}(\textbf{X}', \textbf{Y}', \textbf{Z}')| \le K \cdot D$, assuming that $\Delta := 6L\left(\sqrt{Nd}+\alpha N + \frac{C}{\eps}\right) \le \frac{\eps}{C} \cdot \alpha^2 \cdot N^2 =: D$. Here, $\tilde{T}(\textbf{X}, \textbf{Y}, \textbf{Z})$ is the statistic that adds up the entries of the corresponding matrix $\textbf{T} \in \BR^{N \times N}$, subtracts $Nd$, and clips the result by $0$ from below and $\alpha^2 \cdot N^2$ from above.

Therefore, we can create our $D$-Lipschitz extension $\hat{T}$ of $\tilde{T}$ that sends all datasets $(\textbf{X}, \textbf{Y}, \textbf{Z}) \in (\BR^d)^{3N}$ to $\BR$. Finally, because we also have that the sum of all entries in $\textbf{T}$ equals $Nd + N^2 \|\mu\|^2 \pm O(N \sqrt{d}+\alpha N \sqrt{N})$ with probability at least $0.98$ if each $X^{(i)}, Y^{(i)}, Z^{(i)}$ is drawn i.i.d. from $\mathcal{N}(\mu, I)$, we can apply the same algorithm as in Theorem \ref{thm:slow_upper} to obtain a private algorithm that succeeds with high probability using the same number of samples, up to a factor of $3$. This completes the proof in the computationally inefficient case.

\medskip

For the computationally efficient case, we let $\textbf{V}$ be the matrix where we divide each entry of $\textbf{U}$ by $R$, where $R = \tilde{O}(\sqrt{d})$ is a sufficiently large polylogarithmic multiple of $\sqrt{d}.$ Then, with probability at least $0.98$,
\begin{itemize}
    \item $\sum_{i = 1}^{N} \sum_{j = 1}^{N} V_{i, j} = \frac{N^2 \|\mu\|^2}{R} \pm O(N \sqrt{d}+\alpha N \sqrt{N})$.
    \item For all $i, j \in [N]$, $V_{i, j} \in [-1, 1]$.
    \item For all $i$, $\sum_{j = 1}^{N} V_{i,j} = \pm (\sqrt{N}+\alpha \cdot \frac{N}{\sqrt{d}})$, and for all $j$, $\sum_{i = 1}^{N} V_{i,j} = \pm L(\sqrt{N}+\alpha \cdot \frac{N}{\sqrt{d}})$.
\end{itemize}
These are the same guarantees as in Corollary \ref{cor:concentration_fast_upper}.
In addition, we again have that if we change exactly one entry among $(\textbf{X}, \textbf{Y}, \textbf{Z})$, this changes at most one row and one column. Because of this, we can apply Theorem \ref{thm:matrix} in the same way that we prove Theorem \ref{thm:fast_upper} to complete the proof in the computationally efficient case.

\subsection{Upper Bounds: Product Distributions}

We establish the following theorem.
\begin{theorem} \label{thm:prod_upper}
    Let $\mu^* \in [-1/2, 1/2]^d$ be fixed, and let $\mathcal{H}_0$ consist of $\mathcal{P}(\mu^*)$, and $\mathcal{H}_1$ consist of $\mathcal{P}(\mu)$ over all $\mu$ with $\TV(\mathcal{P}(\mu), \mathcal{P}(\mu^*)) \ge \alpha$. Then:
\begin{enumerate}
    \item \textbf{(Theorem \ref{thm:prod}, computationally inefficient case)} There exists a computationally inefficient algorithm that, using
\[N = \tilde{O}\left(\frac{d^{1/2}}{\alpha^2} + \frac{d^{1/3}}{\alpha^{4/3} \cdot \eps^{2/3}} + \frac{1}{\alpha \cdot \eps}\right)\]
    samples, can $(\eps, 0)$-privately distinguish between $\mathcal{H}_0$ and $\mathcal{H}_1$.
    \item \textbf{(Theorem \ref{thm:prod}, computationally efficient case)} There exists a computationally efficient algorithm that, using
\[N = \tilde{O}\left(\frac{d^{1/2}}{\alpha^2} + \frac{d^{1/4}}{\alpha \cdot \eps}\right)\]
    samples, can $(\eps, 0)$-privately distinguish between $\mathcal{H}_0$ and $\mathcal{H}_1$.
\end{enumerate}
\end{theorem}

The product distribution case will be remarkably similar to the identity-covariance Gaussian case. First, by Lemma \ref{lem:wlog_product_mean_0}, we may assume WLOG that $\mu^* = 0$. In addition, we can use Proposition \ref{prop:mean_tv} to state that for any distribution $\mu$ with $\TV(\mathcal{P}(\mu), \mathcal{P}(\mu^*)) \ge \alpha$, then $\|\mu\|_2 \ge \frac{\alpha}{C_0}$ for some fixed absolute constant $C_0$. So, it suffices to distinguish between the null hypothesis of $\mu = \mu^* = 0$ and $\|\mu\| \ge \frac{\alpha}{C_0}$. We can replace $\frac{\alpha}{C_0}$ with $\alpha$, up to a constant factor loss in our sample complexity. 
We can again apply Remark \ref{rmk:wlog_bounded} to say that $\mathcal{H}_1$ is being given $N$ i.i.d. samples from $\mathcal{P}(\mu)$, where $\alpha \le \|\mu\| \le 2 \alpha$.

Now, we note that Theorem \ref{thm:concentration} works for both identity-covariance Gaussians and product distributions. Because of this, we obtain the same concentration bounds as in the known identity covariance case. Hence, for the computationally inefficient case, Lemmas \ref{lem:lipschitz_1} and \ref{lem:lipschitz_2} can be trivially adapted to the product case, as can Theorem \ref{thm:slow_upper}. Likewise, in the computationally efficient case, we can generate the matrix $\textbf{V}$ and apply Theorem \ref{thm:matrix} in the same way as we did in the identity covariance case.

\subsection{Upper Bounds: Tolerant Testing}

We start with the known-covariance Gaussian case. Specifically, we establish the following theorem.
\begin{theorem} \label{thm:tolerant_upper}
    Let $\mu^* \in \BR^d$ and $\Sigma \in \BR^{d \times d}$ be fixed (with $\Sigma$ positive definite). Let $\mathcal{H}_0$ consist of $\mathcal{N}(\mu, \Sigma)$ over all $\mu$ with $\sqrt{(\mu-\mu^*)^T \Sigma^{-1} (\mu-\mu^*)} \le \frac{\alpha}{2}$, and $\mathcal{H}_1$ consist of $\mathcal{N}(\mu, \Sigma)$ over all $\mu$ with $\sqrt{(\mu-\mu^*)^T \Sigma^{-1} (\mu-\mu^*)} \ge \alpha$. Then:
\begin{enumerate}
    \item \textbf{(Theorem \ref{thm:tolerant}, computationally inefficient Gaussian case)} There exists a computationally inefficient algorithm that, using
\[N = \tilde{O}\left(\frac{d^{1/2}}{\alpha^2} + \frac{d^{1/3}}{\alpha^{4/3} \cdot \eps^{2/3}} + \frac{1}{\alpha \cdot \eps}\right)\]
    samples, can $(\eps, 0)$-privately distinguish between $\mathcal{H}_0$ and $\mathcal{H}_1$.
    \item \textbf{(Theorem \ref{thm:tolerant}, computationally efficient Gaussian case)} There exists a computationally efficient algorithm that, using
\[N = \tilde{O}\left(\frac{d^{1/2}}{\alpha^2} + \frac{d^{1/4}}{\alpha \cdot \eps}\right)\]
    samples, can $(\eps, 0)$-privately distinguish between $\mathcal{H}_0$ and $\mathcal{H}_1$.
\end{enumerate}
\end{theorem}

We may assume WLOG that $\mu^* = 0$ and $\Sigma = I$. First, we recall Lemma \ref{bound:total_gaussian}, which tells us that if $X^{(1)}, \dots, X^{(N)} \overset{i.i.d.}{\sim} \mathcal{N}(\mu, I)$ and $\bar{X} = X^{(1)}+\cdots+X^{(N)}$, then if $\|\mu\| \le \frac{\alpha}{2}$ then $\BE[\|\bar{X}\|^2] \le Nd + N^2 \cdot \frac{\alpha^2}{4}$ and if $\|\mu\| \ge \alpha$ then $\BE[\|\bar{X}\|^2] \ge Nd + N^2 \cdot \alpha^2$. In addition, assuming $\|\mu\| \le 2 \alpha$ and $\alpha \le 1$, then $Var[\|\bar{X}\|^2] = O(N^2 d + N^3 \alpha^2).$

Indeed, we will see that the same algorithms as in Theorem \ref{thm:slow_upper} and Theorem \ref{thm:fast_upper} will exactly apply in the tolerant case. Of course, the algorithms will still be $(\eps, 0)$-differentially private and will still output $1$ with probability at least $2/3$ if the data is drawn i.i.d. from $\mathcal{N}(\mu, I)$ with $\alpha \le \|\mu\| \le 2 \alpha$. So, it suffices to show that the algorithm outputs $0$ with probability at least $2/3$ even if the data is drawn i.i.d. from $\mathcal{N}(\mu, I)$ with $\|\mu\| \le \frac{\alpha}{2}$.

To see why this holds in the computationally inefficient case, we recall that if $X^{(1)}, \dots, X^{(N)} \overset{i.i.d.}{\sim} \mathcal{N}(\mu, I)$ for some $\|\mu\| \le \frac{\alpha}{2}$, then with probability at least $0.99$, $\textbf{X} = (X^{(1)}, \dots, X^{(N)}) \in \mathcal{C}$. This implies that with probability at least 0.98, $\hat{T}(\textbf{X}) \le N^2 \cdot \|\mu\|^2 + C(N \sqrt{d}+\alpha N \sqrt{N})$ for some sufficiently large constant $C$ by Chebyshev's inequality. So, with probability at least $0.9$, we have our final private estimate $\hat{T}(\textbf{X}) + Lap(D/\eps) \le N^2 \cdot \frac{\alpha^2}{4} + C(N \sqrt{d}+\alpha N \sqrt{N}) + 10 \cdot \frac{D}{\eps}$. Recalling that $D = \frac{\eps}{C} \cdot \alpha^2 N^2$, this is less than $N^2 \cdot \frac{\alpha^2}{2}$ as long as 
\[N \ge \Omega\left(\frac{\sqrt{d}}{\alpha^2}\right).\]
But recall that we also needed $\Delta = \tilde{O}(\sqrt{Nd}+\alpha N + \frac{C}{\eps}) \le D$, and this holds as long as $$N \ge \tilde{O}\left(\frac{d^{1/3}}{\alpha^{4/3} \cdot \eps^{2/3}} + \frac{1}{\alpha \cdot \eps}\right).$$
Hence, we still obtain the same number of samples required for our inefficient upper bound.

In the computationally efficient case, we note that with probability at least $0.99$, the matrix $\textbf{V}$ generated from $X^{(1)}, \dots, X^{(N)} \overset{i.i.d.}{\sim} \mathcal{N}(\mu, I)$ has the sum of all entries equal to $\frac{N^2 \cdot \|\mu\|^2}{R} \pm O\left(\frac{N \sqrt{d}+\alpha N \sqrt{N}}{R}\right)$ by Corollary \ref{cor:concentration_fast_upper}, where $R = \tilde{O}(\sqrt{d})$. Recalling that $\|\mu\| \le \frac{\alpha}{2}$ and $\gamma = \frac{\alpha^2}{R},$ this is at most $\frac{\gamma \cdot N^2}{4} + O\left(\frac{N \sqrt{d}+\alpha N \sqrt{N}}{R}\right) \le \frac{1.01 \cdot \gamma \cdot N^2}{4},$ assuming that $N \ge C\left(\frac{\sqrt{d}}{\alpha^2}\right)$ for some sufficiently large constant $C$. Hence, we can apply Theorem \ref{thm:matrix} (with a slight modification to replace $\frac{\gamma \cdot N^2}{4}$ with $\frac{1.01 \cdot \gamma \cdot N^2}{4}$, for which the proof is nearly unaffected) to get that as long as $\gamma \cdot N^2 \ge \tilde{O}\left(\frac{L}{\eps}+\frac{1}{\eps^2}\right)$, the algorithm works, where $\gamma = \frac{\alpha^2}{\tilde{O}(\sqrt{d})}$ and $L = \sqrt{N}+\alpha \cdot \frac{N}{\sqrt{d}}$. As in the proof of Theorem \ref{thm:fast_upper}, it suffices for $N \ge \tilde{\Omega}\left(\frac{\sqrt{d}}{\alpha^2}+\frac{d^{1/4}}{\alpha \cdot \eps}\right)$.

\medskip

Next, we move to the Boolean Product case. Here, we establish the following theorem.
\begin{theorem} \label{thm:prod_tolerant_upper}
    Let $\mu^* \in [-1/2, 1/2]^d$ and $C$ be a sufficiently large constant. Let $\mathcal{H}_0$ consist of $\mathcal{P}(\mu)$ over all $\mu$ with $\TV(\mathcal{P}(\mu), \mathcal{P}(\mu^*)) \le \frac{\alpha}{C}$, and $\mathcal{H}_1$ consist of $\mathcal{P}(\mu)$ over all $\mu$ with $\TV(\mathcal{P}(\mu), \mathcal{P}(\mu^*)) \ge \alpha$. Then:
\begin{enumerate}
    \item \textbf{(Theorem \ref{thm:tolerant}, computationally inefficient Product case)} There exists a computationally inefficient algorithm that, using
\[N = \tilde{O}\left(\frac{d^{1/2}}{\alpha^2} + \frac{d^{1/3}}{\alpha^{4/3} \cdot \eps^{2/3}} + \frac{1}{\alpha \cdot \eps}\right)\]
    samples, can $(\eps, 0)$-privately distinguish between $\mathcal{H}_0$ and $\mathcal{H}_1$.
    \item \textbf{(Theorem \ref{thm:tolerant}, computationally efficient Product case)} There exists a computationally efficient algorithm that, using
\[N = \tilde{O}\left(\frac{d^{1/2}}{\alpha^2} + \frac{d^{1/4}}{\alpha \cdot \eps}\right)\]
    samples, can $(\eps, 0)$-privately distinguish between $\mathcal{H}_0$ and $\mathcal{H}_1$.
\end{enumerate}
\end{theorem}

Again, we may assume WLOG that $\mu^* = 0$, by Lemma \ref{lem:wlog_product_mean_0}. In addition, by Proposition \ref{prop:mean_tv}, we have that $\TV(\mathcal{P}(0), \mathcal{P}(\mu)) \in \left[\frac{\min(\|\mu\|, 1)}{C_0}, C_0 \cdot \min(\|\mu\|, 1)\right]$ for all $\mu \in [-1, 1]^d$, for some fixed absolute constant $C_0$. So, if we set $C = 2C_0^2$, it suffices to distinguish between $\|\mu\| \le \frac{\alpha}{C} \cdot C_0 = \frac{\alpha}{2C_0}$ and $\|\mu\| \ge \frac{\alpha}{C_0}.$ By replacing $\alpha$ with $\frac{\alpha}{C_0}$, up to an asymptotic factor it suffices to distinguish between $\|\mu\|_2 \le \frac{\alpha}{2}$ and $\|\mu\|_2 \ge \alpha$ for all $\alpha \le \frac{1}{2}.$ Again, we may assume WLOG (by Proposition \ref{prop:wlog_bounded} and Remark \ref{rmk:wlog_bounded}) that we are distinguishing between $\|\mu\|_2 \le \frac{\alpha}{2}$ and $\alpha \le \|\mu\|_2 \le 2\alpha$.

At this point, the proof of Theorem \ref{thm:prod_tolerant_upper} proceeds almost identically to the proof of Theorem \ref{thm:tolerant_upper}. Indeed, we have the same concentration bounds for $\|\bar{X}\|^2$ where $\bar{X} = X^{(1)} + \cdots + X^{(N)},$ as well as the same required bounds on the row and column sums, the individual entries, and the sums of submatrices in the matrix $\textbf{T}$ where $T_{i, j} = \langle X^{(i)}, X^{(j)} \rangle,$ assuming that $X^{(1)}, \dots, X^{(N)} \overset{i.i.d.}{\sim} \mathcal{P}(\mu)$. Hence, the same proof as for Theorem \ref{thm:tolerant_upper} will apply to Theorem \ref{thm:prod_tolerant_upper}.

\end{document}